%% file: CCQW.tex
\documentclass[a4paper,11pt]{article}
\usepackage[T1]{fontenc}
\usepackage[utf8]{inputenc}
\usepackage{graphicx}
\usepackage{tensor}
\usepackage{amsmath}
\usepackage{mathtools}
\usepackage{amsfonts}
\usepackage[dvipsnames]{xcolor}
\usepackage{ntheorem}
\usepackage{a4}
\usepackage{hyperref}
\usepackage{cite}
\usepackage{wasysym} 
\input{macros}
\input{macros2}

\usepackage[normalem]{ulem} 

\usepackage{mleftright} \mleftright  

\renewcommand{\ptcheck}[1]{}
\renewcommand{\ptcn}[1]{}
\renewcommand{\tqn}[1]{}
\renewcommand{\rwn}[1]{}
\renewcommand{\rwc}[1]{}
\renewcommand{\red}[1]{#1}

\renewcommand{\bluec}{}

\begin{document}

\title{Cauchy problems for Einstein equations in three-dimensional spacetimes}
\author{Piotr T. Chru\'{s}ciel\thanks{University of Vienna, Faculty of Physics} \thanks{
{\sc Email} \protect\url{piotr.chrusciel@univie.ac.at}, {\sc URL} \protect\url{homepage.univie.ac.at/piotr.chrusciel}}
\\Wan Cong$^*$\thanks{
{\sc Email} \protect\url{wan.cong@univie.ac.at}
}
\\
{Th\'eophile Qu\'eau}\thanks{Université Paris-Saclay, ENS Paris-Saclay, DER de Physique, Gif-sur-Yvette, France}
\thanks{
{\sc Email} \protect\url{theophile.queau@ens-paris-saclay.fr}}
\\
{Raphaela Wutte}\thanks{Department of Physics and Beyond: Center for Fundamental Concepts in Science, Arizona State University}
\thanks{
{\sc Email} \protect\url{rwutte@hep.itp.tuwien.ac.at} {} 
}
}

	\maketitle

\begin{abstract} 
We analyze existence  and   properties  of solutions of two-dimensional general relativistic initial data sets with a negative cosmological constant, both on spacelike and characteristic surfaces. A new family of such vacuum spacelike data parameterised by    poles at the conformal boundary at infinity is constructed.
 We review the notions of global Hamiltonian charges, emphasizing the difficulties arising in this dimension, both in a spacelike and characteristic setting.
One or two, depending upon the topology, lower bounds for energy in terms of  angular momentum, linear momentum, and center of mass are established.
\end{abstract}

	\tableofcontents

\input{beginning}
\input{ConfTransf2}
\input{Airy}

	\input{VecCstr}

\input{Poles}

\input{ConformalLaplacian}

	%
	\input{FiniteMass1}
	%

\input{Lichnerowicz}
\input{BdryRegularity1}


\ptcn{SclCstr commented out, Function Spaces and following moved to an Appendix}

\input{UhlenbeckLorentzian}

\input{ang-momentum-minimal}

\input{SmallData}

\input{Witten}

\ptcn{previous version in Witten3} 
\input{Witten3a}

\input{AddMass}

 
 	\input{Gluings}

 \input{characteristic2d}
 
 \input{summary}

\input{BMS}

\appendix

	\input{HPHomographie}

\input{Lichnerowicz3}

\input{PHdvp}

	
	\input{PHdvp-pole}

\input{maximal}

\input{Appendix}



\input{spinortotalnew} 
%

%
%
%
	\bibliographystyle{amsplain}
	\bibliography{CCQW-minimal}

\end{document}

%% file: macros2.tex
\newcommand{\bcheckpsi}{\blue{\check \psi}}

\newcommand{\fancyD}{\red{D}}

\newcommand{\sqrtbg}{{d\mu_{\backg}}}%
\newcommand{\rsqrtbg}{}%

\newcommand{\gauge}[1]{ {#1}}

\newcommand{\hdherr}{O\left( |h|^2_{\backg}\, |\bcov h|_{\backg}\right)}

\newcommand{\herr}{O\left( |h|^3_{\backg}\right)}

\definecolor{orange}{RGB}{255,127,0}

\newcommand{\back}[1]{\red{\overline{#1}}} %
\newcommand{\backg}{{\overline{g}}} %
\newcommand{\bcov}{\back{D}} %

\newcommand{\btheta}{\back{\theta}} %

\newcommand{\scri}{\blue{\mycal I}}

\newcommand{\hhat}{\blue{\hat h}}

\newcommand{\cut}{\blue{\mathcal S}}

\newcommand{\bmu}{\back{\mu}}

\newcommand{\hbound}{\red{h}}
\newcommand{\fone}{\red{\phi}}

\newcommand{\mcP}{{\bluec{\mycal{P}}}}

\newcommand{\Hh}{\red{H^0}}

\newcommand{\htheta}{\red{\hat\theta}}
\newcommand{\hmu}{\red{\hat \mu}}

\newcommand{\wasxone}{\red{x^n}}

\newcommand{\uHamM}{\underline{\HamM}}

\newcommand{\mcL}{\blue{\mycal L}}

\newcommand{\bfour}{\red{\mathring{\fourg}}} 
\newcommand{\gfour}{\red{\mathbf{g}}} 

\newcommand{\maxgraph}{\red{\red{f}}} 
\newcommand{\bmaxgraph}{\red{\mybar{\maxgraph}}} 

\newcommand{\mybar}{\breve}

\newcommand{\mybarmcL}{\blue{\mybar{\mathcal{L}}}}

\newcommand{\mycdot}{\red{\cdot}}
\newcommand{\myy}{\red{y}}

\newcommand{\mcN}{\blue{\mycal{N}}}

\newcommand{\HamM}{\red{H}}

\newcommand{\zP}{{\mathring P}}

\newcommand{\myf}{\red{v}} 
\newcommand{\hSchw}{\red{\hat S}}

\newcommand{\myj}{\red{j}}
\newcommand{\ourU}{\mathbb U}
\newcommand{\zD}{\mathring{D}}%
\newcommand{\trg}{{\mathrm{tr}_g}}

\newcommand{\rwn}[1]{\mnote{ {\bf Raphaela:} #1}}

\newcommand{\rwc}[1]{\mnote{ {\bf RW check:} #1}}

\newcommand{\bv}{\blue{ v }}

\newcommand{\mcE}{\mycal E}
\newcommand{\mcH}{\mycal H}

\newcommand{\mTB}{\blue{m_{\text{TB}}}}

\newcommand{\mU}{\blue{\breve{U}}}
\newcommand{\mmU}{\blue{\mathring{U}}}

\newcommand{\hodge}{\mathop{}\!\ast}

\newcommand{\myomega}{\red{\Omega}}

\newcommand{\myr}{\red{\bar r}}

\newcommand{\tila}{\red{\tilde a}}
\newcommand{\tilb}{\red{\tilde b}}

%% file: beginning.tex
\section{Introduction}

Most of the literature on the general relativistic constraint equations assumes from the outset that the space-dimension $n$ is larger than or equal to three (see, e.g., \cite{BartnikIsenberg,CarlottoLRR} and references therein). This is presumably based on the preconception, that all three-dimensional vacuum spacetimes are locally isometric, so that no new insights into the theory can be gained when $n=2$. Another reason might be,  that the conformal transformation properties in space-dimension two do not follow the same pattern as these with $n\ge 3$. Since the conformal method remains the best tool to construct general classes of spacetimes, $n=2$ does not fit into the flow of the arguments.

Now, while three-dimensional vacuum spacetimes are all indeed locally isometric, spacetimes with matter are not. Moreover, the global structure matters even in vacuum, and deserves to be understood. Finally, the total mass in two space-dimensions has properties completely different from these of the higher-dimensional models~\cite{ChWutte}, with a behavior which is perplexing enough to deserve further investigation.

For all these reasons we have undertaken to analyze the usual  constructions of  general relativistic initial data sets with a negative cosmological constant in space-dimension two. The results are presented here.
It must be admitted that the bottom line is somewhat disappointing: those known methods which we  looked at work without further due when $n=2$.  
But, while this work has mainly a review character, dotting  previously-undotted i's at some places, we also present here a seemingly new result, namely the existence of a family of solutions of two-dimensional vacuum constraint equations,   and hence of three-dimensional spacetimes, with non-conformally-smooth behavior at infinity, with controlled asymptotics, and with finite total energy. One hence obtains an interesting extension of the usual phase space of $(2+1)$-dimensional gravitational initial data sets.  
Another new result in this work is the proof that a unique definition of mass can be obtained by minimization for asymptotically locally hyperbolic initial data on complete manifolds with matter fields satisfying the dominant energy condition. This is done by proving a lower bound for the mass using the Witten positivity argument. While this result is essentially already contained in \cite{SkenderisCheng,ChHerzlich,CMT}, the ambiguities related both to the definition of mass and to the existence of inequivalent spin structures do not seem to have  been   
neither addressed (ambiguities) nor exploited (inequivalent spin structures) 
in previous works.  
 
This paper is organized as follows: In Section \ref{s11VII24.1} we review the   general form of metrics  with a smooth conformal compactification at infinity, either $(2+1)$-dimensional Lorentzian and vacuum near the conformal boundary,  or $2$-dimensional Riemannian with constant scalar curvature near the conformal boundary.
In Section~\ref{s24VII24.1} we pass to an analysis of the constraint equations.  
 After some preliminary results in Sections~\ref{s22VIII24.1} and~\ref{ss11VIII24.1}, in Section~\ref{ss24VII24.1} we point out that Airy functions can be used to obtain symmetric two-tensors with vanishing divergence. In Section~\ref{s29VII24.1} we recall that constant mean curvature (CMC) solutions of the vacuum vector constraint equation can be parameterized by holomorphic functions. In Section~\ref{ss11VIII24.8} we point out that meromorphic functions with  poles on the conformal boundary lead to a  class of  solutions, with finite energy, which does not seem to have been noticed so far.
  In Section~\ref{ss11VIII24.2} we analyze the regularity at the conformal boundary of general relativistic initial data constructed by the conformal method. 
 In Section~\ref{s2X24.1} we point out a formula, essentially due to Uhlenbeck~\cite{UhlenbeckHyperbolic},
  for the explicit 
  form of solutions of the vacuum Cauchy problem in a neighborhood of the initial data surface. 
  We expect that many of these metrics provide explicit examples of dynamical black holes. 
In Section~\ref{s14VIII24.1} we review the notion of angular momentum and mass
for asymptotically locally hyperbolic initial data sets. 
We point out ambiguities in the definition of total mass  resulting from asymptotic symmetries, and show that these can be resolved for complete initial data sets by minimization. 
This is done by proving  lower bounds for the total mass for  a class of small perturbations of Anti-de Sitter spacetime in Section~\ref{30VII24},  
and for general asymptotically locally hyperbolic  data on complete manifolds with matter fields satisfying the dominant energy condition in Section~\ref{ss19IX24}.   
 In Section~\ref{s11VIII24.11} we discuss some specific matter models: in Section~\ref{ss11VIII24.12} Einstein-Maxwell initial data are considered, while  scalar fields are discussed in Section~\ref{11VIII24.21}. Existence of associated initial data is addressed in Section~\ref{11VII24.4}.
 In Section~\ref{s11VIII24.14} we point out that the known gluing results for spacelike initial data sets also apply to space-dimension $n=2$.
 In Section~\ref{s18VIII24.1} we review the characteristic Cauchy problem in Bondi coordinates, 
 point-out that the vacuum characteristic-gluing is all-radial-conservation laws,
 and note that the  difficulties with the definition of total mass, identical to those already encountered in the spacelike case, arise both for the Trautman-Bondi mass and for angular momentum at null infinity. 
  \ptcn{ add characteristic gluing for a scalar field; and Maxwell?}
  In Appendix~\ref{App11VIII24.2} a formula for transformation of poles between the Poincar\'e-disc model of hyperbolic space and the half-space model is worked-out. In Appendix~\ref{s26VII24.2} the problem of regularity at the conformal boundary at infinity of solutions of the Lichnerowicz equation is reviewed. In Appendix~\ref{s11VII24.1b} we establish the leading-order behavior of the solutions with poles at the conformal boundary and we provide a full formal asymptotic expansion of these solutions.  
  In Appendix~\ref{app10IX24.1} existence results for, and asymptotic behavior of,  maximal surfaces in  {3-dimensional} asymptotically locally AdS spacetimes are reviewed. 
  In Appendix~\ref{App2VIII24.1} an ODE result relevant for the asymptotics of the initial data sets with  poles-at-the-boundary  is established. 
  Explicit formulae for imaginary Killing spinors in various coordinate systems on ALH manifolds are given in Appendix~\ref{s24X24.1}.
    
We will mostly
 assume that $\Lambda =-1$  without explicitly saying so. 
While we are mainly interested in a negative cosmological constant here, in some rare cases,
 namely in  Sections~\ref{s2X24.1}  and \ref{s18VIII24.1}, a positive cosmological constant is also  considered.

\bigskip
\input{Acknowledgements}
 
\input{Banados}

\input{Vacuum}

\section{Spacelike general relativistic initial data}
 \label{s24VII24.1}

Some definitions are in order. 
	As already mentioned, a spacelike general relativistic initial data set   is a triple $(M,\pg,\pK)$, where $M$ a smooth $n$-dimensional manifold, $\pg$ a Riemannian metric
	on $M$, and $\pK$ a symmetric tensor field on $M$. These fields are moreover required to satisfy a set of constraint equations: 
	\bea\label{CstrVect}
	& D_j 
 \big(
  \iju \pK - \tr_\pg(\pK) \iju \pg
  \big) = J^i  \,,
	&
	\\
	\label{CstrScl}
	& R_{\pg} - \normsq \pK \pg + \big(
 \tr_{\pg}(\pK)
  \big)^{2} = 2\matter + 2\Lambda \,,
	&
	\eea
	with $(\matter,J^i)$ the mass density and current,
which vanish in the vacuum case. $D$ is the Levi-Civita covariant derivative of the metric $\pg$, and $R_\pg$ the Ricci scalar of $\pg$. Equation \eqref{CstrVect} is known as the vector constraint equation, while  \eqref{CstrScl} is the scalar constraint equation. The coefficient $2$ in front of $\matter$ is sometimes replaced by $1$ (or by $16 \pi$), in which case a coefficient $1/2$ (or $8\pi$) has to be inserted in front of $J^i$. A negative sign, which can often be found in front of $J^i$, can be achieved by changing $K$ to its negative. See~\cite{BartnikIsenberg,CarlottoLRR} for an exhaustive discussion, where however space-dimension two is mostly ignored.

It may be convenient, e.g.\ for gluing purposes, to have compactly supported tensors $K^{ij}$ at our disposal, or metrics which are compactly supported deformations of fiducial ones. We will therefore make remarks on this as we go along.

	A data set $(M,g,K)$ is said to be \emph{asymptotically locally hyperbolic} (short ALH), or is said  to have a \emph{smooth conformal completion at infinity},  if there exists a smooth defining function $\defining \in\Cinf {\ov M}$ such as $\defining ^2 \pg$ extends to a smooth  metric on $\ov M$ 
with
	\begin{equation*}
		\begin{cases}
			\defining  > 0 \textrm{ on } M \,,\\
			\defining |_{\partial M}=0 \,,\\
			|d\defining |_{\textit{\pg}}=1 \textrm{ on } \partial M\,.
		\end{cases}
	\end{equation*}
	Note that $\partial M$ represents the conformal boundary at infinity, and the sectional curvatures of $\pg$ approach $-1$ when $\partial M$ is approached.  
We will further require  that the tensor field $  K_{ij} dx^i dx^j$ extends smoothly across $\partial M$; this is the case for metrics of the Ba\~nados form \eqref{metrics}. 

\subsection{Vacuum time symmetric data}
 \label{s22VIII24.1}

The induced metric, together with the extrinsic curvature of a hypersurface in a vacuum spacetime, provide general relativistic
vacuum  initial data. 
Hence,  large families of conformally smooth examples can be obtained from data induced on hypersurfaces by the metrics \eqref{metrics}. In particular
one thus obtains the following family of constant-scalar-curvature metrics 
\begin{equation}
  g = r^{-2}  dr^2 + \big(
   r^2 + f_2(\varphi)   +  \frac{f_2(\varphi)^2}{4 r^2}
   \big)
    d\varphi^2 
     \label{26VII24.1ty}
  \,,
 \end{equation}
 where $f_2$ is an arbitrary function of $\varphi$. This, together with  $K_{ij}\equiv 0$, provides 
 a large family of conformally smooth initial data which are vacuum and time symmetric near the conformal boundary at infinity. 
  
Given a point $\varphi_0$ on $S^1$, the   function $f_2$ can always  be transformed, 
 by a coordinate transformation preserving this form of the metric,
 to \emph{any} desired function
 \emph{near  $\varphi_0$}, e.g. $f_2 =0$. The question, if and when this can be done globally,  is non-trivial; see~\cite{ChWutte} and references therein. More on this in Section~\ref{subsec:asym} below.
 
 One can now use compactly supported  deformations of  the hypersurface $\{t=0\}$ in the associated spacetime to obtain vacuum initial data with, if desired, a compactly supported extrinsic curvature tensor.

\subsection{The vector constraint equation and conformal transformations}
 \label{ss11VIII24.1}

We will use the conformal method to study properties of
asymptotically hyperbolic (AH)
solutions 
 of the constraint equations \eqref{CstrVect}-\eqref{CstrScl}. This is particularly natural in space-dimension $n=2$ in view of the uniformization theorem (cf., e.g., \cite{TroyanovHulin,MazzeoTaylor,Borthwick}). 
Hence, we will look for $g$ of the form
	\bel{Scalingu}
		\ijd\pg = e^{-2u} \ijd\tg \,.
	\ee
	It will sometimes be convenient to work with the Euclidean
metric
$$
 \qg \equiv \delta
 \,,
$$ 
therefore we also  introduce the (possibly local) rescaling
	\bel{Scalingphi}
		\ijd\pg = e^{-2\qp} \ijd\qg 
\,,
 \qquad
 \qg_{ij} = \delta _{ij}\,.
	\ee
Unless explicitly indicated otherwise, we will assume that both $u$ and  $\qp$ are smooth.

In the conformal method one typically assumes
\be
D_i (\tr_{\pg} \pK) = 0 \,,
 \label{13VII24.1}
\ee
which will be sometimes done in what follows. 
Under \eqref{13VII24.1}
the analysis of the vacuum vector constraint equation, in space-dimension $n$,
 reduces to the  study of the trace-free part of the extrinsic curvature tensor $\pK^{ij}$,
\begin{equation}\label{TraceFreeL}
	\iju \pL = \iju \pK - \frac{\tr\pK}{n} \iju \pg \,,
\end{equation}
so that
\begin{equation*}
	D_i \iju \pK = 0
\quad
 \Longleftrightarrow \quad D_i \iju \pL = 0 \,.
\end{equation*}

%% file: Acknowledgements.tex
{\noindent
\sc
Acknowledgements:}
 PTC acknowledges the hospitality of the Beijing Institute of Mathematical Sciences and  Applications in Huairou and  the Mathematisches Forschungsinstitut in Oberwolfach during part of work on this paper. 
 His
research was further supported in part by the NSF under Grant No. DMS-1928930
while he was in residence at the Simons Laufer Mathematical Sciences
Institute (formerly MSRI) in Berkeley during the Fall 2024 semester.
We are grateful to Ilka Agrikola, Bobby Beig, Lan-Hsuan Huang, Ines Kath, Rafe Mazzeo, Andre Neves, Nikolai Saveliev, Georg Stettinger, Erik Verlinde and Rudolph Zeidler for useful discussions or bibliographical advice. 
TQ and RW are grateful to the University of Vienna for hospitality.
 RW acknowledges support by the Heising-Simons Foundation under the “Observational Signatures of Quantum Gravity” collaboration grant 2021-2818 and the U.S. Department of Energy,
Office of High Energy Physics, under Award No.\ DE-SC0019470. 

%% file: Banados.tex
\section{All conformally smooth spacetimes which are vacuum near the conformal boundary}
 \label{s11VII24.1}

A pseudo-Riemannian manifold $(\mcM,\fourg)$ will be said to be \emph{conformally smooth} if there exists a manifold with boundary $\tmcM$, an embedding $i:\mcM\to\tmcM$ of $\mcM$ into the interior of $\tmcM$, a defining function $\Omega$ for the boundary of $\tmcM$ (by definition:  $\Omega$ vanishes precisely on the boundary and has non-vanishing gradient there) and a metric field $\tfourg$ which extends smoothly, as a metric, across the boundary,  such that 
\begin{equation}\label{24VII24.1}
  \fourg = i^*(\Omega^{-2}\tfourg)
   \,.
\end{equation}
Equivalently, if one thinks of $\mcM$ as a subset of $\tmcM$, then the metric $\Omega^2\fourg$ extends smoothly, as a metric, to the boundary of $\tmcM$.  

\input{AsymptoticSymm2}

%% file: AsymptoticSymm2.tex
It turns out that one can write explicitly all conformally smooth vacuum three-dimensional metrics~\cite{Banados:1998gg,Barnich:2010eb}. 
For this, consider three-dimensional metrics of the form
\begin{equation}
  \label{Feff1}
{\fourg } = \frac{d\rr^2}{\rr^2} +
  \underbrace{ g_{A B}(\rr, x^C) dx^A dx^B}_{=:g}
\end{equation}
%
with
\begin{equation}
  \label{Feff2}
  g(\rr, x^C) = \rr^2 \eta + O(1)\,,
\end{equation}
where $\eta$ is the Minkowski metric  $\eta = - dt^2 +d \varphi^2$ 
and $(x^A)=(t,\varphi)$.
The conformal completion is obtained by changing $r$ to a new coordinate $x=1/r$ and setting $\Omega=x$, with the conformal boundary at $x=0$.

 Barnich and Trossaert have shown~\cite{Barnich:2010eb} that three-dimensional Lorentzian
 metrics which are vacuum and have a smooth conformal completion as $r\to\infty$ can be written in the Ba\~nados form~\cite{Banados:1998gg}, for small $r\ne 0$,%
\footnote{The proof of this result makes use of the freedom of conformal transformations at the boundary at infinity. Similar expansions which do not exploit this freedom have been derived in~\cite{Skenderis:1999nb}.}
\begin{equation}
  \label{metrics}
  {\fourg }  = \frac{d \rr^2}{\rr^2}  - \left( \rr d x^+ - \frac{\mathcal{L}_-(x^-)}{\rr} d x^- \right)
  \left( \rr d x^- - \frac{\mathcal{L}_+(x^+)}{\rr} d x^+ \right)
  \,,
\end{equation}
where $x^\pm = t \pm \varphi$, 
with arbitrary functions $\mathcal{L}_\pm$. Here one needs
\begin{equation}\label{24IX23.1}
  \mathcal{L}_+
   \mathcal{L}_- \ne r^4
   \,,
\end{equation}
to avoid a vanishing determinant of the tensor field \eqref{metrics}.

If  the functions $\mathcal{L}_\pm$ are constants (or can be transformed to constants by a change of coordinates preserving the above form of the metric) one obtains the  Ba\~nados-Teitelboim-Zanelli black holes~\cite{Banados:1992wn}:
\begin{equation}
  \label{BTZmetric}
  {\fourg }  = \frac{d \rr^2}{\rr^2}  - \left( \rr d x^+ - \frac{M+J}{4\rr} d x^- \right)
  \left( \rr d x^- - \frac{M-J}{4\rr} d x^+ \right)
  \,,
\end{equation}
thus
\begin{equation}
  \label{changeparameters}
  \mathcal{L}_+ = \frac{M-J}{4}
  \,, \qquad
  \mathcal{L}_- = \frac{M+J}{4}
  \,.
\end{equation}

 Otherwise, in vacuum the metrics are expected to be singular (see~\cite{ChWutte} for some results on the time-symmetric case). Note that the character of the singularities
  cannot be probed with the curvature tensor, as all  these metrics are  locally isometric to the anti-de Sitter metric.

%% file: Vacuum.tex
\subsection{Vacuum time-symmetric geometrically finite initial data sets}
 \label{s1XI23.1} 
 
Spacelike general relativistic initial data consist of a set $(M,g,K)$, where $(M,g)$ is an $n$-dimensional Riemannian manifold and $K$ is a symmetric two-tensor on $M$.   In vacuum $(\matter =0)$ and under time-symmetry ($K_{ij}=0$) the scalar constraint equation, namely
 $$
 R_g= 2 \matter + 2\Lambda + |K|^2 _g- (\tr_gK)^2
 \,,
 $$
 where $R_g$ is the   scalar curvature of the metric $g$, becomes the requirement that $(M,g)$ has Constant Scalar Curvature (CSC). 
 
In this work we will mostly be interested in the case $n=2$. Recall that a two-dimensional manifold is called \emph{geometrically finite} if it has finite Euler characteristic. Equivalently, for non-compact manifolds, $M$ is diffeomorphic to a compact manifold from which a finite number of points, say  $\{p_i\}_{i=1}^N$,  $N\ge 1$, has been removed. 
 
 \begin{Remark}
  \label{R25VII24.1}
 On such a manifold,  complete metrics with constant negative scalar curvature 
 can be constructed by solving the two-dimensional Yamabe equation,  i.e.
\begin{equation}
	 \label{26VII24.1}
	 2\Delta_{\mzg} \zu = R_g  - e^{2\zu}R_{\mzg}
 \,,
\end{equation}
with $\pg = e^{-2\zu}\mzg$,
as follows: Let $\mzg$ be any metric on $M$ which, in local 
 polar coordinates near each of the $p_i$'s equals
  \begin{equation}\label{25VII24.1}
    \frac{dr^2}{r^2} + r^2 d\varphi^2
    \,.
  \end{equation}
The analysis in~\cite{AndChDiss}, 
reviewed in Appendix~\ref{s26VII24.2},
applies to provide existence, uniqueness, and polyhomogeneity  at the conformal boundary  
of a conformal factor $\zu$ such that the metric $g= e^{\zu}\mzg$ has constant scalar curvature equal to $-2$ (cf.\ also~\cite{TroyanovHulin} for existence). Since $R_{\mzg}=-2$ near the conformal boundary, defined as 
$$\{0=x:=1/r\}
\,,
$$
 expanding both sides of \eqref{26VII24.1} 
 in terms of the functions $x^i\ln^jx$ and comparing terms one finds an expansion with  no log terms and only even powers of $x$: 
\begin{equation}\label{26VII24.1rfk}
  \zu = u_2(\varphi) x^2 +  \frac{1}{10} \left(2u_2(\varphi)^2-u_2''(\varphi)\right)x^4  +
   \ldots
   \,,
\end{equation}
where $u_2$ is an arbitrary function on $S^1$. 
One can now transform to the coordinates \eqref{26VII24.1ty} below to obtain an asymptotic expansion with a finite 
number of terms. 

Assuming geometric finiteness and completeness, it follows from Theorem~\ref{t26XII23.1} below and the results in~\cite{ChWutte} that there exists a coordinate transformation which transforms the function $u_2(\varphi)$ to a constant $m\ge 0$. 
\qed
\end{Remark}
 
 The classification of two-dimensional, non-compact, geometrically finite,  negatively curved CSC metrics is well understood; a pedagogical presentation can be found in~\cite{Borthwick}. We review these models, following~\cite{ChWutte}.

\emph{Elementary} hyperbolic manifolds are the hyperbolic space 
and the  (complete) manifold $\R\times S^1$ with the metric
\begin{equation}\label{26XII23.1}
 \frac{dr^2}{r^2}+ r^2 d\varphi ^2\,,
 \quad
  e^{i\lambda \varphi}\in S^1\,,
  \quad
  \lambda \in (0,\infty)
  \,.
\end{equation}
We will refer to this surface as the  \emph{hyperbolic trumpet} (compare~\cite{Hannam:2009ib}).
Rescaling $\varphi$ and $r$, without loss of generality one can assume $\lambda = 1$, but then the angle $\varphi$ will range over  $[0,2\pi/\lambda]$.
 
Given $r_0\in \R$, the region $r\le r_0$ with a metric   \eqref{26XII23.1} will be referred to as a \emph{hyperbolic cusp}, and the region $r\ge r_0$ will be called a \emph{hyperbolic end}. Further \emph{hyperbolic ends} are defined as 
the manifolds $[r_0,\infty)\times S^1$ with metrics of the form 
\begin{equation}\label{26XII23.2}
 \frac{dr^2}{r^2- \mc }+ r^2 d\varphi ^2\,,
 \quad
  e^{i\lambda \varphi}\in S^1\,,
  \quad
  \lambda \in (0,\infty)
  \,,
\end{equation}
where $\mc \in \R$, with $r_0 >0$, with $r \ge  \sqrt{\mc}$ when $\mc>0$. Note that a rescaling of $r$, $\mc$ and $\varphi$ leads again to $\lambda =1$.
When $\mc>0$ one can indeed allow $r_0= \sqrt{\mc}$ because we have
\begin{equation}
     \label{27VII23.9110}
      g =
       \frac{dr^2}{ r^2  - \mc }  + r^2 d\varphi^2
       =
       d u^2 + \mc \cosh^2(u) \, d\varphi^2
       \,,
        \quad
        \mbox{where}\  e^{i\varphi} \in S^1
        \,.
     \end{equation}
   This further  shows that in this case the submanifold $r=r_0$ minimizes length, hence forms a closed geodesic.

A \emph{funnel} is defined as $[r_0,\infty)\times S^1$ with a metric \eqref{26XII23.2} with $\mc >0$, $r_0 = \sqrt{\mc}$, and $\lambda \ge 1$.
Note that a rescaling of $r$, $\mc$ and $\varphi$ leads to $\lambda =1$.
The boundary $\{r=\sqrt{\mc}\}$ is the shortest closed geodesic within the funnel.  In the physics literature funnels are  known as \emph{non-rotating Ba\~nados-Teitelboim-Zanelli (BTZ) black holes};
more precisely, time-symmetric slices of non-rotating BTZ black holes.
The geodesic $\{r= \sqrt{\mc}\}$ is   referred to as \emph{event horizon}, or  \emph{apparent horizon}, or \emph{outermost apparent horizon}.  
 In the associated vacuum spacetime, the surface $r=\sqrt{\mc}$ becomes the bifurcation surface of a bifurcate Killing horizon.
 
A \emph{hyperbolic bridge} is defined as the doubling of a funnel across its minimal boundary; equivalently, this is the rightmost metric \eqref{27VII23.9110} defined on $\{u\in \R \,,\, e^{i\varphi}\in S^2\}$.

The fundamental result is (cf., e.g., \cite[Theorem~2.23]{Borthwick}):

\begin{theorem}
  \label{t26XII23.1}
Consider a complete non-compact two-dimensional hyperbolic manifold $(M,g)$ with finite Euler characteristic. Then either $(M,g)$ is hyperbolic space, or a hyperbolic trumpet, or it is the union of a compact set with a finite number of cusps and a finite number of funnels.
\end{theorem}

%% file: ConfTransf2.tex
Let us recall the formula for  the behavior of the divergence of a tensor under conformal rescalings.
Regardless of the metric $\tg$, under the rescaling $\ijd\pg = e^{-2u} \ijd\tg $
the Christoffel symbols transform as 
\begin{equation}\label{ChristoffelTransfo}
	{\tilde{\Gamma}}^i_{jk} = \Gamma^i_{jk} + \delta^i_j u_{,k} + \delta^i_k u_{,j} - g^{il}g_{jk}u_{,l} 
 \,.
\end{equation}
Hence, in dimension $n$, for a $\qg$-traceless symmetric tensor $\tL^{ij}$,
\begin{equation*}
	\begin{split}
		\tD_i \tL^{ij} & = D_i \tL^{ij} + n u_{,k} \tL^{kj} + ( \delta^j_i u_{,k} + \delta^j_k u_{,i} - \pg^{jl}\pg_{ik}u_{,l} )\tL^{ik} \\
		& = D_i \tL^{ij} + (n+2) u_{,k} \tL^{kj} \\
		& = D_i (e^{(n+2)u} \tL^{ij}) e^{-(n+2)u} \,. 
	\end{split}
\end{equation*}
For $n=2$ one obtains 
%
\begin{equation}\label{24VII24.31nb}
D_i L^{ij} = 0
 \quad
  \Longleftrightarrow 
   \quad
   \tD_i \tL^{ij} = 0
   \,,
\end{equation}
provided that 
\begin{equation}\label{LTransfo}
	\pL^{ij} = \tL^{ij} e^{4u} \,.
\end{equation}

There is an obvious corresponding formula with matter fields. Writing again
	%
	\bel{Scalingu5}
		\ijd\pg = e^{-2u} \ijd\tg \,.
	\ee
and setting  
\begin{equation}\label{LTransfo1}
	\pL^{ij} = \tL^{ij} e^{4 u} \,,
    \quad
 	J^{i} = \tilde J^{i} e^{-4 u} \,,
\end{equation}
we find
\begin{equation}\label{24VII24.31}
D_i L^{ij} = J^j 
 \quad
  \Longleftrightarrow 
   \quad
   \tD_i \tL^{ij} = \tilde J^j\,.
\end{equation}

%% file: Airy.tex
\subsection{The vacuum vector constraint equation and Airy functions}
 \label{ss24VII24.1}
 
In two dimensions,  
symmetric tensors  which are divergence-free with respect to the flat metric 
can be obtained from an \emph{Airy function} $\alpha$. 
 Indeed, and in any dimension, given any function $\alpha$ the tensor field 
\begin{equation}\label{24VII24.11a}
k_{ij} = \delta_{ij} \Delta \alpha - \partial_i \partial_j \alpha 
\end{equation}
obviously satisfies
\begin{equation}\label{24VII24.12}
  \partial_ik^{ij} = 0
  \,,
\end{equation}
where the indices have been raised with the flat metric. Conversely, in two space dimensions 
any symmetric tensor field satisfying \eqref{24VII24.12} can be written in the form \eqref{24VII24.11a}.%
\footnote{We are grateful to Bobby Beig for pointing this out. 
Equation~\eqref{24VII24.11a} can be viewed as the two-dimensional version of the constructions in~\cite{ChBeigTT,ChBeig2020}.}

We have
\begin{equation}\label{24VII24.13}
  k^{ii}=0
  \qquad
  \Longleftrightarrow
  \qquad
  \Delta_\delta \alpha =0
  \,.
\end{equation}
Hence any harmonic function $\alpha$ provides a trace-free solution $\qL^{ij}$ of the vacuum vector constraint equation in the flat metric. 
Using \eqref{LTransfo} one can thus obtain all  trace-free and divergence-free tensors for any metric conformal to the flat metric,
e.g.\ the hyperbolic metric in the half-space model or in the Poincar\'e-disc model.

Note that nontrivial compactly supported divergence-free symmetric tensors can be obtained in this way, by taking $\alpha$ to be compactly supported. But no such tensors will 
have constant trace.
Indeed, it follows from \eqref{24VII24.11a}  that $\alpha - k^{ii}(x^2+y^2)/2$ would then be a compactly supported harmonic function, hence zero by  the maximum principle, or by Liouville's theorem, or by unique continuation in more general settings. This means that one cannot use the conformal covariance \eqref{24VII24.31} to obtain compactly supported \red{$TT$} tensors.

%% file: VecCstr.tex
\subsection{The vacuum vector constraint equation and holomorphic functions}
 \label{s29VII24.1}

We first consider open subsets of the Euclidean plane, writing as before
$$\qg_{ij}=\delta_{ij}
 \,.
 $$
Any $\qg$-traceless symmetric tensor takes the form
$$
\qL = \Lxx d\bx^2 + 2 \Lxy d\bx d\by - \Lxx d\by^2
\,.
$$
%
Using the Euclidean metric and Euclidean coordinates to raise and lower indices on $\qL$, we have
$\iju \qL = \ijd \qL$. The condition that $\qL$ be divergence-free reads
 \ptcheck{17VI}
\begin{equation*}
	\begin{cases}
		\partial_{\bx}\Lxx + \partial_{\by}\Lxy = 0 \\
		\partial_{\bx}\Lxy - \partial_{\by}\Lxx = 0
	\end{cases}
	\quad\Longleftrightarrow\quad (\partial_{\bx} + {\rm i}\partial_{\by})(\Lxy + {\rm i} \Lxx) = 0 \,.
\end{equation*}
Thus $\qL$ satisfies the vector constraint equation if and only if the function 
\begin{equation}\label{24VI24.11}
 \qf(\bx,\by) := \Lxy + {\rm i} \Lxx
\end{equation}
is holomorphic on its domain of definition. So holomorphic functions immediately provide a wealth of solutions of the vacuum
vector constraint equation.
We will say that $\qL$ derives from a holomorphic function $\qf$. 

The   calculation
 \ptcheck{17VI}
\begin{equation}\label{19IV24.p1}
	\qL = \Lxx(d\bx^2-d\by^2) + 2 \Lxy d\bx \, d\by = \frac 1{2{\rm i}} ( \qf dz^2 - \bar  \qf d\bar z ^2) =
 \Im (\qf dz^2)
	\,,
\end{equation}
relates $\qL$ to  so-called ``quadratic differentials''.

 In view of what has been said so far, we find: a traceless symmetric tensor $L^{ij}$ satisfies $D_i \pL^{ij} = 0$ in the metric $	\ijd\pg = e^{-2\qp}   \delta _{ij}  \quad\Longleftrightarrow\quad e^{-4\qp} \pL^{ij} $ derives from a holomorphic function $\qf$.

A tensor field $t^{ij}$  is called transverse if $D_i t^{ij}=0$; transverse-traceless ($TT$) if moreover its trace vanishes.

We have:

\begin{prop}
 \label{p24VII24.1}
 On non-compact two-dimensional manifolds 
there are no non-trivial compactly supported symmetric \red{$TT$ tensor}s.
\end{prop}

\proof
In local coordinates in which the metric is conformally flat,  the tensor $\qL^{ij}$ vanishes if and only if $\qf$ vanishes.
On hyperbolic space our claim is simply the  statement that the only  holomorphic functions on the unit disc which are constant near the unit circle are constants. More generally, the result follows from unique continuation for holomorphic functions in local coordinates and a simple covering argument.
\qedskip

For further use we note
 \ptcheck{17VI}
\begin{equation}\label{19IV24.p4}
	\normsq{\qL}{\delta} = \delta^{ik}\delta^{jl}\qL_{ij}\qL_{kl} = 2\Lxx^2 + 2\Lxy^2 = 2|\qf|^2
	\,.
\end{equation}
%
%
From  \eqref{Scalingphi}, \eqref{LTransfo}, and \eqref{19IV24.p4} we conclude that
\begin{equation}\label{19IV24.p5}
	\normsq{\pL}{\pg} = 2 e^{4\qp} |\qf|^2
	\,.
\end{equation}

%% file: Poles.tex
\subsection{Poles at the conformal boundary}
 \label{ss11VIII24.8}

\ptcn{symplectic form should be checked; what would be the poisson bracket in a space with such solutions?}

An interesting class of solutions arises from functions $\qf$ which are meromorphic near the unit disc with poles exactly on the unit circle.
Thus, given a finite number of points $e^{i\theta_j}$, $j=1,\ldots,N$, lying on the unit circle, together with
nonzero numbers $a_j\in \mathbb{C}$, and  a
function $f$ which is holomorphic in a neighborhood of the unit disc $D(0,1)$, we set
\begin{equation}\label{19IV24.p1c}
	\qL =  
 \Im (\qf dz^2)
\end{equation}
with 
\begin{equation}\label{11VII24.31}
  \qf(z) = f(z) 
   + \sum_{j=1}^{N}\frac{a_j}{z-e^{i\theta_j}} 
  \,.
\end{equation}
We show in Appendix~\ref{s11VII24.1b}
 that such tensors $\qL$ lead, through the conformal method, to initial data sets with a conformal completion where the conformal factor extends in $C^1$ across the conformal boundary (cf. Proposition~\ref{p2VIII24.1}). 
 
We check in  Section~\ref{ss2VII24.1b} below that  the finite-mass condition \eqref{11VIII24.41} is satisfied by the functions \eqref{11VII24.31}. 
 
 We provide a formal asymptotic expansion of the associated solutions of the Lichnerowicz equation near each pole in Appendix~\ref{s11VII24.1b} as well. For this, it is convenient to work in the half-plane model of hyperbolic space; a formula how poles transform when passing to the half-plane model can be found in Appendix~\ref{App11VIII24.2}.

%% file: ConformalLaplacian.tex
\subsection{Matter currents and the conformal vector Laplacian}
 \label{ss24VII24.11b}
 
 Consider the vector constraint equation with sources, 
\begin{equation}\label{24VII24.31b}
D_i L^{ij} = J^j  
 \,.
\end{equation}
The standard way to solve this equation, usually in dimensions $n\ge 3$, appeals to the conformal vector Laplacian: One writes
\begin{equation}\label{24VII24.31c}
L^{ij} = D^i X^j + D^j X^i - \frac{2}{n} D^k X_k g^{ij}
 \,.
\end{equation}
The method can also be used in dimension $n=2$, leading to an elliptic equation for the vector field $X$: 
\begin{equation}\label{24VII24.31bb}
D_i
  \big(
   D^i X^j + D^j X^i -  D^k X_k g^{ij}
    \big) 
     = J^j  
 \,.
\end{equation}
Commuting derivatives and using $R_{ij} = R g_{ij}/2$, this can be rewritten as a Laplace equation
\begin{equation}\label{24VII24.31bbb}
D_i D^i X^j + \frac R 2 X^j = J^j  
 \,.
\end{equation}
Integration by parts shows  that when $R\le 0$,  the operator at the left-hand side has no kernel on fields which decay sufficiently fast as a conformal boundary at infinity is approached, or when solving the equation on a bounded set with zero Dirichlet data on the boundary. This further guarantees that there are no conformal Killing vector fields on negatively curved manifolds with appropriate asymptotic or boundary conditions, and finally implies unique solvability for compactly supported or rapidly decaying sources. 

%% file: FiniteMass1.tex
	\subsection{A ``finite-mass'' condition?}
 \label{ss2VII24.1b}

Typical physically-relevant solutions have finite mass. In space-dimensions $n~\ge~3$, the usual Hamiltonian analysis of the notion of mass (cf., e.g., \cite{CJL}) leads to the condition that the trace-free part $\pL$ of the extrinsic curvature satisfies  
\begin{equation}
	\int_{M }  \normsq\pL\pg \,d\mu_{\pg} < \infty
  \,.
   \label{11VIII24.41}
\end{equation}
Now, when $n=2$ the question of a well-defined Hamiltonian mass is less clear-cut~\cite{ChWutte};
we review this in Section~\ref{s14VIII24.1} below.
It seems nevertheless natural to impose this condition also in the two-dimensional case; we will do it in what follows.

Given the above let us assume, in vacuum,  the existence of a coordinate neighborhood $[r_1,1]\times S^1$   of  the conformal boundary, which we assume to be located at 
$\{1\}\times S^1$, on which the metric is conformally flat as in \eqref{Scalingphi}. Then, for \red{$TT$ tensor}s arising from a holomorphic function $\qf$ one finds
\begin{equation*}
	\int_{M }  \normsq\pL\pg \,d\mu_{\pg}
 = 
	\int_{M }   e^{2\qp} \normsq\qL\delta \,d\mu_{\delta}
 \ge  \int_{\red{[r_1,1]}\times S^1} 2 e^{2\qp} |\qf|^2 d\mu_{\delta}  
 \,.
\end{equation*}
In the unit-disc Poincar\'e model the conformal factor $e^{-\qp}$ behaves as $1/(1-r)$ for $r$ near to $1$, 
which leads to the condition  
\begin{equation*}
	 \int_{r_1}^{1} \oint (1-\br)^{2} \normsq\qL\delta  d\br d \varphi = 2 \int_{r_1}^{1} \oint (1-\br)^{2} |\qf|^2 d\br d \varphi < \infty \,.
\end{equation*}
\tqn{$(r,\varphi)$ coordinates as in \ref{s11VII24.1}}
This will be satisfied if
\begin{equation} \label{FiniteMassCond}
	|\qf| = \gO{{(1-\br)^{-\alpha}}}  
\textrm{ with } \alpha < 3/2 \,.
\end{equation}
In particular, for functions $\qf$ which are meromorphic in a $\mathbb{C}$-neighborhood of $S^1$, 
the condition allows $\qf$ to have poles of order one on the unit circle, but not higher. 

%% file: Lichnerowicz.tex
%
\subsection{The Lichnerowicz equation}
 \label{ss2II25.1}

  \ptcn{there are imaginary    Killing spinors on cusps, to make it nontrivial needs matter models on cusps, check with Troyanov and Mazzeo Taylor?}
With the scaling \eqref{Scalingu}, the scalar curvature becomes  
 \ptcheck{28VII24, with Ergebnisse}
\begin{equation}
	R_{\tg} = e^{-2u} (R_{\pg} - 2\Delta_{\pg} u) \,,
\end{equation}
where
$$
\Delta_{g} v =\frac{1}{\sqrt{\det g}} \partial_k \left(\sqrt{\det g}\, g^{kl} \partial_l v \right)
$$
is the Beltrami-Laplace operator, which under conformal rescalings 
		$\ijd\pg = e^{-2u} \ijd\tg$ transforms as
\begin{equation}
	\Delta_{\pg} = e^{2u} \Delta_{\tg}\,.
\end{equation}
We set
\begin{equation}\label{TraceFreeLasdf}
	\iju \pL = \iju \pK - \frac{\tr\pK}{2} \iju \pg \,,
\end{equation}
without necessarily assuming that $\tr\pK$ is constant.
Then
\begin{equation}
	\normsq\pK\pg = \normsq\pL\pg + \frac{1}{2}(\tr_{\pg} \pK )^{2}  \,,
\end{equation}
which leads to the following form of the scalar constraint equation:
 \ptcheck{28VII24}
\begin{equation}\label{29VII24.3}
2\Delta_{\tg} u = -R_{\tg} + e^{2u} \normsq\tL\tg +  e^{-2u}
    \big(2 \Lambda_\tau + \matter)\,,
\end{equation}
with 
$$
 \Lambda_\tau = \Lambda - \frac{1}{4}(\tr_gK)^{2}  
  \,.
$$
To avoid ambiguities: $\tg$ is a seed metric, which can be prescribed arbitrarily, and the conformally rescaled metric $\pg$ is the initial data metric
which is required to satisfy the scalar constraint equation.

In what follows, unless explicitly specified otherwise we will assume that 
\begin{equation}\label{28VII.51}
 \Lambda_\tau   <0 
  \,.
\end{equation}

%% file: BdryRegularity1.tex
\subsection{Boundary regularity}
 \label{ss11VIII24.2}

We first consider the  vacuum case, $\matter=0=J^i  $. We assume that $\tr\pK$ approaches a constant
 as the conformal boundary at infinity, say $\{x=0\}$, is approached; 
 otherwise, the metric is unlikely to be locally asymptotically hyperbolic.  
 Redefining  $\Lambda$ as 
\begin{equation}\label{29VII24.1}
  \Lambda \mapsto \Lambda - \lim_{x\to0}\frac{(\tr\pK)^2}{4}
   < 0 
  \,,
\end{equation}
shifting $K_{ij}$ by its asymptotic trace, and rescaling the metric by a constant, we can without loss of generality assume that
\begin{equation}\label{28VII24.52}
  \Lambda=-2
  \,,
   \quad
   \tr\pK\to_{x\to0} 0
   \,. 
\end{equation}
Let us assume further that $R_{\tg}=-2$ near the conformal boundary (compare 
Remark~\ref{R25VII24.1}). Near every point at the conformal boundary, we can then find \emph{local} coordinates so that $\tg$ takes the
form \eqref{26VII24.1ty} with $f_2=0$  
\begin{equation}
 \tg = x^{-2} (dx^2 + dy^2)
 \,.
  \label{2VIII24.2}
 \end{equation}
Since $\tr K = O(x)$, we have  the following asymptotics
\begin{eqnarray}\label{29VII24.2} 
  \frac 14 ( \tr K)^2  &= &
    k_2(y) x^2 + k_3(y) x^3 + \ldots  
   \,.
\end{eqnarray}

 Let us further assume that  $\qL^{ij}$ extends smoothly to $\{x=0\}$ in these coordinates; recall that this will be the case for \red{$TT$ tensor}s $\qL$ arising from a function $\qf$ which is holomorphic near $\{x=0\}$, as in Section~\ref{s29VII24.1}. Then
\begin{eqnarray}\label{28VII24.53}
 \qL^{ij}  & = & O(1)
  \,,
 \\
 \tL^{ij}  & = & x^4 \qL^{ij} = O(x^4)
  \,,
 \\
 |\tL|_{\tg}^2  & = & x^{-4} \tL^{ij}\tL^{ij}
 \nonumber
 \\
&= &
  \tltwo_4(y) x^4
   +  \tltwo_5(y) x^5
   + \ldots 
   \,.
\end{eqnarray}
Appendix~\ref{s26VII24.2} guarantees existence, uniqueness, and polyhomogeneity of solutions of the associated solutions of constraint equations when $\tr K\equiv 0$. A similar result can be obtained by perturbation methods for 
small non-constant $\tr K$  which extend smoothly across the conformal boundary.

 Inserting the asymptotic expansions above in the Lichnerowicz equation \eqref{29VII24.3} one finds the following polyhomogeneous
 expansion of $u$, for small $x$:
 \begin{equation}\label{31VII24.1}
   u(x,y) = -\frac{k_2(y)}{3} x^2\log(x) + u_2(y) x^2
    -\frac{k_3(y)}{4} x^3+ 
     O\big(
      x^4\log^2(x)\big)
   \,,
 \end{equation}
with a function $u_2$ which is uniquely defined by $(M,\tilde g)$ up to asymptotic symmetries~\cite{ChWutte}.

If $k_2\equiv 0$, 
an identical calculation gives instead 
\begin{eqnarray}\label{29VII24.4} 
  \frac 14 ( \tr K)^2  &= &
    k_4(y) x^4 + k_5(y) x^5 + \ldots  
   \,.
\end{eqnarray}
It follows from \cite{AndChDiss} (compare Appendix~\ref{s26VII24.2} below) that  the solutions are then conformally smooth, 
with 
 \begin{equation}\label{31VII24.2}
   u(x,y) =  u_2(y) x^2
   +
   \frac{1}{20} \left(-2 k_4(y)+\ell_4(y)-2 u_2''(y)-4 u_2(y)^2\right)
   x^4
   + 
     O(
      x^5)
   \,.
 \end{equation}
In particular, in the CMC case the solutions are conformally smooth. 

%% file: UhlenbeckLorentzian.tex
\subsection{Explicit solutions of the vacuum Cauchy problem}
 \label{s2X24.1}  

Let $(\Sigma,\gamma)$ be a two-dimensional Riemannian manifold and suppose that the Ricci scalar $R(\gamma)$ of $\gamma$ 
satisfies
\begin{equation}\label{23IX24.1a}
  R(\gamma) = -2 \sigma^2 +  |K|_\gamma^2 - (\tr_\gamma\! K)^2
  \,,
\end{equation} 
where  $\sigma>0$ is a constant, and where the symmetric tensor field $K_{ab}$ satisfies the two-dimensional vacuum vector constraint equation,
$$
D_a(K^{ab} - \tr_{\gamma}\! K \gamma^{ab}
)
 = 0
  \,.
$$
Then the metrics  (cf.\ e.g.~\cite{KrasnovSchlenker})
\begin{align}\label{23IX24.1xdb}
  g  =
     -dt^2
  + \gamma^{ab} 
  & \
  \big(
    \cos ( \sigma t) \gamma_{ac} +  \sigma^{-1} \sin ( \sigma  t)  K_{ac}
    \big)  
    \nonumber
\\
 & \quad   
      \times
      \big(
    \cos ( \sigma  t) \gamma_{bd}   +  \sigma^{-1} \sin ( \sigma  t)  K_{bd}
    \big)  
    dx^c dx^d
\end{align}
satisfy the vacuum  Einstein equations with a negative cosmological constant. The extrinsic curvature tensor of $\{t=0\}$ equals $K$.
The hypersurface $\{t=0\}\approx \Sigma$ has constant mean curvature if the $\gamma$-trace of $K$ is constant, vanishing if the trace of $K$ vanishes.

The tensor fields  \eqref{23IX24.1xdb}  are globally defined on $\R\times\Sigma$ but the signature always drops down at  the non-empty set of  spacetime points at which $ \cos ( \sigma t) \gamma_{ac} +  \sigma^{-1} \sin ( \sigma  t)  K_{ac}$ degenerates. The nature of the boundaries that so arise requires careful analysis, we will return to this elsewhere. 
When $K\equiv 0$ one obtains a subset of the Anti-de Sitter spacetime; cf., e.g., \cite[Chapter 5]{Griffiths:2009dfa}.

One can take $(\gamma,K)$ to be ALH,  
in which  case \eqref{23IX24.1xdb} provides an explicit  evolution of the vacuum data discussed in the previous sections. In such a case the  initial data manifold $(\Sigma,\gamma)$ has  a smooth conformal completion at infinity. 
Whether or not the spacetime obtained by evolution of the data $(\Sigma,\gamma,K)$ has such a completion requires further analysis.  
 
Configurations where $\gamma_{ab}$ has several ALH ends,  with non-zero $K_{ab}$ 
 vanishing at large distances, 
 contain compact marginally outer trapped surfaces~\cite{GSW2d,AEM}. Such surfaces are often accompanied by event horizons, and if so the associated spacetimes would   describe dynamical black holes. But note that the usual notion of event horizon requires a conformal completion at infinity, the existence of which will be analysed elsewhere.
\ptcn{global structure to understand?}

When  $\Sigma$ is compact one obtains an explicit form of the spacetimes considered in e.g.~\cite{MoncriefTeichmueller,Moncrief:2+1,Witten2+1}, but not in a CMC time-slicing.

A complete description of $(2+1)$-dimensional maximal globally hyperbolic vacuum spacetimes with a negative or vanishing cosmological constant can be found  in \cite{Mess,AnderssonEtAlOnMess,Barbot,BenedettiBonsante}, and  in \cite{Scannell} with positive $\Lambda$.

We finally note the positive-cosmological-constant counterpart of the metrics~\eqref{23IX24.1xdb}: 
\begin{eqnarray}
  g 
    & =  &
    -dt^2   + \gamma^{ab} 
  \big(
    \cosh(  \sigma t) \gamma_{ac} + \sigma^{-1} \sinh(  \sigma t)  K_{ac}
    \big) 
     \nonumber
\\
 & &
 \qquad 
     \times  \big(
    \cosh( \sigma t) \gamma_{bd} +  \sigma^{-1} \sinh( \sigma t)  K_{bd}
    \big) dx^c dx^d
    \,.
    \label{18IX24.1bc}
\end{eqnarray}
These  will  satisfy the Einstein equations  with a positive cosmological constant if $K$ satisfies again the vacuum vector constraint equation and if 
\begin{equation}\label{23IX24.1}
  R(\gamma) = 2 \sigma^2 +  |K|_\gamma^2 - (\tr_\gamma\! K)^2
 \,.
\end{equation}
%
%

The  metrics in this section are the Lorentzian counterpart of the Riemannian metrics written-down by Uhlenbeck in~\cite{UhlenbeckHyperbolic}.%
\footnote{We are grateful to Andre Neves for bringing this paper to our attention, and to Lan-Hsuan Huang for sharing her notes on the Uhlenbeck metrics.}
 The fact that they solve the Einstein equations follow from the calculations there, together with a straightforward variation of the usual constraint-propagations argument.

%% file: ang-momentum-minimal.tex
 \section{Angular momentum and mass}
\label{s14VIII24.1}

\renewcommand{\hyp}{\mycal{S}}

Global invariants such as energy, momentum, and angular momentum, provide useful information about three-dimensional general relativistic initial data sets. For example, suitably regular, say vacuum, asymptotically hyperbolic initial data sets   with vanishing \textit{mass} can be embedded in Anti-de Sitter spacetime. Another example is provided by the statement that stationary, again suitably regular, asymptotically flat vacuum spacetimes with vanishing \textit{angular momentum} are static. Therefore it is of interest to reexamine these invariants in the context of two-dimensional initial data sets. 
 
Consider then an $(n+1)$-dimensional spacetime containing a spacelike hypersurface $\hyp$. On $\hyp$ consider the collection of general relativistic initial 
data sets $(g,K)$ which asymptote to background fields $(\mathring g, \mathring K)$ induced by a spacetime background metric $\mathring \fourg$. Let $\mathring X$ be a Killing vector field of $\mathring \fourg$, let $\mathring Y $ denote the part of $\mathring X$ tangent to $\hyp$, as decomposed with respect to $\mathring \fourg$, and let   $\mathring V$ denote its component along the $\mathring\fourg$-unit normal to $\hyp$. 
Suppose that the spacetime contains a region
$\mcM_{ext} \subset \mcM$ together with a diffeomorphism
\begin{equation}
  \label{Nman}\Phi^{-1}:\mcM_{ext} \to [R_0,\infty)\times N^{n-1}
   \,,
\end{equation}
for some $R_0\in\R$, 
where $N^{n-1}$ is a compact $(n-1)$-dimensional manifold.
It was shown in \cite{CJL} that
 the Hamiltonian charge associated with the flow of $\mathring X$ equals 
\begin{equation}
    H (\red{\mathring V},\red{\mathring Y}):=
  c_n
    \lim_{R\to\infty} \int_{r\circ \Phi^{-1}=R}
\left(\ourU^i(\red{\mathring V}
 )
  +  {\mathbb V}^i(
  \red{\mathring Y} 
)\right) dS_i
\,,
 \label{17VIII24.1}
\end{equation}
with $dS_i$ being the hypersurface form $\partial_i\rfloor dx^1\wedge \cdots \wedge dx^n$, and 
where 
\begin{eqnarray} \label{eq:3.3} & {}\ourU^i (\red{\mathring V}):=  2\sqrt{\det
    g}\,\left(\red{\mathring V}g^{i[k} g^{j]l} \zD_j g_{kl}
    +D^{[i}\red{\mathring V} 
    g^{j]k} e_{jk}\right) 
    \,,
    \end{eqnarray}
   \bel{matVn} {\mathbb V}^l( \red{\mathring Y}):=  2\sqrt{\det g}\left[ (P^l{_k}
   -\zP^l{_k}) \red{\mathring Y}^k -\frac12
      \red{\mathring Y}^l\zP^{mn}e_{mn} +\frac12
       \red{\mathring Y}^k \zP^l{_k}\red{\mathring g}^{mn}e_{mn}  \right] 
     \,, 
   \ee
with 
\begin{equation}
P^{kl}:=g^{kl}\trg  K- K^{kl}\,, \qquad \trg  K:=g^{kl}K_{kl}\,,
\label{defP}
\end{equation}
$$
 e_{ij}:= g_{ij}-\mathring g_{ij}
 \,,
$$
and with $\mathring D$ denoting the Levi-Civita covariant derivative of $\mathring g$. 
Similar relations for
$\mathring K$ and $\zP$. 
Indices on $K$ and $P$ are moved
with $g$ while those on $\mathring K$ and $\zP$  with $\mathring g$. Here, (anti)-symmetrization of 
two-tensors is defined with a factor $1/2$. The formulae apply in all space-dimensions $n\ge 2$, and
we choose the dimension-dependent multiplicative constant $c_n$ in \eqref{17VIII24.1} to be equal $1/(2\pi)$ when $n=2$.

It follows from the Killing equations that the function $\red{\mathring V}$ and the vector field $\red{\mathring Y}$ are solutions of the \emph{background KID equations} \cite{ChBeigKIDs,CJL}. 
%
So, an equivalent perspective is to consider the Hamiltonian charges as 
being associated with the solutions $(\mathring V, \mathring Y)$ of these equations, called Killing Initial Data (KIDs).

\subsection{Ba\~nados metrics}
\label{sec:spacetime}

Let us choose the background spacetime metric to be 
%
\begin{equation}
 \bfour = - r^2 dt^2 + \frac{dr^2}{r^2} + r^2 d \varphi^2\,.
 \label{backgroundban}
\end{equation}
The Riemannian background metric $\mathring g$ induced on the level sets of $t$, and the background extrinsic curvature $\mathring K$ of these level sets read
%
\begin{equation}
 \mathring g  = ({\theta}^2)^2 + ({\theta}^1)^2\,, \qquad {\mathring K}_{i j} = 0\,,
 \label{7IX24.15} 
\end{equation}
where we use the following $\bfour$-orthonormal frame
\begin{equation}
    \{{\theta}^2 = \frac{d \rr}{r}\,,  {\theta}^1 = \rr d \varphi\,,  {\theta}^0 = \rr d t\,\}
     \,.
     \label{7IX24.13}
 \end{equation}
The Cauchy data $(g,K)$ induced on the level sets of  $t $ by $\gfour$, given by \eqref{metrics}, 
are found to be 
\begin{align}
    g 
     &
     =
    \
       \theta^2 \theta^2
  +\left(1 + \left(\mathcal{L}_-(\red{t} - \varphi)+\mathcal{L}_+(\red{t} + \varphi)\right) \rr^{-2} + \mathcal{L}_+(\red{t} + \varphi)\mathcal{L}_-(\red{t} - \varphi)\rr^{-4}\right) \theta^1 \theta^1
   \nonumber
\\
  & =: 
  \, 
    \mathring g + \red{r^{-2}}\mu_{ij} \theta^i \theta^j + O (\red{r^{-4}}) 
      \label{bansolutionsONB}\,, \\
    K 
    &
    =
    \ \left( \frac{2 (\mathcal{L}_+(\red{t} + \varphi)-\mathcal{L}_-(\red{t} - \varphi))}{r^2} + O(r^{-4})
    \right)
    \theta^1 \theta^2 \nonumber
    \\ 
 &~-
    \left(
    \frac{\mathcal{L}_-'(\red{t}-\varphi
   )+\mathcal{L}_+'(\red{t}+\varphi )}{2
   r^3} +O(r^{-5}) \right)
    \theta^1 \theta^1 
    \,,
    \label{2VIII24.31}
\end{align}
where $f'(x) = df(x)/dx$, and where $K$ is defined using the future pointing normal.

\input{mass}
The integrand of the 
$1/(2\pi)$-normalized total mass integral
 will be called the
\emph{mass aspect function}, denoted by $\mu$, and  reads
\begin{equation}
  \mu := \mu_{22}+2 \mu_{11} = 2 \left(\mathcal{L}_-(\red{t} - \varphi)+\mathcal{L}_+(\red{t} + \varphi)\right) \,.
   \label{8IX24.1}
\end{equation}
%
The integrand, say $\myj$, of the limit, as $r$ tends to infinity, of
$1/(2\pi)$-normalized total   
angular momentum integral on
circles of constant $t$ and $r$, equals
\begin{equation}
  \label{angmomentum}
\myj=
   2(\mathcal{L}_-(\red{t} - \varphi)- \mathcal{L}_+(\red{t} + \varphi))
    \,.
\end{equation}
We will refer to $\myj$ as the \emph{angular-momentum aspect function}.%

We note that for the metrics under consideration we have
\begin{equation}\label{31VII24.6}
  (\partial_{t} + \partial_\varphi) (\mu+ \myj) = 0
   \,,
   \qquad
  (\partial_{t} -\partial_\varphi) (\mu- \myj) = 0
   \,,
\end{equation}
in particular both $\mu$ and $\myj$ satisfy the two-dimensional wave equation.

Upon imposing 
\begin{equation}
  \mathcal{L}_-(\red{t} - \varphi) - \mathcal{L}_+(\red{t} + \varphi) = 0
  \ \mbox{and} \ 
 \mathcal{L}_-'(\red{t}-\varphi
  )+\mathcal{L}_+'(\red{t}+\varphi ) = 0
\end{equation}
the extrinsic curvature tensor of the level sets of $t$ vanishes.

\subsubsection{Example: Ba\~nados-Teitelboim-Zanelli black holes} 
 \label{ss21X24.1}
\input{BTZmetricCauchy}

\renewcommand{\bv}{v}

\subsection{Asymptotic Symmetries}
\label{subsec:asym}

We consider coordinate  transformations of the form 
\begin{align}
  \label{xprr}
    x^+ &= \bv_+({\bar x}^+)
    + \frac{\bv^\prime_+({\bar x}^+) \bv^{\prime \prime}_-({\bar x}^-)}{2 {\bar\rr}^2 \bv^{\prime }_-({\bar x}^-)}
    \nonumber \\
    &~~~~+ \frac{\bv_+''({\bar x}^+) \left(4
    \mathcal{L}_-(\bv_-({\bar x}^-))
    \bv_-'({\bar x}^-)^4+\bv_-''({\bar x}^-)^2
    \right)}{8 {\bar\rr}^4 \bv_-'({\bar x}^-)^2}
    + O\left(\frac{1}{{\bar\rr}^6} \right)
    \,, \nonumber \\
    x^- &= \bv_-({\bar x}^-)
    +\frac{\bv_-'({\bar x}^-) \bv_+''({\bar x}^+)}{2 {\bar\rr}^2
    \bv_+'({\bar x}^+)} \nonumber \\
    &~~~~+ \frac{\bv_-''({\bar x}^-)
    \left(4 \mathcal{L}_+(\bv_+({\bar x}^+))
    \bv_+'({\bar x}^+)^4+
    \bv_+''({\bar x}^+)^2\right)}{8 {\bar\rr}^4
    \bv_+'({\bar x}^+)^2}
    + O\left(\frac{1}{{\bar\rr}^6} \right)
    \,, \nonumber
\\
    \rr &=  \frac{{\bar\rr}}{\sqrt{\bv_+'({\bar x}^+)\bv_-'({\bar x}^-) }} -\frac{\bv_-''({\bar x}^-) \bv_+''({\bar x}^+)}{4
    {\bar\rr}  \sqrt{\bv_-'({\bar x}^-)
    \bv_+'({\bar x}^+)}^3}
    + O\left(\frac{1}{{\bar\rr}^3} \right)
\,,
\end{align}
where $\bv_\pm$ are, say smooth, functions satisfying
$$
 \bv_-({\bar x}^-+ 2\pi) =
 \bv_-({\bar x}^-)+ 2\pi
  \,,
   \qquad
 \bv_+({\bar x}^++ 2\pi) =
 \bv_+({\bar x}^+)+ 2\pi
  \,.
$$
These transformations 
 preserve
the precise form \eqref{metrics} of the metric, to the order induced by \eqref{xprr}, with $\mathcal{L}_\pm$ replaced by new functions $\bar{\mathcal{L}}_\pm$ given by
\begin{subequations}
    \label{barLs}
    \begin{align}
        \bar{\mathcal{L}}_+(\bar{x}^+) &= \mathcal{L}_+(\bv_+({\bar x}^+)) \bv_+'^2({\bar x}^+) - \frac{1}{2}
        \red{\hSchw\big[\bv_+'](\bar x^+)}
         \,,\\
        \bar{\mathcal{L}}_-(\bar{x}^-) &= \mathcal{L}_-(\bv_-({\bar x}^-)) \bv_-'^2({\bar x}^-)
         - \frac{1}{2}\red{\hSchw\big[\bv_-'\big](\bar x^-)}\,,
    \end{align}
    \end{subequations}
where   
\begin{equation}
 \hSchw[f'](x)
  \equiv  \red{S[f](x)} 
 :=
    \frac{f^{(3)}(x)}{f'(x)} - \frac{3}{2} \left( \frac{f''(x)}{f'(x)}\right)^2 
    \,,
    \label{10IX23.31}
  \end{equation}
 and where
 $S[f](x)$ denotes the Schwarzian derivative.

 Let us write the coordinate transformation \eqref{xprr} as
 \begin{align}
     t+\varphi &= \myf_+(\bar t + \bar \varphi) + O(\bar r^{-2})
     \,,
     &
     t-\varphi &= \myf_-(\bar t - \bar \varphi ) + O(\bar r^{-2})
     \,,
    \\
     r &= f (\bar t + \bar \varphi, \bar t - \bar \varphi, \bar r) \,,
 &
  \label{31VII24.12}
 \end{align}
Equivalently, and in more detail as needed later,
 \begin{align}
 &
     t = \frac{1}{2}
      \big(
       \myf_+  
       +\myf_-
       \big) 
       + \frac{(\bv^\prime_+)^2 \bv^{\prime \prime}_- + (\bv_-')^2 \bv_+''}{4 {\bar\rr}^2 \bv^{\prime }_- \bv^\prime_+}  
       + O(\bar r^{-4})
     \,,
     &
  \label{31VII24.13a}
 \\
     &
     \varphi =  \frac{1}{2}
      \big(
       \myf_+  
       -\myf_-
       \big)
       + \frac{(\bv^\prime_+)^2 \bv^{\prime \prime}_- - (\bv_-')^2 \bv_+''}{4 {\bar\rr}^2 \bv^{\prime }_- \bv^\prime_+}  
       + O(\bar r^{-4})
     \,,
     &
 \\
     &
     r = \frac{{\bar\rr}}{\sqrt{\bv_+'\bv_-' }} 
     + O(\bar r^{-1})
 \,.
  & 
  \label{31VII24.13c}
 \end{align} 
 %
 It holds that
 \begin{align}
 &
     (\partial_{\bar t} + \partial_{\bar \varphi}
     \big)t =  
       (\myf_+)^\prime    + O(\bar r^{-2})
     \,, 
     \qquad
     (\partial_{\bar t} - \partial_{\bar \varphi}
     \big)t =  
       (\myf_-)^\prime      + O(\bar r^{-2})
     \,, &
  \label{31VII24.14}
  \\    
  &
     (\partial_{\bar t} + \partial_{\bar \varphi}
     \big)\varphi= 
       (\myf_+)^\prime      + O(\bar r^{-2})
     \,, 
     \qquad
     (\partial_{\bar t} - \partial_{\bar \varphi}
     \big)\varphi  = 
     -
       (\myf_-)^\prime      + O(\bar r^{-2})
     \,. 
     &
  \label{31VII24.15}
 \end{align}
 This leads to the following transformation law of the mass aspect and of the angular momentum aspect: 
 \begin{eqnarray} 
   \mu  
   \mapsto 
    \bar \mu 
    &= & {\frac{1}{2}}
    \Big[ 
    (\mu - \myj)\big( (\partial_{\bar t} + \partial_{\bar \varphi})t\big)^2
    + 
    (\mu + \myj)\big( (\partial_{\bar t} - \partial_{\bar \varphi})t\big)^2
    \Big]
    \nonumber
\\
     && 
     \phantom{\frac12 \Big[} - \hSchw\big[(\partial_{\bar t} + \partial_{\bar \varphi})t\big]
      - \hSchw\big[(\partial_{\bar t} - \partial_{\bar \varphi})t\big]
 \label{eqmuchange1}
      \\
  &=& \mu  \left((\partial_{\bar \varphi} t)^2+(\partial_{\bar t} t)^2\right)-2 \myj\,
  \partial_{\bar \varphi} t
 \,\partial_{\bar t} t 
  - \hSchw\big[(\partial_{\bar t} + \partial_{\bar \varphi})t\big] \nonumber \\
     && \phantom{\frac12 \Big[} - \hSchw\big[(\partial_{\bar t} - \partial_{\bar \varphi})t\big]
 \label{eqmuchange2}
      \,,
      \\
  \myj
   \mapsto 
    \bar \myj
     &  = & 
     \red{\frac{1}{2}}
     \Big[ 
    (\mu+\myj)\big( (\partial_{\bar t} - \partial_{\bar \varphi})t\big)^2
    - (\mu-\myj)\big( (\partial_{\bar t} + \partial_{\bar \varphi})t\big)^2
    \Big] 
    \nonumber
 \\
     && 
      \phantom{\frac12 \Big[}
      - 
      \hSchw\big[(\partial_{\bar t} - \partial_{\bar \varphi})t\big]
     +
     \hSchw\big[(\partial_{\bar t}  +\partial_{\bar \varphi})t\big]
 \label{eqmuchange3}
     \\
     &=& 
     \myj \left( (\partial_{\bar \varphi} t)^2+(\partial_{\bar t} t)^2\right)- 2 \mu  \,
    \partial_{\bar \varphi} t \,
   \partial_{\bar t} t
      - \hSchw\big[(\partial_{\bar t} - \partial_{\bar \varphi})t\big]
       \nonumber \\ 
     && \phantom{\frac12 \Big[}+ 
     \hSchw\big[(\partial_{\bar t}  +\partial_{\bar \varphi})t\big]
     \label{eqmuchange4}
     \,.
 \end{eqnarray}    
 A comment is in order here. It follows from \eqrefl{31VII24.14}-\eqrefl{31VII24.15} that the limits, as $\bar r$ tends to infinity, 
 of $   (\partial_{\bar t} \pm \partial_{\bar \varphi}
     \big)t $ are functions of one variable only. The operator $\hat S$ appearing in the last set of equations and 
     which, by definition, acts on functions of one variable,
      has to be understood as acting on these limits. 
     To be precise: all occurrences of $\hSchw\big[(\partial_{\bar t} + \partial_{\bar \varphi})t\big]$ are functions of   $ (\bar t + \bar \varphi)$, and all occurrences of $\hSchw\big[(\partial_{\bar t} - \partial_{\bar \varphi})t\big]$ are functions of   $(\bar t -\bar \varphi)$. 
     
 We see that both aspect functions have a rather complicated transformation law, with $\mu\pm \myj = 
 4 \mathcal{L}_{\mp}$ 
 transforming independently:
 \begin{eqnarray}   
   \mu  + \myj 
   \mapsto 
    \bar \mu +\bar \myj 
    &= &  
    (\mu +\myj)\big( (\partial_{\bar t} - \partial_{\bar \varphi})t\big)^2
      - 2\hSchw\big[(\partial_{\bar t} -\partial_{\bar \varphi})t\big]  
      \,,
  \\    
   \mu  - \myj 
   \mapsto 
    \bar \mu -\bar \myj 
    &= &  
    (\mu -\myj)\big( (\partial_{\bar t} + \partial_{\bar \varphi})t\big)^2
      - 2\hSchw\big[(\partial_{\bar t} + \partial_{\bar \varphi})t\big]
     \,.
 \end{eqnarray}

It follows from the analysis in~\cite{ChWutte}, based on~\cite{Oblak,Balog},  that:
 
 \begin{enumerate}
   \item
   Both $\HamM $ and  each of the integrals
   \begin{equation}\label{15VIII24.1a}
     \HamM \pm J =\frac{1} {2\pi} \int_{S^1} (\mu\pm j) d\varphi
   \end{equation}
   can be made arbitrarily large using asymptotic symmetries.
   
   \item A sufficient, but not necessary, condition for $\HamM +J$ to be bounded from below is $\mu+j \ge -1$; similarly for $\HamM -J$.
   
   \item Each of the functions $\mu\pm j$ can be assigned a \emph{monodromy type}. The type allows one to find a canonical form of $\mu\pm j$ which is achievable by applying an asymptotic symmetry. A constant value of $\mu+ j$  can be attained if and only if $\mu+ j$ is either of hyperbolic type ${\mycal H}_0$, or of elliptic type ${\mycal E}_m$, 
 $m \in \mathbb{R}_{<0}$,
 or of  parabolic type  ${\mycal P}_0$; similarly for $\mu-j$. 
 
 \item  
 Suppose that $\mu+j$ is   of hyperbolic type $\mcH_0$, or of parabolic type $\mycal P_0$, 
 or of  elliptic type $\mcE_{m}$ with $m\ge -1$. We will refer to such functions as of \emph{good type}. 
For such functions the integral $\HamM+J$ is bounded from below, and a unique value of $\HamM+J$ can be obtained by minimizing over asymptotic symmetries. The minimum value is achieved by an asymptotic symmetry which renders $\mu+ j$ constant.

It follows from \eqref{barLs} that,  in the class of metrics \eqref{metrics} the value of  $\HamM+J$ can be changed while leaving $\HamM-J$ unchanged.

For the remaining types, namely $\mcH_n$,  $\mycal P^q_n$, with $n\in \N^*$, $q\in \{\pm 1\}$,
 and $\mcE_{m}$ with $m<-1$, which we will refer to as \emph{bad},   the integral $\HamM+J$ can achieve arbitrary values (both positive and negative) by applying an asymptotic symmetry~\cite{Balog} (compare~\cite{ChWutte}), 
 with one exception. Namely $\HamM+J$ is bounded from below by  {$-1$}  for the type $\mycal P^{-1}_1$,  can be made arbitrarily close to {$-1$} by a choice of asymptotic gauge, but {$-1$} is never attained by a choice of gauge. We will refer to the  $\mycal P^{-1}_1$ type as \emph{marginal}.

 Similarly there exist ``bad'' functions $\mu-j$ such that the integral $\HamM-J$ can achieve arbitrary values by applying an  asymptotic symmetry, while leaving $\HamM+J$ unchanged.

   \item Let $(M,g,K=0)$ be obtained by a Maskit-gluing, reviewed in Section~\ref{s11VIII24.14} below, of two ALH time symmetric initial data sets. In~\cite{ChWutte}
   explicit formulae have been given for the aspect function $\mu$ of the glued metrics,   and the type of $\mu$, 
   in terms of the original ones. 
 \end{enumerate} 
  
  Since
 \begin{equation}\label{30VII24.2}
   \HamM = \frac{1}{2}(\HamM+J) +  \frac{1}{2}(\HamM-J)
   \,,
 \end{equation}
 with each part transforming independently,
whenever either of the functions $\mu\pm j$ is in the ``bad'' class, 
 $\HamM$ can achieve arbitrary values. 
 
 So,   a meaningful value of mass 
  can only be  
   assigned to a spacetime 
   when both functions $\mu\pm j$ belong to a class where the associated integrals $\HamM\pm J$ are bounded from below. In such cases a unique number,
  larger than or equal to $-1$, can be assigned by minimization to each of $\HamM\pm J$.

\input{RemAS}

Let $\mlambda \in \R$. We say that a boundary $\partial M$ of $M$ is \emph{weakly future $\mlambda$-trapped} if 
the null expansion scalar $\theta$  of  $\partial M$, defined with respect to a null normal which points to the future and towards $M$, satisfies 
$$
 \theta \le \mlambda 
 \,.
$$
We will say that $\partial M$ is weakly $\mlambda$-trapped if a time orientation can be chosen so that it is weakly future $\mlambda$-trapped.
  
In Section~\ref{ss19IX24} we will show: 
 
 \begin{theorem}
  \label{T15VIII24.1}
  Let  $(M,g)$ be a smooth, complete Riemannian manifold, 
  possibly with a compact weakly $1$-trapped boundary, 
  and
   suppose that $(M,g,K)$ satisfies the dominant energy condition. 
  If $(M,g,K)$ contains a locally asymptotically hyperbolic end, then   both $\mu+j$ and $\mu-j$ are of good 
  or marginal
   type 
   and 
  \begin{equation}
   \uHamM +1 
     \ge | \underline J| + \sqrt{|\vec { \underline C}|^2 
  + |\vec{  \underline H}|^2 
  + 2 |\star (\vec{  \underline H}\wedge \vec 
  {\underline C}) 
    }|
   \, ,
   \label{15VIII24.2dr}
   \end{equation}
  where $\vec{\underline C}$ and $\vec{\underline H}$ 
     are defined in Section~\ref{30VII24} below. 
  \end{theorem}
  
Equality in \eqref{15VIII24.2dr} is achieved on hyperbolic space. 

The inequality \eqref{15VIII24.2dr} is complemented by a second inequality
 if $M$ has a weakly $1$-trapped compact boundary, or if more is known about the geometry of $M$:

\begin{theorem}
 \label{T30X24.1}
Under the remaining conditions of Theorem~\ref{T15VIII24.1}, suppose  that  $(M,g)$ has either a non-empty 
weakly $1$-trapped 
compact boundary or at least two ALH ends. Then, 
in addition to \eqref{15VIII24.2dr} we have
  \begin{equation}
   \uHamM \ge |\underline J|
   \,.
   \label{15VIII24.2drbd}
   \end{equation}
\end{theorem}  
 
The equality in \eqref{15VIII24.2drbd} is attained on extremal BTZ black holes. 
 
\begin{remark}
 \label{R13XI24}
From what has been said it follows that if $\underline H>-1$, then there exists an asymptotic symmetry which preserves the initial data surface and renders the mass aspect function constant, equal to $\underline H$. 

Both the angular momentum aspect and the mass aspect can  be rendered constant if the associated spacetime has a sufficiently large conformal completion at infinity, as making both these functions constant simultaneously typically requires passing to another hypersurface in spacetime.
\qed
\end{remark}  
  
In the presence of
boundaries  
with $\theta =0$  
one expects  a Penrose-type inequality
 \begin{eqnarray}
 \uHamM \ge 
    \big(\frac{\ell}{2\pi}\big)^2
   \,,
    \label{25XII23.3px}
 \end{eqnarray}
 where $\ell$ is the length of $\partial M$, 
 with the lower bound attained for all BTZ black holes. 
 Some evidence towards this has been given in~\cite{ChWutte}.

%% file: mass.tex
We define the total Hamiltonian mass $\HamM $, respectively the total angular momentum $J$, as the Hamiltonian associated with  time translations, respectively with rotations. One has (see~\cite{ChWutte} 
and references therein for an analysis of $\HamM $) 
\begin{eqnarray}
  \HamM  &:=& H(\red{\mathring V=r},0)
  =
   \frac{1 }{2\pi}
    \int_{S^1} (\mu_{22}+2 \mu_{11} ) d\varphi
  \,,
    \label{7IX24.12}
\\
  J &:=&
   H(0,\red{\partial_\varphi})=  \lim_{R\to\infty}
   \frac{1 }{\pi}
    \int_{r=R} P^r{}_\varphi \sqrt{\det g} \,d\varphi
   \,.
    \label{7IX24.11}
\end{eqnarray}

%% file: BTZmetricCauchy.tex
As an example, in the case of   BTZ black holes \eqref{BTZmetric} the spacetime metric can be written as
\input{BTZaddOn}
This gives 
\begin{align}
  g 
   &
   =
  \
  {\bar \theta}^1 {\bar \theta}^1
 +\left(1 + \frac{M} {\myr^{2}}  + \frac{M^2-J^2/4}{\myr^{4}}
 + O(\myr^{-6})
  \right) {\bar \theta}^2 {\bar \theta}^2
  \,, 
\\
  K 
  &
  =
  \ - \left( \frac{J}{\myr^2} + O(\myr^{-4})
  \right)
  {\bar \theta}^1 {\bar \theta}^2 
  \,,
  \label{1811024}
\end{align} 
where
\begin{equation}
  {\bar \theta}^2 = \frac{d \myr}{\myr}\,,  
  \quad
  {\bar \theta}^1 = \myr d \varphi\,, 
  \quad
  {\bar \theta}^0 = \myr d t
   \,.
   \label{181024}
\end{equation} 
Evaluating the integral  \eqref{7IX24.12} one finds $\HamM=M$, 
with the angular-momentum integral \eqref{7IX24.11} 
reproducing the parameter $J$ appearing in \eqref{difffourg}.

%% file: BTZaddOn.tex
\begin{equation}
  \label{difffourg}
  {\fourg }  = -\left(\myr^2 - M\right) dt^2+ \frac{4 \myr^2 d\myr^2}{J^2-4 M \myr^2+4 \myr^4} - J dt d\varphi
  +\myr^2 d\varphi^2\,,
\end{equation}
for some constants $M$ and $J$,
where the radial coordinates  in \eqref{BTZmetric} and \eqref{difffourg} are
related by 
\begin{equation}
  \label{rtransform}
  r =\sqrt{ \frac{\sqrt{\frac{J^2}{4}-M
  \myr^2+ \myr^4}-\frac{M}{2}+ \myr^2}{\sqrt{2}}}\,
\end{equation}
outside the outer horizon, i.e.\ $\myr> \myr_+$, where
%
$\myr_\pm = \sqrt{\frac{1}{2} (M \pm \sqrt{M^2-J^2})}$.  

%% file: RemAS.tex
Alternatively, one may want to assign  mass and angular momentum to a given Cauchy   surface. 
In such a case 
only  those  asymptotic symmetries
 which preserve the initial data hypersurface are relevant. Assuming that this hypersurface is given by $\{t=0\}$,
a necessary and sufficient condition for this is 
(cf.\ \eqref{31VII24.13a} with  the error-terms  set to zero  there)
\begin{equation}\label{3X24.1}
    \myf_+( \bar \varphi)  = 
       - \myf_-(- \bar \varphi) 
        \,.
\end{equation}
Under such transformations  the mass aspect $\mu$ transforms as
\begin{equation}\label{3X24.2}
    \mu  
    \mapsto 
     \bar \mu (\bar \varphi)
     =  (\mu\circ \myf_+)
      ( \bar \varphi) 
      \,
       \big(\myf_+^\prime( \bar \varphi) \big)^2 
     - 2 \hSchw\big[\myf_+^\prime( \bar \varphi)\big]   
  \,.
\end{equation}
We can associate a number $\uHamM\in \R\cup\{-\infty\}$ to the section $\{t=0\}$ of the conformal boundary at infinity by setting 
\begin{equation}\label{26X24.1}
  \uHamM= \inf \HamM
  \,,
\end{equation}
where the infimum is taken over asymptotic symmetries satisfying \eqref{3X24.1}.
 
For the angular momentum aspect we find
\begin{equation}\label{3X24.2j}
    j
    \mapsto 
     \bar j (\bar \varphi)
     =  
     (j\circ \myf_+)
      ( \bar \varphi) 
      \,
       \big(\myf_+^\prime( \bar \varphi) \big)^2 
  \,.
\end{equation}
Hence  
\begin{eqnarray} 
 \bar  J &:=&   \lim_{R\to\infty}
   \frac{1 }{\pi}
    \int_{\bar r=R} P^{\bar r}{}_{\bar \varphi } \sqrt{\det g} \,d\bar \varphi
     \nonumber
\\
 & = & 
  \frac{1}{2 \pi}    
  \int_{S^1} 
  j \big(
   \myf_+( \bar \varphi) \big)
    \,\big(\myf_+^\prime( \bar \varphi) \big)^2  \,d\bar \varphi   
    \nonumber
\\
 & = & 
  \frac{1}{2 \pi}    
  \int_{S^1} 
  j(  \varphi)  \,\big( \myf_+^\prime 
  \circ \myf_+^{-1}
  \big)(\varphi)   \,d  \varphi
    \nonumber 
  \\
 & = &  
  \frac{1}{2 \pi}    
  \int_{S^1} 
   \,\frac{j(  \varphi) }{\big(\myf_+^{-1}
  \big)^\prime (\varphi)}   \,d  \varphi
  \,.
    \label{7IX24.11x}
\end{eqnarray}
\ptcn{wrong stuff went to WrongJ}
We conclude that a choice of a section of the conformal boundary does not resolve the ambiguities in the definition of $J$ either.

Note that on a given Cauchy surface both $\mu+j$ and $\mu-j$ transform in the same way as $\mu$;
\begin{equation}\label{3X24.2dw}
    \mu   \pm j 
    \mapsto 
    ( \bar \mu \pm \bar j)  (\bar \varphi)
     = 
     \big ( ( \bar \mu \pm \bar j) \circ \myf_+
      \big)
      ( \bar \varphi) 
      \,
       \big(\myf_+^\prime( \bar \varphi) \big)^2 
     - 2 \hSchw\big[\myf_+^\prime( \bar \varphi)\big]   
  \,.
\end{equation}

Now, if a mass $\uHamM>-\infty $ can be defined by minimizing $\HamM$,  
and a gauge minimizing $\uHamM$ exists,  
we set
$$
 \underline J =J 
 \,,
$$ 
where the integral defining
$J$ is calculated in the gauge which minimizes $\HamM$. 

When the mass aspect function $\mu$ is marginal, in the context of  Theorem~\ref{T15VIII24.1}  we set
\begin{equation}\label{1XI24.2}
  \underline J =0
  \,.
\end{equation}
The reason for this will be made clear below. 
 
In fact, we conjecture that a marginal mass aspect cannot occur under the hypotheses 
of the theorems that follow, but we have not been able to exclude this possibility. 

For definiteness, we will say that an initial data set $(M,g,K)$ is \emph{locally asymptotically hyperbolic} (ALH) if 
the fields $(g,K)$ have asymptotic expansions in the spirit of in \eqref{bansolutionsONB}-\eqrefl{2VIII24.31}, with $M$ viewed as the hypersurface   $\{t=0\}$:
\begin{align}
    g  
  & = 
  \, 
    \mathring g + \big(
    \mu_{ij} 
     r^{-2}+ o ( {r^{-2}}) \big)
     \theta^i \theta^j 
      \label{bansolutionsONBt}\,,
       \\
    K 
    &
    =
    \, - j  r^{-2} 
    \theta^1 \theta^2  
     +o(r^{-2}) 
    \theta^i \theta^j
    \,,
    \label{2VIII24.31t}
\end{align}
and with the error terms behaving in the obvious way under differentiation.

%% file: SmallData.tex
\renewcommand{\bcheckpsi}{{\check{\psi}}}
\renewcommand{\mcL}{{\mycal{L}}}
\renewcommand{\hhat}{{\hat{h}}}
\subsection{Small perturbations of hyperbolic space}
 \label{30VII24}

\input{begSmallData} 

When $R\ge \bar R$ the Hamiltonian $\Hh\equiv H(V^0,0)$ is manifestly positive if the \emph{harmonicity vector}
 $ \bcheckpsi$ can be made to vanish (which is a gauge choice) or made small enough, and if the tensor field $h$ is sufficiently small
  in a  weighted $C^1$ topology, with a weight chosen so that the integrals involving the error terms converge and remain small.
It has been pointed out in~\cite{BCN}  that the gauge-choice  $ \bcheckpsi$
 can always be made when  the requirements of a well-defined $H^\mu$ are met and when  $n\ge3$. So, to verify positivity of $\Hh$, for $h$ small enough, when $n=2$ it remains to check the vanishing of $ \bcheckpsi$ in this dimension.

Suppose thus that 
$$ 
  |h| + | \bcov h| 
    \le \epsilon
\,,
$$
for some small $\epsilon>0$. Then
\begin{align}\label{30VIII24.1bc}
 \bcheckpsi^i=  &
 \ 
  \psi^i + \frac{1}{2} g^{i k} \bcov_k \phi 
  \nonumber
   \\
    = & 
    \ 
    \bcov{_ \red{j}}  g^{ij} 
    +\frac{1}{2} g^{i k} \bcov_k (\bar g^{j\ell  }  {h}_{j \ell} ) + O (|h||Dh|)
    \nonumber
   \\
    = & 
    \ 
   -  \bcov{^j}  h^{\red{i}}{}_{j} + O (|h||Dh|)
    +\frac{1}{2} \bar g^{i k} \bcov_k (\bar g^{j\ell  }  {h}_{j \ell} ) + O (|h||Dh|)
    \,.
\end{align}
Consider  a  vector field $X$ satisfying
\begin{equation}\label{1IX24.1}
  |X| + |\bcov X|  + |\bcov \bcov X| \le \epsilon
  \,.
\end{equation}
Under the time-one-flow of $X$ 
 the tensor field $g_{ij}$ transforms as 
\begin{equation}\label{30VIII24.2cd}
  g_{ij}=\bar g_{ij}+h_{ij} \mapsto \bar g_{ij}+h_{ij} + \mcL_X \bar g_{ij} 
  + O (\epsilon^2)
  \,.
\end{equation}
Applying the associated diffeomorphism to the metric redefines $g_{ij}-\bar g_{ij}$, with the leading order contribution of $X$  to \eqref{30VIII24.1bc} arising from the terms 
\begin{equation}\label{30VIII24.5}
   -  \bcov{^j}  
   (h^\red{i}{}_{j} + \mcL_X \bar g^\red{i}{}_{j})
    +\frac{1}{2} \bar g^{i k} \bar g^{j\ell  } \bcov_k 
    (  {h}_{j \ell} + \mcL_X \bar g_{j\ell} ) 
    \,.
\end{equation}
Choosing  $X$ as the solution of the equation
\begin{equation}\label{30VIII24.60}
   -  \bcov{^j}  (\bcov_i X_j + \bcov_j X_i)  
    +  \bcov_i \bcov^j X_j   = -\bcheckpsi_i
\end{equation}
will lead to a gauge-transformed field $\bcheckpsi$ satisfying  
\begin{equation}\label{30VIII24.6}
   \bcheckpsi^i
    = O(\epsilon^2)
    \,,
\end{equation}
which is good enough for positivity of $\Hh$ when  $\epsilon$ is small enough, provided that the relevant terms in \eqref{26IV18.8} decay fast enough so that their order, namely $O(\epsilon^4)$,  is not affected by integration.

Commuting derivatives, \eqref{30VIII24.60} can be rewritten as
\begin{equation}\label{30VIII24.7}
LX^i:=   \Delta_{\backg} X^i  +\underbrace{ 
   \bar R^i{}_j X^j  }_{=-(n-1)X^i} =  \bcheckpsi^i
    \,.
\end{equation}
Since the Laplacian is negative, one readily checks that this equation has unique solutions decaying to infinity for all sources decaying sufficiently fast.  
In fact, from what has been said together with the implicit function theorem,   we find that the gauge $
  \bcheckpsi^i=0$ can be achieved whatever $n\ge 2$ by a suitable coordinate transformation when $\epsilon$ is small enough.

  In order to analyze the relevant decay thresholds, and the 
asymptotic behavior of the solutions,  
let us denote by $\wasxone$ the defining coordinate for the conformal boundary at infinity.
Near $\{\wasxone=0\}$  
we use Fefferman-Graham coordinates, so that $\backg$ takes there the form 
\begin{equation}\label{2IX24.2}
    \backg= \frac{1}{{(\wasxone )}^2} \big( \underbrace{(d{\wasxone })^2 + \exp(2 \fone (\wasxone )) \hbound 
    }_{=:\hhat}
      \big) {=:\exp(2 \myomega) \hhat}
    \,,
\end{equation}
where  
for definiteness we assume that $\phi$ is a smooth function of its
arguments near $\{\wasxone=0\}$, and where 
$\hbound $ is a metric on the conformal boundary that we 
momentarily leave unspecified; if $\hbound$ is the round metric on $S^{n-1}$ then, for
$0\le \wasxone  <1$,
\begin{equation}\label{7IX24.1}
  \fone (\wasxone ) = 
    \mathrm{log}(\frac{1- ({\wasxone })^2}{2})
     \,.
\end{equation}
Hence
%
\begin{align}\label{31VIII24.5} 
\Delta_\backg X^{i} &= e^{-2 \myomega}\Big(\Delta_\hhat X^i + X^i \Delta_\hhat \myomega - (\nabla^i_\hhat \myomega)(2 \nabla_j^\hhat X^j
+(n-2) X^j \partial_j \myomega) \nonumber \\  
&~~~+(\partial_j \myomega) (2 \nabla^i_\hhat X^j + n \nabla^j_\hhat X^i + (n-2) X^i \partial^j \myomega) \Big)\\ 
&= (\wasxone )^2 \Delta_\hhat X^i 
+ X^i \big(1 + \frac{2 (n-1) (\wasxone )^2}{1- (\wasxone )^2}\big)
+{\hat h^{i 1}}(\wasxone  2 \nabla_j^\hhat X^j
-(n-2) X^1 ) \nonumber \\
&~~~ - \wasxone  (2 \nabla^i_\hhat X^1 + n \nabla^1_\hhat X^i) \red{+} (n-2) X^i 
\,,
\end{align}
where the second equality assumes \eqref{7IX24.1}.
The indicial exponents of $L$  can be found by setting $X^i =(\wasxone )^\sigma A^i $, for a vector $A^i$ with constant entries and a number $\sigma\in \R$, 
and finding the zeros of the determinant of the linear map $A^i\mapsto W^i[A](\sigma)\equiv W^i{}_j(\sigma) A^j$, determined from the equation 
\begin{equation}\label{1IX24.2}
L((\wasxone )^\sigma A^i )\equiv   \Delta_{\backg} ((\wasxone )^\sigma A^i )  
 -(n-1)(\wasxone )^\sigma A^i  = : (\wasxone )^\sigma W^i[A](\sigma)  
+  O((\wasxone )^{\sigma+1})
\,.
\end{equation} 
One finds 
that 
$W$ is diagonal, with $W^1{}_1 (\sigma)= \sigma^2 - (n+1)\sigma + 2-n$, and with the remaining eigenvalues
$\sigma^2 -(n+1)\sigma$. 
This leads to the characteristic exponents 
%
$$
 \Big\{ 0, n+1, \frac{1}{2} 
 \big(n+1 \pm \sqrt{(n+1)^2 + 4 (n-2)}
 \big)
 \Big\}
=|_{n=2}
 \{0,3\}
 \,.
$$
For an $n$-dimensional metric of the form 
\begin{align}
    g 
    = & 
  \ 
    \backg + \big((\wasxone )^{n} \bmu_{ij}  + O ((\wasxone )^{n+1})\big) \btheta^i \btheta^j
      \label{2IX24.5rd}\,,  
\quad \mbox{with} \  \btheta^i:= \frac{dx^i}{\wasxone }
\,,
\end{align}
where the $\bmu_{ij}$ 's are $\wasxone $-independent, 
the harmonicity vector   $ \bcheckpsi$ is $O((\wasxone )^n)$ in norm. 
More precisely, in dimensions $n\ge 2$, 
and continuing to use
the coordinate system of \eqref{2IX24.2}, we have
%
\begin{align}\label{2IX24.1}
 & 
  \bcov_i g^{i1} = - \bmu^j{}_{j} (\wasxone )^{n+1} + O((\wasxone )^{n+2})
\,,
\quad
  \bcov_i g^{ip} = O((\wasxone )^{n+2})
\,,
 &
 \nonumber
\\
& 
 \backg^{ik} \bcov_k (\backg^{j\ell} h_{j\ell})
=  n\bmu^j{}_{j} \delta^i_1 (\wasxone )^{n+1} + O((\wasxone )^{n+2})
\,,
 & 
\end{align}
where $\bmu^j{}_{j} = \delta^{k j} \bmu_{k j}$,
so that
\begin{equation}\label{2IX24.3}
  \bcheckpsi
 = 
 \Big( \frac{n -2}{2} (\wasxone )^{n+1}\bmu^j{}_{j} + O ((\wasxone )^{n+2})
 \Big) \red{\partial_n} + 
O ((\wasxone )^{n+2})
 \partial_{\red a}
 \,.
\end{equation}
For   $n>2$ we are in a resonant case, 
where the smallest strictly positive indicial exponent coincides with the decay rate of the source term. 
This leads to logarithmic terms in the transformed metric, which do not change the mass when $n\ge 3$, so that the positivity calculation applies. 
However, such terms would create annoying (though perhaps not essential) problems when $n=2$, as the  total mass is already ill-defined without logarithmic terms in the metric. 

As such, given that the offending, dominant term in $\wasxone $, vanishes when $n=2$, in this dimension it follows 
from~\cite{AndChDiss,Lee:fredholm} 
that there exists a  vector field   $\mathring X$, globally defined by the data at hand, satisfying  
$$
 \partial_n \mathring X=0 
$$
such that
the (unique) solution of \eqref{30VIII24.7} takes the form 
\begin{equation}\label{2IX24.6}
  X^i   =  (\wasxone )^{3} \mathring X^i+ O\big((\wasxone )^{4}  \big)
 \,.
\end{equation}
This leads to a new transformed metric
\begin{equation}\label{4IX24.1}
  g = \backg +
    \big(
     \underbrace{
    \bmu_{ij} -2 \mathring X^1 \delta_{ij} 
     + 3 \mathring X_i  \delta_j + 3 \mathring X_j  \delta_i)
    }_{=: \overline{\bmu}_{ij}}
      + O(\wasxone ) 
       \big)
        \btheta^i \btheta^j
  \,,
\end{equation}
with new mass aspect tensor $\overline{\bmu}_{ij}$. Interestingly enough, this does not change the 
mass aspect function $\bmu$, 
$$
 \bmu = \bmu_{11} + 2 \bmu_{22} \mapsto \overline{\bmu} =
  \underbrace{
   \overline{\bmu}_{11} 
   }_{\bmu_{11} + 4 \mathring X^1} 
    +  2   \underbrace{
   \overline{\bmu}_{22} 
   }_{\bmu_{22} -2 \mathring X^1}  = \bmu
 \,,
 $$
 and hence (cf.\ \eqref{9IX24.11}) the total energy-momentum vector $H^\mu$.  
 
 \begin{remark}
 \label{R2XII24.1}
 One could wonder, how is it possible that the left-hand side 
 of \eqref{26IV18.8}
depends upon the asymptotic gauge, since the right-hand side is uniquely defined by the background and by the metric. The answer is that every choice of asymptotic gauge defines a different background. 
 \qedskip
 \end{remark}
  
The above considerations lead to:

  \begin{theorem}
    \label{t8IX24.1}
    Let $(M,\backg)$ denote the hyperbolic space. 
Consider a metric $g$ of the form
$$
g=\backg+h
\,,
\ 
\mbox{with}
\ 
  |h| + | \bcov h|
    \le \epsilon
\,,
$$
with well-defined total energy-momentum, and  with
$$
 R(g) \ge -2
 \,.
$$
There exists $\epsilon_0>0$ such that for all  $0< \epsilon\le \epsilon_0 $
the energy-momentum vector is timelike future pointing, 
  \begin{equation}\label{8IX24.61}
    H^0 > |\vec H|  
    \,,
  \end{equation}
except if $g=\backg$.  
  \end{theorem}

 {\noindent\sc Proof:}   
As already emphasized, we can  transform the metric to the gauge $  \bcheckpsi^i=0$ without changing the mass aspect function. 
 The result follows then from \eqref{26IV18.8} as in~\cite{BCN}.
 \hfill
 $\Box$
 \bigskip

 \input{Proof2}

 %

%% file: begSmallData.tex
In this section we shall prove Theorem~\ref{T15VIII24.1} for a large class of small perturbations of hyperbolic space with scalar curvature bounded from below by $-2$. The result applies to general relativistic data sets satisfying the dominant energy condition on maximal surfaces,  $\tr_g K=0$.  
From the perspective of Theorem~\ref{T15VIII24.1} the calculations here are made obsolete by the arguments in Section~\ref{ss19IX24}, but we believe they have some interest of their own.  

A comment on the choice of normalization of the Hamiltonian mass is in order. In \eqref{15VIII24.2dr} the zero-point of the mass has been chosen so that the mass is positive for all BTZ black holes, which leads to $\HamM=-1$ for hyperbolic space. 
This normalization is not natural for perturbations of hyperbolic space, where further Hamiltonian charges are present.
Indeed, recall
that Hamiltonian charges are associated with Killing vector fields of the Lorentzian background. For  
static
ALH metrics \emph{other than Anti-de Sitter spacetime}, 
 the space of such Killing vectors splits into a \emph{one-dimensional} space of vector fields normal to the time-symmetric slices, say $\hyp_t$, and into tangential Killing vector fields, if any. On the other hand, the $(n+1)$-dimensional Anti-de Sitter spacetime itself 
 has an $(n+1)$-dimensional space of Killing vectors normal to $\hyp_t$. In the coordinate system in which the hyperbolic metric $\backg$ takes the form
  \ptcn{ $\mzg$ is used for the cusp, cf.  
  \eqref{backgroundban}, \eqref{bansolutionsONB}}
\begin{equation}\label{7IX24.5}
  \backg = \frac{dr^2}{1+r^2} + r^2 d\Omega^2
  \,,
\end{equation}
where $d\Omega^2$ is the round metric on $S^{n-1}$, the space of $\backg$-normal Killing vectors of $\backg$ is spanned by vector fields whose normal components are 
\begin{equation}\label{7IX24.6}
 V^0 = \sqrt{1+r^2}
 \,,
  \quad
  V^i = x^i
  \,.
\end{equation}
One sets
\begin{equation}\label{7IX24.7}
  H^\mu := H(V^\mu,0)
  \,,
\end{equation}  
where $H$ is given by \eqref{17VIII24.1}.  Under the usual decay conditions for a well-defined mass, under isometries of hyperbolic space the ``three-vector'' 
$$
(H^\mu) \equiv (H^0, \vec H)
$$
 transforms as a Lorentz vector. 
\input{KIDs.tex}%

A natural class of metrics associated with the background \eqref{7IX24.5} is provided by
the ALH  metrics \eqref{bansolutionsONBt}, which can be rewritten in 
 the form 
\begin{equation}\label{7IX24.5we}
 g  =  \backg + \Big(
  \frac{\bmu_{ij}}{r^n} + o(r^{-n})
  \Big)  \btheta^i \btheta^j
  \,,
   \quad \btheta^n = \frac{dr}{\sqrt{1+r^2}}
   \,,
   \
   \btheta^a =r \psi^a
  \,,
\end{equation}
where $\{\psi^a\}_{a=1}^{n-1}$ is an orthonormal frame for the metric $d\Omega^2$, 
with $\partial_r\bmu_{ij}=0$.
 In space-dimension $n=2$ a calculation gives 
\begin{eqnarray} 
  \Hh&=&   
   \frac{1 }{2\pi}
    \int_{S^1} (
    \underbrace{
    \bmu_{22}+2 \bmu_{11} 
    }_{=:\bmu}) 
    d\varphi
  \,, 
    \label{9IX24.11} 
\\
  H^1 &=&
   \frac{1 }{2\pi}
  \int_{S^1} \cos(\varphi)
   \,
  \bmu \,d\varphi 
  \,, 
    \label{9IX24.12} 
\\
  H^2 &=&
    \frac{1 }{2\pi}
  \int_{S^1} \sin(\varphi)
   \,
  \bmu \,d\varphi 
  \,,
    \label{9IX24.13} 
\end{eqnarray}
and for 
ALH 
initial data as in \eqref{bansolutionsONBt}-\eqref{2VIII24.31t}
we further have 
\begin{align}
C^1 &= - \frac{1 }{2\pi}
\int_{S^1} \sin(\varphi)
 \,
j \,d\varphi 
\,,  \\
C^2 &=  \frac{1 }{2\pi} \int_{S^1} \cos(\varphi)
\,
j \,d\varphi 
\,.  
\end{align}
\input{TrafoMassAspect}

Let us  pass now  to our main argument. In~\cite{BCN} the following identity has been derived for small perturbations of the hyperbolic space metric $\backg$, in space-dimension $n$: 
\begin{eqnarray}
 H(V,0)
 &=&
 \int_M \
  \Big[
  R - \overline R
  +
   \frac{n+2}{8n}
   |\bcov \, \phi|^2_{\backg}
   +
   \frac{n-2}{4n}
   |\bcov \hat{h}|^2_{\backg}
    +
    \frac{ n^2-4}{8n}
    \phi^2
    \nn
    \\
    &&
    +  \frac 1 {2n}
    \big(
     |\overline{{\fancyD}} \hat h|_{\backg{}}^2
    +
      |\overline{\mathrm{div}}\, \hat h - \hat h_{dV}|_{\backg{}}^2
      \big)
    \Big]
     V \sqrtbg
      \phantom{xxx}
\nn
  \\
    &&
     -
    \int_M
    \Big(
    \gauge{
    \frac{1}{2}
    \big(
       |\bcheckpsi|^2_{\backg}
       -
       \bcheckpsi^i\bcov_i \phi )
       V
       \rsqrtbg
 -
  \big(
    \tensor{h}{^k_i} \bcheckpsi^i
    +
    \frac 12
    \phi\bcheckpsi^k
   \big)
   \bcov_k V
       }
\nn
\\
&&
       +
     \herr V
    +
    O(|h|_{\backg} |\bcov h|^2_{\backg})\, V
          +
    \hdherr  |\bcov V|_{\backg}
     \Big) \sqrtbg
    \,,
     \phantom{xxxx}
     \label{26IV18.8}
\end{eqnarray}
where
\begin{eqnarray}
&
 h_{ij}:= g_{ij} - \backg_{ij}
 \,,
 &
  \label{28II18.1-}
\\
 &
 \psi^j:=\bcov{_ i}  g^{ij}
 \quad
 \Longleftrightarrow
 \quad
 g^{ij} \bcov_i h_{j\ell} =-g_{\ell j}  \psi^j
 \,,
  \label{28II18.1}
  &
\\
 &
 \phi := g^{ij} h_{ij}
 \,,
 \quad 
  \overline \phi := \backg^{ij} h_{ij} = \phi +  O\left(|h |^2_{\backg} \right)
  \,,
  &
 \\
 &
  \bcheckpsi^i:=   \psi^i + \frac{1}{2} g^{i k} \bcov_k \phi 
 \,,
  &
  \\
  & 
 \check{h}_{ij}
 :=
 h_{ij} - \frac{1}{n}  {\phi} \, g_{ij}
 \,,
 \qquad
 \hat{h}_{ij}
 :=
 h_{ij} - \frac{1}{n} \overline{\phi} \, \backg_{ij}
  \,.
\end{eqnarray}
Actually in~\cite{BCN} it was assumed that $n\ge3$ but one readily checks that the calculations also apply with $n=2$. 

%% file: KIDs.tex
The remaining KIDs are linear combinations of the fields (cf.~Equation (3.6) in \cite{CMT} for the formulae
in the Poincar\'e ball model)
%
\begin{equation}
 \Omega_{k i} = x^k \partial_i - x^i \partial_k
 \,,
\end{equation}
where $1 \leq k < i \leq n$,
and 
\begin{equation}
    \mathcal{C}_k = \sqrt{1+r^2 }\partial_k\,.
    \label{16X24.1}
\end{equation}
The charge associated to $\mathcal{C}_k$ will be referred to as \emph{center of mass}
\begin{equation}
    \label{51024}
\vec C := H(0, \vec{\mathcal{C}})\,,
\end{equation}
whereas the charge associated to $\Omega_{k i}$
 will be referred to as  \emph{angular momentum}
\begin{equation}
\mathcal{J}_{k i } := H(0, \Omega_{k i}) = 
\begin{pmatrix}
    0 & J_i\\
   - J_i & 0 
    \end{pmatrix}\,.
\end{equation}

%% file: TrafoMassAspect.tex
\begin{remark}
 \label{R8IX24.1}
Still in dimension $n=2$, consider metrics of the form
\begin{align}
    g 
    = & 
  \ 
    \hg + \big(\hx^{2} \hat \mu_{ij}  + o (\hx^{2})\big) \htheta^i \htheta^j
      \,,  
\quad \mbox{with} \  
\htheta^1:= \frac{d\hy}{\hx}\,, \ 
\htheta^2:= \frac{d\hx}{\hx}
\,, 
      \label{2IX24.5}
\end{align}
with $\partial_{\hx}\hat \mu_{ij}=0$, where
\begin{equation}\label{10V24.3notapp}
	\hg
	=
	\frac{d\hx^2 + d\hy^2}{\hx^2}
	\,.
\end{equation}
For such metrics a \emph{mass aspect function} $\hat \mu$ can be defined as in \eqref{8IX24.1} or in \eqref{9IX24.11}:
\begin{equation}
  \hmu := \hmu_{22}+2 \hmu_{11} 
   \,.
   \label{8IX24.2}
\end{equation}
There is a simple relation  between $\hmu$ and $\bar{\mu}$. Indeed, 
the metrics \eqref{7IX24.5we} can be rewritten in the form \eqrefl{2IX24.5} by setting $\hx=1/r$, $\hy =\varphi$,  which gives
\begin{eqnarray}
  \backg 
   & = &
    \frac{dr^2}{1+r^2} + r^2 d\varphi^2 +  \Big(
  \frac{\bmu_{ij}}{r^2} + o(r^{-2})
  \Big)  \btheta^i \btheta^j
   \nonumber
\\   
  & = &
   \Big(
    1- \frac{1}{r^2} + O(r^{-4}
    \big)) \frac{dr^2}{r^2} +r^{-2} d\varphi^2 +  \Big(
  \frac{\bmu_{ij}}{r^2} + o(r^{-2})
  \Big)  \btheta^i \btheta^j
   \nonumber
\\   
  & = & \hg
  + 
  \big(  \bmu_{22} -1 + o(1)
  \big)
     \htheta^2\htheta^2
 + \sum_{ij\ne 22} 
 \big( \bmu_{ij} +   o(1)
 \big)
  \htheta^i \htheta^j
  \,. 
 \label{7IX24.5a}
\end{eqnarray}
It follows that $\hmu= \bmu - 1$ and 
\begin{equation}\label{7IX24.2}
  \Hh= \HamM+1
  \,,
\end{equation}
with $\HamM$ defined as in \eqref{7IX24.15} -\eqref{7IX24.12}. 
Note that 
\begin{eqnarray}  
  H^1 &=&
   \frac{1 }{2\pi}
  \int_{S^1} \cos(\varphi)
   \,
  \bmu \,d\varphi
  =
   \frac{1 }{2\pi}
  \int_{S^1} \cos(\varphi)
   \,
  \mu \,d\varphi
  \,, 
    \label{9IX24.12we} 
\\
  H^2 &=&
    \frac{1 }{2\pi}
  \int_{S^1} \sin(\varphi)
   \,
  \bmu \,d\varphi
  =
    \frac{1 }{2\pi}
  \int_{S^1} \sin(\varphi)
   \,
 \mu \,d\varphi
  \,,
    \label{9IX24.13we} 
\end{eqnarray}
where $\mu=\mu_{22} + 2 \mu_{11}$, with $\mu_{ij}$   as in \eqref{bansolutionsONBt}.
\qed
\end{remark}

\begin{remark}
 \label{R8IX24.1er}
In Remark~\ref{R8IX24.1} we assumed implicitly  that $\varphi$ is a coordinate on $S^1$, 
hence defined mod $2\pi$, and so would therefore be the new coordinate $\hy$. But one can instead map directly the metric \eqref{7IX24.5}
into the form \eqrefl{10V24.3notapp} by the coordinate transformation
\eqref{trafo6924} of Appendix~\ref{App11VIII24.2} below, 
in which case the initial range of $\varphi \in (0,2\pi)$ corresponds to $\hy\in \R$. A calculation shows that  
 the mass aspect function transforms as 
\begin{equation}
    \bar{\mu}(\varphi) = \frac{\big(\hat \mu\circ\tan\big)(\varphi/2)}{(1+ \cos(\varphi))^2}
    = \frac{1}{4} (1+\hy^2)^2 \hat \mu(\hy)\,.
     \label{19IX24.21}
\end{equation}
For bounded functions $\bar\mu$ the function $\hat \mu$ decays as $\hy^{-4}$ when
 $\hy\to\pm\infty$, and is  \emph{not} periodic in $\hy$, whether or not it was in $\varphi$, unless equal to zero. 
 
Equation \eqref{19IX24.21} leads to
\begin{eqnarray} 
  \Hh&=&   
   \frac{1 }{2\pi}
    \int_{S^1} \bmu(\varphi) \,
    d\varphi
    \nonumber
\\
 & = &
   \frac{1 }{2\pi}
    \int_{\R} \hmu(\hy) \frac{1+\hy^2}{2}
    d \hy
  \,. 
    \label{7IX24.12gfx} 
\end{eqnarray}
We emphasize that we are \emph{not} using this correspondence in what follows, but that of Remark~\ref{R8IX24.1}. 
\qed
\end{remark}

%% file: Proof2.tex
\renewcommand{\scri}{{\mycal I}}

We can use the above to prove a version of Theorem~\ref{T15VIII24.1}. 
The argument that follows requires that the spacetime be sufficiently large to contain enough families, as described in the proof, of maximal surfaces spanned on cuts of a conformal completion at infinity, which we denote by $\scri$. 

For this, 
consider a  metric $\fourg$ defined on 
a neighborhood of  a static slice $\{t=0\}$ of  $(2+1)$-dimensional Anti-de Sitter space $(\mcM,\bfour)$, with
\begin{equation}
 \label{16IX24.5}
 \|\fourg-\bfour\| \le \epsilon
  \,,
\end{equation}
for some small number $\epsilon>0$.
The norm  $\| \cdot\|$ here, which we leave unspecified,
 \ptcn{reconsider}
   involves a finite number of weighted derivatives of the metric,
and can be made precise by the requirement of existence of maximal surfaces as guaranteed by~\cite{ShiMax} or~\cite{ChGallowayMax}, with moreover enough regularity of the metric induced on these  to define the Hamiltonian charges; the chasing through these references of the weights and of the number of derivatives is left to the reader.   

We further assume that matter fields, if any,  satisfy the positivity condition
\begin{equation}\label{16IX24.1}
  T_{\mu\nu}X^\mu X^\nu \ge 0 \ \mbox{for timelike vectors $X^\mu$.}
\end{equation}
We require the matter fields to decay sufficiently fast,
 so that there exists a smooth conformal completion at infinity with leading-order asymptotics as in Section~\ref{s11VII24.1}.
More precisely we assume that $(\mcM,\fourg)$ contains a region which can be covered by coordinates as 
in Section~\ref{s11VII24.1} for an interval of  time 
$t\in[-C\epsilon,C\epsilon]$, where $C$ is a large constant which can in principle be determined by the constructions 
in the proof of Theorem~\ref{t8IX24.1}. Note that this would be  easy to satisfy for weak fields which have a well posed initial value problem with boundary at infinity  with the  asymptotic behavior just described; whether or not this is the case for  Einstein equations in $2+1$ dimensions, be it in vacuum or with sources, is however not clear.
 \ptcn{deserves a thought}
 
 \medskip
 
 Assuming the above, we have:
 
  \begin{theorem}
    \label{t8IX24.1a} 
There exists $\epsilon_0>0$ such that for metrics satisfying \eqref{16IX24.5}
with $0\le \epsilon\le \epsilon_0 $, on each slice of constant time the following holds:

\begin{enumerate}
\item  Both $\mu\pm j$ are of monodromy type  $\mcH_0$, or  $\mycal{P}_0$,   
 or $\mcP^{-1}_1$, 
or $\mcE_m$ with $m\ge -1$,  and 
in each asymptotic gauge
we have
\begin{equation}\label{8IX24.6}
 \HamM^0 - |J|   \ge 0
  \,.
\end{equation}

  \item  
 The mass aspect function $\mu $ is of hyperbolic type $\mcH_0$,  
or elliptic type $\mcE_m$ with $m\ge -1$.

  \item  In each asymptotic gauge the energy-momentum vector is timelike future pointing, 
  \begin{equation}\label{8IX24.6b}
    H^0 > |\vec H|  
    \,,
  \end{equation}
except if $\fourg=\bfour$. 
\end{enumerate}  
  \end{theorem}

 {\noindent\sc Proof:} 
 Our aim here is to sketch how to obtain the result from the small-data analysis above. A
  complete proof in the current setting would require working out analytical estimates, and formulating more precise assumptions on the spacetime,  which would be a wasteful effort as  these can be avoided using Theorem~\ref{T15VIII24.1}. 
 
 So, suppose that one of $\mu\pm j$ is of a bad monodromy type. It follows from \eqref{30VII24.2} that one can then apply an asymptotic symmetry, of size $O(\epsilon)$,   
 to the metric so that $H^0<0$ on the new hypersurface, say  $\{\bar t=0\}$. 
 By e.g.~\cite{ChGallowayMax}
there exists a maximal hypersurface spanned on the new section $\{\bar t=0\}$ of $\scri$ (which will in general be different from the original section $\{ t=0\} \cap \scri$).
We show in Appendix~\ref{app10IX24.1} that the metric induced on  the maximal hypersurface has the same 
mass aspect function as the one on $\{\bar t=0\}$.
 From what has been said one can  transform the metric induced on the maximal surface  to the gauge $  \bcheckpsi^i=0$, again without changing the mass aspect function, and hence maintaining $H^0<0$. This contradicts the identity \eqref{26IV18.8} when $\epsilon_0$ is chosen small enough.
 
 Strict positivity of $H^0$ if $\fourg\ne \bfour$  excludes  the type $\mcP_0$ 
 which has zero mass,
and we expect that the marginal $\mcP^{-1}_1$ case
can also be excluded by a careful analysis of the sequence of gauges along which the mass tends to zero.
 
 This establishes points 1.\ and 2.

Point 3. is  obtained by applying  Theorem~\ref{t8IX24.1} to the metric induced on the maximal surface spanned on $\{t=0\}\cap\scri$.
 \hfill
 $\Box$
 %

%% file: Witten.tex
\subsection{A Witten-type argument}
 \label{ss19IX24} 
 
\ptcn{the file mass\_questions.tex has various equations which can perhaps be recycled here, commented out}
Recall that in dimensions $n\ge 3$ inequalities in the spirit of  \eqref{8IX24.6}  have been established in~\cite{Maerten,CMT} (compare \cite{CTAngMom}) when $\Lambda<0$ using a Witten-type argument. The standard approach to this,
and which we will follow,
   requires the existence of a spin bundle which carries imaginary    Killing spinors for a background metric in the asymptotic region;
 see, however, \cite{SkenderisCheng} for an alternative approach. 
 
\begin{remark}
 \label{R3XI24.1} 
Some preliminary comments about spin structures are in order.%
\footnote{We are very grateful to Nikolai Saveliev and Rudolph Zeidler for providing us with this description. The results summarized here can be inferred from e.g.\ \cite{MilnorSpin}.}
First, every two-dimensional manifold $M$ carries exactly one  spin structure if it is simply connected, and at least two distinct spin structures when   $\pi^1(M)$ is nontrivial.
Next, assume that $M$ has a nonempty boundary $\partial M$ and let $A \subset M$ be a closed annulus with one boundary coinciding with a connected component of $\partial M$.  Then  $A$ carries exactly two distinct spin structures  when $A$ is viewed as a manifold on its own, with two boundary components. Let us call   \emph{canonical}  
the spin structure induced on $A$ by restriction when $A$ is viewed as a subset of a disc, and let us call  \emph{twisted}  the remaining one.  
If $M$ is conformally compact   
and $\partial M$ is connected, then $M$ induces on $A$ the canonical spin structure. 
If $M$ is non-compact with compact conformal infinity, or if   
 the boundary of the conformally completed manifold has more than one component (whether at finite distance or at infinity),
 then
$M$ carries at least two spin structures, with one of them inducing on $A$ the canonical structure, and another inducing the twisted one.
%
%
\qed
\end{remark}

We are ready now to pass to our first proof of positivity.

\medskip

{\noindent\sc Proof of Theorem~\ref{T15VIII24.1}:}
Consider a complete two-dimensional ALH initial data set $(M,g,K)$ equipped with the canonical spin structure on an annular neighborhood of a connected component of the conformal boundary at infinity $\partial M$.
As in Section \ref{30VII24} the mass aspect is defined by   writing   $g$  as (compare Remark~\ref{R8IX24.1})
\begin{eqnarray}
  g 
   &= & 
     \frac{ dr^2}{r^2+1 } + r^2 d\varphi^2 
       + 
       \big( \bmu_{ab}(\varphi) 
       + o(r^{-2}) 
       \big)
       \theta^a \theta^b      
        \,.
       \label{22X24.46bw} 
\end{eqnarray}
%
We thus choose the background as
 \ptcn{should it be $\backg$ or $\mzg$? \\ -- \\ rw: this should be $\backg$, since we use
 $\mzg$ for our standard background, the cusp, compare to   
 \eqref{backgroundban}, \eqref{bansolutionsONB}}
\begin{equation}\label{24X24.6}
 \frac{ dr^2}{r^2+1 } + r^2 d\varphi^2
  \,.
\end{equation}
 %
For this background, in addition to the static KIDs already mentioned, and to the rotational Killing vector  of $\mzg$,  
    we   have two Killing vector fields $\vec{\mathcal C}$ of 
    \eqref{16X24.1}, with associated Hamiltonian charges $\vec C$  defined in \eqref{51024}. 
    
The full set of imaginary Killing spinors for the background \eqref{24X24.6},  in a  spin frame which avoids the coordinate singularity at $r=0$, can be found in Appendix~\ref{s24X24.1}. 
 
We wish to show  that the Witten-type positivity proof as in~\cite{ChHerzlich,CMT,GHWSpires} leads to 
\begin{align}
\label{27IX24.9}
 H^0   \ge |J| + \sqrt{|\vec C|^2 + |\vec H|^2 + 2 |\star ( \vec H\wedge \vec C) }|
  \,,
\end{align}
where 
$$\star (\vec H\wedge \vec C) = H^1 C^2 - H^2 C^1
 \,.
$$
\ptcn{some obsolete remarks in WittenComment commented out}
%
Indeed, using two-component spinors,  the analysis in~\cite{CMT} shows that the quadratic polynomial
\begin{equation} 
\label{26IX24.7}
 \mathbb{C}^2 \ni u \mapsto 
  \bar u^t  Q u
  \,,
\end{equation}
with the  
  matrix $Q$ 
given by 
\begin{equation}\label{26IX24.5}
  Q = H^0 + i \gamma^k H_k + \gamma^0 \gamma^k C_k + \frac{1}{2} i \gamma^0 \gamma^{k j} J_{kj}
  \,,
\end{equation} 
with 
%
\begin{align}
  \gamma^{k j} &:= \frac{1}{2} \left(\gamma^k \gamma^j - \gamma^j \gamma^k \right)\,, \\
  J_{kj} &=  \left(
    \begin{array}{cc}
      0& J \\
      -J &  0\\
    \end{array}
    \right)
    \,,
  \end{align}
 takes values in $[0,\infty)$
(see the unnumbered equation above \cite[Equation~(3.12)]{CMT}).

\input{WittenR}

%% file: WittenR.tex
If we use the following  representation of the $\gamma$-matrices,  
\begin{equation}
\gamma^1 = i \sigma^1\,,
\qquad 
\gamma^2 = i \sigma^2\,,
\qquad 
\gamma^0 = \sigma^3\,,
\label{firstirrep}
\end{equation}
where the $\sigma^i$'s are the Pauli matrices,
one finds 
\begin{equation}
Q = \left(
  \begin{array}{cc}
    H^0 + J & i
     C^1+ C^2 - H^1 + i
     H^2 \\
   -i C^1+ C^2 - H^1 - i
   H^2 & H^0 + J \\
  \end{array}
  \right)\,.
   \label{24X24.7}
\end{equation}
(We use the convention where the $\gamma$-matrices 
 satisfy
 \begin{equation}\label{27IX24.1}
   \gamma^\mu \gamma^\nu + 
   \gamma^\nu \gamma^\mu 
    = - 2 \eta^{\mu\nu}
    \,,
 \end{equation}
with $\eta^{\mu\nu}= \mathrm{diag}(-,+,+)$.) 
The eigenvalues of $Q$ are
\begin{align}
 H^0 + J \pm  \sqrt{(C^2 - H^1)^2 + (C^1 + H^2)^2}\,, 
\end{align}
and positivity of \eqref{26IX24.7} shows that
\begin{align}
 H^0 + J \ge \sqrt{(C^2 - H^1)^2 + (C^1 + H^2)^2} 
 \equiv \sqrt{|C|^2 + |\vec H|^2 -  2 \star ( \vec H\wedge \vec C) }
 \,.
\end{align}
Choosing instead  ${\gamma^\prime}^\mu = - \gamma^\mu$ one finds
\begin{equation}
Q^\prime = \left(
  \begin{array}{cc}
    H^0 - J & i
     C^1+ C^2 + H^1 - i
     H^2 \\
   -i C^1+ C^2 + H^1 + i
   H^2 & H^0 - J \\
  \end{array}
  \right)\,,
   \label{1XI24.3}
\end{equation}
and positivity gives
\begin{align}
 H^0 - J \ge  
  \sqrt{|C|^2 + |\vec H|^2 +2 \star ( \vec H\wedge \vec C) }
  \,.
\end{align}

 Applying the argument to the same initial data set endowed with the opposite space-orientation
 (equivalently, exchanging the previous $\gamma^1$ and  $\gamma^2$)  leads to two inequalities where $\vec H$ and $\vec C$ are interchanged, 
and \eqref{27IX24.9} follows. 

Now,  since \eqref{27IX24.9} holds in all asymptotic gauges, the mass aspect function $\mu$ can only be of good or of marginal type. 

In the marginal case there exists a sequence of gauges in which the integrals of  $\mu$, say $H^0_n$, tend  to zero as $n$ tends to infinity. It follows from \eqref{27IX24.9} that the  values of momentum, center of mass, and angular momentum in these gauges, say $\vec H_n$, $\vec C_n$, and $J_n$, also tend to zero as $n\to \infty$. This gives a clear motivation for the following definition in the marginal case:
\begin{equation}
 \label{1XI24.5}
  \underline{\vec H} = 0 = \underline{\vec C} = \underline{J}
  \,.
\end{equation}
In all remaining  cases  the $H^0$-minimizing gauge is attained, and we obtain
\begin{align}
\label{27IX24.9rd}
 \underline H^0   \ge | \underline J| + \sqrt{|\vec { \underline C}|^2 
  + |\vec{  \underline H}|^2 
  + 2 |\star (\vec{  \underline H}\wedge \vec 
  {\underline C}) 
    }|
  \,,
\end{align}
where $\vec{  \underline C}$ and $\vec{  \underline H}$  and $ \underline J$ are the values of the charges calculated in the minimizing gauge.

Clearly \eqref{27IX24.9rd} holds both in the good and in the marginal cases,  so that the claimed inequality \eqref{15VIII24.2dr} holds. 
\qed

 \begin{remark}
  \label{R30X24.1}
The vanishing of any of the four eigenvalues listed above implies the existence of an imaginary   Killing spinor.
 The Witten identity shows then that matter must be lightlike, in the sense that the length of the  
 momentum density vector $|\vec J|$ equals the energy density $\rho$.
 Metrics with such spinors have been listed in e.g.\ \cite[Theorem~5.3]{LeitnerIKS}. It would be of interest to determine which of 
 these metrics  are ALH. 
It follows from  \cite[Proposition~2.1]{LeitnerIKS} that if the space of imaginary   Killing spinors has dimension larger than one, which occurs when either of the matrices  $Q$ of \eqref{24X24.7} or $Q'$ of \eqref{1XI24.3} vanish,
 then the initial data set is vacuum.
 \qed
\end{remark}

%% file: Witten3a.tex
{\noindent \sc Proof of Theorem~\ref{T30X24.1}:}
We define  the mass aspect using a mass-aspect tensor 
$\mu_{ab}$ obtained by writing   $g$ as
\begin{eqnarray}
  g 
   &= & 
     \frac{ dr^2}{r^2  } + r^2 d\varphi^2 
       + 
       \big( \mu_{ab}(\varphi) 
       + o(r^{-2}) 
       \big)
       \theta^a \theta^b      
        \,.
       \label{22X24.46bwx} 
\end{eqnarray}

We start by noting that when $(M,g)$ has more than one ALH end,  the null expansion $\theta$ of the level sets of $r$, when determined with respect to the normal pointing towards the manifold, tends to $-1$ when $r$ tends to infinity. 
Hence, we can cut-off the manifold $M$ at some large $r$ near one of the remaining ALH ends to obtain a manifold with trapped boundary. This reduces the problem to one where $(M,g)$ has a non-empty trapped or weakly trapped boundary.
 
Thus, without loss of generality, we  can assume that we have two spin structures at disposal near that  conformal boundary at infinity where we measure the mass. If we choose the canonical spin structure, we can repeat the argument   of Theorem~\ref{T15VIII24.1},
which applies in the presence of  weakly trapped compact boundaries, leading to \eqref{27IX24.9rd}. But we can  choose instead the twisted spin structure. Then, using the imaginary Killing spinors 
 \eqref{22X24.42}-\eqref{22X24.43} with $a_2=0=b_2$, 
 the $Q$ matrix of \eqref{24X24.7} becomes  
\begin{equation}
Q = \left(
  \begin{array}{cc}
    \HamM - \epsilon J & 0
       \\
   0 &  \HamM - \epsilon J \\
  \end{array}
  \right)\,,
   \label{24X24.8}
\end{equation}
keeping in mind that $\Hh$ in \eqref{24X24.7} needs now to be replaced by  $
 \HamM$.
Positivity of $Q$ leads therefore to 
\begin{equation}\label{24X24.8456}
  \HamM=  \Hh -1 \ge |J|
  \,.
\end{equation}
Minimizing over asymptotic gauges gives the result.
 \qedskip

\begin{remark}\label{R30X24.2}
Repeating the last proof using instead the imaginary  Killing spinors of the extreme BTZ black hole leads to the same inequality,
with saturation  if $(M,g,K)$ are initial data for an extreme BTZ black hole. We expect that the ``if'' of the last sentence 
 is ``if and only if'', with a proof which 
should follow from   \cite{LeitnerIKS}, or from calculations mimicking \cite[Section~3.1]{CMT}.
 \ptcn{to do eventually? but then the same problem as in Remark \ref{R30X24.1}}
 \qed
\end{remark}

%% file: AddMass.tex
	\section{Some matter models}
	\label{s11VIII24.11}

The aim of this section is to  include some matter field models in the problem at hand.
For this it is  convenient to use the ADM notation:
 \tqn{checked the consistency except for §7}
\begin{equation}\label{3VI24.1}
	\ADMg = - N^2 dt^2 + \pg_{ij}(dx^i + N^i dt) (dx^j + N^j dt)
	\,,
\end{equation}
so that in particular
\rwc{checked on 19824}
\begin{equation}\label{3VI24.2a}
	\ADMg^ {\mu\nu}\partial_\mu \partial_\nu  = - N^{-2} (\partial_t - N^i \partial_i)^2 +   \pg^{ij} \partial_i \partial_j
	\,,
\end{equation}
and
\rwc{checked on 19824}
\begin{equation}\label{3VI24.2}
	\det \ADMg_{\mu\nu} = - N^2 \det \pg_{ij}
	\,.
\end{equation}
We will denote by $n$  the field of future-oriented unit-normals to the spacelike hypersurfaces $\{t=\mathrm{const}\}$:
\rwc{checked on 19824}
\begin{equation}
	n^\mu \partial_\mu = N^{-1} \left(\partial_t - N^{i}\partial_i\right)\,, \qquad n_\mu dx^\mu = -N \, dt\,.
\end{equation}
A non-vanishing energy-momentum tensor $T$  affects the constraint equations through the energy density
%
$$
\matter=n^\mu n^\nu \mnd T
$$
and  the matter current
$$
J^i = n^\mu g^{ij} T_{\mu j}
\,.
$$
As such, there is some freedom in the conformal method to prescribe the conformal transformation properties of the fields at hand,   
which will be exploited below.

\subsection{Maxwell fields in vacuum}
 \label{ss11VIII24.12}
Let us consider an electromagnetic field
\begin{equation}
	F=dA \,,
\end{equation}
satisfying Maxwell's equations
\begin{eqnarray}\label{MaxwellGauss}
	\nabla_\nu  F^{\mu\nu} &=& 4\pi\,j^{\mu}\,, 
\\
	\nabla_\mu {} \hodge F^{\mu} &=& 0\,,
\end{eqnarray}
where $\nabla$ is the connection of the spacetime metric $\ADMg$, $j^{\nu}$ is the charge-current three-vector, and $\hodge F$ is the dual of $F$ defined by 
\begin{equation}
	\label{hodgedualF}
	\hodge F^{\mu} = \frac12\epsilon^{\mu\nu\lambda}F_{\nu\lambda}\,,
\end{equation}
where $\epsilon^{\mu\nu\lambda}$ is the Levi-Civita tensor satisfying
\begin{equation}\label{6IX24.1}
	\epsilon^{\mu\nu\lambda} = \frac{\mathring \epsilon^{\mu\nu\lambda}}{\sqrt{-g}} \,, \qquad \epsilon_{\mu\nu\lambda} = \sqrt{-g} \,\mathring \epsilon_{\mu\nu\lambda}\,,
\end{equation}
with $\mathring \epsilon^{\mu\nu\lambda}$ the flat Levi-Civita symbol such that $\mathring \epsilon_{012}=1$ and $\mathring \epsilon^{012}=-1$.
\ptcheck{23VIII}

In $2+1$ dimensions we define the electric and magnetic fields as
\rwc{checked on 19824}
\begin{align}\label{3VI24.3}
	E^i &:= -n_\nu F^{\nu i} \equiv N\,F^{0i}\,, 
	\quad
 E^0 := 0
\,,
\\
	\quad
	B &:= n_\nu {} \hodge F^\nu \equiv \frac 12 \epsilon^{ij} F_{ij}\,, 
\end{align}
where  
	$$
\epsilon^{\mu \nu} := n_\rho \epsilon^{\rho \mu \nu}
 \,.
$$ 
This leads to the following decomposition of $F$: 
\begin{equation}
	\label{decomp}
	F^{\mu \nu} = \epsilon^{\mu \nu} B + n^\mu E^\nu - E^\mu n^\nu\,, \quad \hodge F^\mu =  n^\mu B - \epsilon^{\mu\nu}E_\nu\,.
\end{equation}
In this notation   the 
	Gauss constraint equation reads
\begin{eqnarray}
	D_i E^i &=& 4\pi\rho_Q \,, \label{4VI24.1} 
\end{eqnarray}
where $D$ is the covariant derivative of the metric $\pg_{ij}dx^idx^j$, and where 
$\rho_Q = N j^0$ is the charge density.

In the absence of charges and currents the equations can be derived from the (randomly-normalized) Lagrangian
\begin{equation}\label{26VIII24.1}
  \mycal{L} = \frac 1 {4} \sqrt{-\det g} F^{\alpha \beta} F_{\alpha\beta}
  \,,
\end{equation}
leading to an energy-momentum tensor  
\begin{align}
	T_{\mu\nu} &= 
 F_{\mu \lambda}F^{~ \lambda}_\nu - \frac14 g_{\mu\nu} F_{\rho\lambda} F^{\rho\lambda} 
 \nonumber
  \\
	&= 
\frac{B^2 + \normsq{E}{\pg}}2 (\pg_{\mu\nu} + \red{2} n_\mu n_\nu) - E_\mu E_\nu + (\epsilon_{\mu \alpha}n_\nu \red{+} \epsilon_{\nu \alpha}n_\mu)E^\alpha B
\,.
\end{align}
Hence,
\begin{align}
	\matter = \frac{B^2 + \normsq{E}{\pg}}{2} \,,\\
	J^i =  \epsilon^{ij}E_j B \,.
	\label{9II25.1}
\end{align}

The question arises, how should the fields above transform under the conformal rescaling $	\ijd\pg = e^{-2u} \ijd\tg$. One finds that \eqref{4VI24.1} will have convenient transformation properties if we require that under this conformal scaling  the electric field transforms as
\begin{equation}
	\tE^i = e^{-2u} E^i\,;
\end{equation}
in particular if $g$ is asymptotic to $ x^{-2} (dx^2 + dy^2)$ then the coordinate components of $E$ transform as
\begin{equation}
	\tE^i = x^{-2}  E^i\,.
\end{equation}
Since the left-hand side of the vector constraint equation transforms as 
$$
\tD_i \tL^{ij} = e^{-4u} D_i L^{ij}\,,
$$
it is convenient to choose
\begin{equation}
	\tB = e^{-2u} B\,,
\end{equation}
so that, using \eqref{9II25.1} and \eqref{6IX24.1},
$$ 
\tilde J^i = e^{-4u} J^i,
$$
hence,
$$
D_i L^{ij} = J^j 
 \quad
  \Longleftrightarrow 
   \quad
   \tD_i \tL^{ij} = \tilde J^j\,.
$$
As a consequence the energy density $\matter$ of the Maxwell field equals
\begin{equation}
	\matter = \matter_B + \matter_E = e^{4u}\widetilde{\matter_B} + e^{2u}\widetilde{\matter_E}\,,
\end{equation}
where
$$
\matter_B 
 : = \frac{B^2}{2}\,, \quad \matter_E 
 := \frac{\normsq{E}{\pg}}{2}\,, \qquad \widetilde{\matter_B} 
 := \frac{\tB^2}{2}\,, \quad \widetilde{\matter_E} 
 := \frac{\normsq{\tE}{\tg}}{2}\,.
$$
In coordinates in which  $g$ is asymptotic to $x^{-2}(dx^2+dy^2)$, with $\tg = dx^2+dy^2$ so that $e^u \sim x $, the finite mass condition imposes
\begin{equation}
	\tB = \gO{x^{-\frac32 +\varepsilon}}\,, \quad 
 \big( \tE^i = \gO{x^{-\frac12 +\varepsilon}}
 \ \Leftrightarrow \
  |E|_g = \gO{x^{\frac12 +\varepsilon}} 
 \big )
\,,
\label{9II25.2}
\end{equation} 
for some $\varepsilon>0$,
where $\tB$  is the magnetic field and the  $\tE^i$'s are the coordinate components of the electric field, both rescaled to  the Euclidean metric.

Finally, the conformally-rescaled constraint equations written with respect to the hyperbolic metric $\tg$ become
\begin{eqnarray}
	\tD_i \tL^{ij} &=& \tilde J^j 
\,,
\\
	\Delta_{\tg} u &=& \qa \, e^{2u} + \qb - e^{-2u}\,,
	\label{9II25.4}
\end{eqnarray}
where
\begin{equation}
	\qa = \frac12 \normsq{\tL}{\tg} + \widetilde{\matter_B} \,,\quad \qb = 1 + \widetilde{\matter_E}\,,
 \label{9II25.p1}
\end{equation}
with  $\qa,\qb \ge 0$,

\subsection{Scalar field} 
 \label{11VIII24.21}
 
Another field, relevant e.g.\ for cosmology, is the scalar field, studied for $n=3$ in \cite{BruSFAE} on asymptotically Euclidean manifolds and in \cite{BruSFCompact} on compact ones. Following these authors we consider a scalar field $\MyPsi$   satisfying 
\begin{equation}
	\nabla_\mu \nabla^\mu \MyPsi = V'(\MyPsi)\,,
\end{equation}
where $V= V(\MyPsi)$ is the potential. The energy-momentum tensor reads
\begin{equation}
	T_{\mu\nu} = \partial_\mu \MyPsi \partial_\nu \MyPsi - g_{\mu\nu} 
 \Big(\frac12 g^{\alpha \beta} \partial_\alpha \MyPsi \partial_\beta \MyPsi + V(\MyPsi)\Big)\,.
\end{equation}
We will only consider field configurations where $\myPsi$ approaches a constant at the conformal boundary. Without loss of generality we can then assume that 
$$
 \myPsi\to 0
$$
as the conformal boundary at infinity is approached. Redefining the cosmological constant if necessary, and assuming that the redefined cosmological constant remains negative, 
we can without loss of generality assume that 
\begin{equation}\label{18VIII24.1}
  V(0)=0
  \,.
\end{equation}

We use again the ADM notation and  continue to use the symbol $\myPsi$ for the restriction of $\MyPsi$ to the initial data surface. Another relevant quantity is the field 
\begin{equation}
	\pi := N^{-1} (\partial_t - N^i\partial_i)\myPsi 
\,.
\end{equation}
Using this notation, one has 
\begin{equation}
	2\matter = \pi^2 + \normsq{d\myPsi}{\pg} + 2V(\myPsi) \,,
	\qquad
	J^i = \pi \pg^{ij} \partial_j \myPsi\,.
\end{equation}

We will use the conformal method, with a rescaling as in \eqref{Scalingu}, i.e.\ $	\ijd\pg = e^{-2u} \ijd\tg$.
There is no constraint on the conformal transformation of $\myPsi$, we choose
\begin{equation}
	\tilde\myPsi = \myPsi\,,
\end{equation}
and therefore we will interchangeably use the symbols $\tilde \myPsi$ and $\myPsi$ in what follows.
However, following \cite{BruSFAE}, $\pi$ is defined with the lapse function $N$ that we require to behave  under conformal scalings as a density:
\begin{equation}
	\frac{N}{\sqrt{\det\pg}} = \frac{\tilde N}{\sqrt{\det\tg}}\,.
\end{equation}
Hence
\begin{equation}
	N = e^{-2u}\tilde N\,,
\end{equation}
and, choosing $N^i$ to be invariant under conformal transformations, we obtain
\begin{equation}
	\pi = e^{2u} \tilde\pi\,.
\end{equation}
Finally, 
\begin{equation}
	2\matter = e^{4u}|\tilde\pi|^2 + e^{2u}\normsq{d\tmyPsi}{\tg} + 2V(\tmyPsi)\,,
	\qquad
	J^i = -e^{4u}\tilde\pi\tg^{ij}\partial_j\tmyPsi =: e^{4u} \tilde J^i\,,
\end{equation}
so that the constraint equations become 
\begin{eqnarray}
	\tD_i \tL^{ij} &=& \tilde J^j
	\,,
	\\
	\Delta_{\tg} u &=& \qa \, e^{2u} + \qb - \qc \, e^{-2u}\,,
	\label{9II25.3}
\end{eqnarray}
with
\begin{equation}
	\qa = \frac12 \normsq{\tL}{\tg} + |\tilde\pi|^2_{\tg} \,,\quad \qb = 1 + \normsq{d\tmyPsi}{\tg} \,,\quad \qc = 1 - 2 V(\tmyPsi)\,.
 \label{9II25.p2}
\end{equation}
Note that the finite-mass condition imposes the asymptotic behaviors
\begin{equation}
	\qpi = \gO{\bro^{-\frac32 +\varepsilon}}\,, \quad |D\qpsi| = \gO{\bro^{-\frac12 +\varepsilon}}\,, \quad V(\qpsi) = \gO{\bro^{1+\varepsilon}}\,,
\end{equation}
for some $\varepsilon>0$ as the conformal boundary at infinity $\{\bro=0\}$  is approached, with $\qpsi\equiv \mypsi$ and  $\qpi$ denoting the fields scaled as relevant  for the Euclidean metric.

Using the notation \eqref{24VI24.11}, the vector constraint equation reads\footnote{While the complex notation $\partial_z \qpsi$  applies, it might be  somewhat misleading in that we assume that $\myPsi$ is real-valued.}
$$
\partial_{\bar z}  \qf = -{\rm i} \qpi \partial_z \qpsi\,.
$$
The function 
$\qf$ will be holomorphic only if $\qpsi$ is constant, so complex analysis does not seem to be useful here. 
However, we show in Section~\ref{11VII24.4} that one can always find a solution.

As for the existence and uniqueness of $u$,  in order to be able to apply the analysis that follows without further due, we will 
only consider initial data sets for which
\begin{equation}\label{9VII24.41}
	V(\tmyPsi)\le \frac 12
 \,.
\end{equation}
\subsection{Existence and regularity in the non-vacuum case}
 \label{11VII24.4} 
We start with the problem
\begin{eqnarray}
	\tD_i \tL^{ij} &=& \tilde J^j 
\,,\label{11VII24.2}
\\
	\Delta_{\tg} u &=& \qa \, e^{2u} + \qb - \qc \, e^{-2u}
  \label{11VII24.3}
   \,,
\end{eqnarray}
and consider the question of existence and uniqueness of solutions.

Using the York decomposition
\begin{equation}
	\tL^{ij} = \tilde D^i \tilde Y^j + \tD^j \tilde Y^i - (\tD ^k\tilde Y_k)\delta^{ij}\,,
\end{equation}
the vector constraint equation \eqref{11VII24.2} becomes
\begin{equation}
	\Delta_{\tg} \tilde Y^j +  \tilde R^j{}_k \tilde Y^k={\tilde{J}}^j\,,
\end{equation}
for which we can always find a solution with $\tilde Y^i|_{\partial\B}=0$ when the Ricci tensor $\tilde R_{ij}$ of $\tg$ is non-positive, as is the case of the hyperbolic metric.
 
Smoothness of $\tilde Y^j$ at the conformal boundary follows from standard elliptic theory when the conformally rescaled field $\tilde J^j$ extends smoothly through the boundary. More generally, a polyhomogeneous expansion of $\tilde J^j$ at the conformal boundary at infinity leads  to a polyhomogeneous expansion of $\tilde Y^j$, with coefficients which can be obtained by matching polyhomogeneous expansions of both sides of the equation.
 
For the scalar constraint \eqref{11VII24.3}, the arguments presented in Appendix~\ref{s26VII24.2} provide existence and uniqueness of solutions on Riemannian manifolds which are the union of a compact set and a finite number of ALH ends whenever the functions $\qa>0$, $ \qb$  and $\qc >0$ are uniformly bounded. 

Concerning  regularity near the boundary, let us for simplicity assume that the functions $a$, $b$ and $c$ in \eqref{11VII24.3} 
extend smoothly across the conformal boundary. Rewriting the Lichnerowicz equation as 
$$
\aL u := \aF(\mycdot,u) + \aS(\mycdot) \,,
$$
where $\mycdot$ stands for the coordinates $(x,y)$ in the half-space model, so that
$$
 \tg = \frac{dx^2+dy^2}{x^2}
$$
with
$$
\aL:= \Delta_{\tg} - 2(\qa + \qc)\,,
\
\aS:=\qa + \qb - \qc \,,
\
\aF(\mycdot,u):= \qa (e^{2u} -1 -2u) - \qc (e^{-2u} -1 +2u)
\,;
$$
we can apply \cite[Theorem 7.4.5]{AndChDiss} to obtain a
 polyhomogeneous expansion of $u$. 
 Writing  
\begin{equation}\label{15VIII24.1b}
  a(x,\myy) = a_1 (\myy) x
  + \ldots
  \,,
  \quad 
  b(x,\myy) =1+b_1 (\myy) x
  + \ldots
  \,,
  \quad 
  c(x,\myy) =1+c_1 (\myy) x
  + \ldots
  \,,
\end{equation}
where we have assumed that $a$ vanishes at the boundary (as is the case under the finite mass condition), 
we find 
\ptcheck{18VIII}
\begin{align}\label{15VIII24.1c}
  u 
  &  
= 
     \frac{(c_1-b_1-a_1)  } 2 x 
     \nonumber
\\
 &
      + \frac{1}{6} \left(a_1 (2 c_1-4 b_1)-3 a_1^2+2
   a_2-b_1^2+2 b_2+c_1^2-2 c_2\right)x^2 \log x  + O(x^2)
  \,.
\end{align}
In particular if 
\begin{equation}\label{18VIII24.2}
  a=O(x^3)\,, 
  \quad 
   b=1 +O(x^3)
   \,,
   \quad
   c= 1+O(x^3)
   \,,
\end{equation}
the solution $u$ will be smooth at (and near) $\{x=0\}$,
\ptcheck{18VIII}
with 
\begin{align}\label{15VIII24.2}
  u (x,y)
  &  
= u_2(y) x^2 
 + 
     \frac{(a_3+b_3-c_3) (y) } 4 x ^3   + \ldots
  \,,
\end{align}
with a function $u_2(y)$ which is defined globally by the initial data.

For instance, if the scalar field   is smooth at the conformal boundary, Equation \eqref{18VIII24.2} will hold for initial data satisfying 
\begin{equation}\label{18VIII24.31}
  V(\myPsi) = O(x^3)
  \,,
  \quad
 \partial \mypsi =  O(x) 
 \,,
 \quad
 |\tilde\pi|^2_{\tg} = O(x^3)
   \,,
\end{equation}
which also (more than) guarantees that our  finite total energy condition is satisfied. 
Similarly, in the case of  Maxwell fields which are smooth at the conformal boundary, we will need
\begin{equation}\label{18VIII24.3}
	\tilde B =  O(x^2) 
	\,,
	\quad
	 |E^i| = O(x^3)
	\,,
\end{equation}
where $E^i$ are the components of the electric field in local coordinates near the conformal boundary.

%% file: Gluings.tex
\section{Maskit-type gluing of ALH spacelike  initial data}
 \label{s11VIII24.14}

In \cite{ChDelayExotic}, ``Maskit-type'' gluing results are proved for AH data sets in dimension $n\ge3$. We check here 
that the main results there continue to hold when $n=2$. 

We start with  reviewing some notation from \cite{ChDelayExotic}, and  refer  to this last reference for   definitions of function spaces. 

We consider a smooth two-dimensional manifold $\overline M$ with boundary $\partial M$ and an open subset $\Omega \subset M$ such that $\partial \Omega \cap M$
 \begin{figure}
	\centering
 \includegraphics{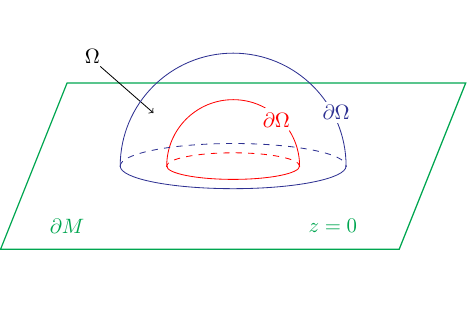}
    \caption{The set $\Omega$ and its boundaries. In this section the boundary $\partial M$ is given by $\{z=0\}$. The boundary of $\Omega$ has three components, one is a subset of $\partial M$. The function $x$ vanishes precisely on both the blue and red components of the boundary of $\Omega$.}
	\label{F26VIII21.2}
\end{figure}
 is the union of two disjoint smooth hypersurfaces $\Sigma_0, \Sigma_1$. We denote by $\gz$, resp. by $\gx$, the defining functions for $\partial M$, resp. $\overline{\partial \Omega \cap \partial M}$.
(In order to ease comparison with \cite{ChDelayExotic} we thus use $z$ for a coordinate which is denoted by $x$ or $\hat x$ at some other places in this paper.)
  We set 
 $$
  \gr:=\sqrt{\gx^2+\gz^2}
  \,,
$$
hence the inequalities on $\Omega$
\begin{equation}\label{8VI24.1}
	0 < \frac \gx \gr \le 1 , \quad 0 < \frac \gz \gr \le 1 \,.
\end{equation}
Finally, we will denote again by $\zg$ the hyperbolic model metric on $M$.
 
Let us define the constraint operator
\begin{equation}
	\sources[(K,g)]:=
	\begin{pmatrix}
		2\big(D_j (\iju K - \tr_g(K) \iju g)\big) \\
		R_{g} - \normsq K g + (\tr_{g}(K))^{2} - 2\Lambda
	\end{pmatrix}
	\,.
\end{equation}
We also introduce  $P_{(K,g)}$ as the linearization of $\sources[(K,g)]$ along with its adjoint $P^*_{(K,g)}$; and the operator $\Phi: (A,B)\mapsto(\phi A, \phi^2 B)$. We remind the reader that   elements of the kernel of $P^*_{(K,g)}$ are known as ``KIDs''.
From two AH data sets $(\Ka,\ga)$ and $(\Kb,\gb)$ close on $\Omega$, we build the interpolated couple
\begin{equation}
	(K, g)_\chi \equiv (K_\chi, g_\chi):= (1-\chi) (\Ka,\ga) + \chi (\Kb,\gb) \,,
\end{equation}
where $\chi$ is a smooth cut-off function with values in [0,1], with $\chi$ equal to zero near $\Sigma_0$ and equal to one near $\Sigma_1$.
We would like to build glued Einstein initial data sets, i.e. find data sets $(K,g)$ such as $\sources[(K,g)]$ interpolates between $\sources[(\Ka,\ga)]$ and $\sources[(\Kb,\gb)]$. For this we set
\begin{eqnarray}
	\sources_\chi &:=& (1-\chi) \sources[(\Ka,\ga)] + \chi \sources[(\Kb,\gb)] \,,\\
	\delta\sources_\chi &:=& \sources_\chi - \sources[(K_\chi,g_\chi)] \,.
\end{eqnarray}
The following theorem, proved in \cite{ChDelayExotic} assuming $n\ge 3$,   generalizes without further due to $n=2$ to show  that such data sets can be found in a neighborhood of $(K,g)_\chi$.
\begin{theorem}
	\label{5VI24.G1}
	\cite[Theorem 3.10]{ChDelayExotic}
	Let $2\le k\in  \N$, $b\in[0,1]$, $\sigma>b+\frac 12$, and
$$
\mbox{$\phi = \gx/\gr$,
 $\psi= \gx^a\gz^b\gr^c$, $\tau \in \R$.
 }
 $$
 Suppose that $(\ga-\zg) \in C^{k+4}_{1,z^{-1}}$ and
	$(\Ka-\tau\zg)\in C^{k+3}_{1,z^{-1}}$.
	For all real numbers $a$ and $c$ large enough and all
	$(\Kb,\gb)$ close enough to $ (\Ka,\ga)$ in $C^{k+3}_{1,z^{-\sigma}}(\Omega )\times C^{k+4}_{1,z^{-\sigma}}(\Omega )$
	there exists a unique couple of two-covariant symmetric tensor fields
 of the form 
	%
	$$
	(\delta K, \delta g)
	= \Phi\psi^2 \Phi P^*_{(K,g)_\chi}(\delta Y,\delta N)\in \psi^2 \left(\phi\Lpsikg{k+2}{g_\chi}\times
	\phi^2\Lpsikg{k+2}{g_\chi} \right)
	$$
	such that  $(K_\chi+\delta K,g_\chi+\delta g)$ solve
	\bel{fullcolle}
	\sources
	\left[(K,g)_\chi+(\delta K,\delta g)\right]
	- \sources[
	(K,g)_\chi]
	=\delta\sources_\chi
	\,.
	\ee
	Furthermore,  there exists a constant $C$ such that
	\begin{align}
	\label{estimatesolutionfullC}
 \|(\delta K,\delta g)
  &\|_{\psi^2
  \left(\phi\mathring H^{k+2}_{\phi,\psi}(g_\chi)\times \phi^2\mathring H^{k+2}_{\phi,\psi}(g_\chi)\right)}
\nonumber
 \\
	& \leq C \left\|
	\sources(\Kb,\gb)-\sources (\Ka,\ga)
	\right\|_{ \mathring H^{k+1}_{1,z^{-b}}(g_\chi)\times \mathring H^{k}_{1,z^{-b}}(g_\chi)}
	\,.
	\phantom{xx}
	\end{align}
	The   tensor fields $(\delta K,\delta g)$ vanish at $\partial \Omega$ and can be $C^{k+1}$-extended by zero across $\partial \Omega$.
\qed
\end{theorem}
\ptcn{proof moved to ProofGluing}
\ptcn{one should add something about asymptotics in the gluing region here; mass convergence? a decent one would be obtained by a gluing a la Mao Tao only in any case}

From the previous theorem, we deduce the following generalization of \cite[Theorem 3.12]{ChDelayExotic} to $n=2$,
where we consider initial data set which are vacuum near the conformal boundary:

\begin{theorem}\label{T3X15.1}
	Let $(M_a, K_a, g_a)$, $a=1,2$ be two asymptotically locally hyperbolic and asymptotically CMC
	two-dimensional initial data sets satisfying the Einstein (vacuum) constraint
	equations, with the same constant asymptotic values of  $\tr_{g_1}K_1$ and $\tr_{g_2}K_2$ as $\partial M$ is approached and with locally conformally flat boundaries at infinity
	$\partial\overline M_a$. Let $p_a\in \partial M_a$ be  points on the conformal boundaries.
	Then for all $\varepsilon$ sufficiently small there exist asymptotically locally hyperbolic 
 and asymptotically CMC vacuum initial data sets
	$(M_\varepsilon, K_{\varepsilon }, G_{\varepsilon })$
	such that
	\begin{enumerate}
		\item
		$M_\varepsilon$ is diffeomorphic to the interior of a boundary connected sum of the $M_a$'s, obtained
		by excising small half-balls $B_1$ around
		$p_1$ and $B_2$ around $p_2$, and identifying their boundaries.
		
		\item
		On the complement of coordinate half-balls of radius $\varepsilon$ surrounding $p_1$ and $p_2$,
		and away from the
		corresponding neck region in $M_\varepsilon$, the data
		$(K_{\varepsilon }, g_{\varepsilon })$  coincide with the original ones
		on each of the $M_a$'s.
	\end{enumerate}
\end{theorem}
%

\input{Bogovskii}

%% file: Bogovskii.tex
We note that the new approach to gluing of Mao, Oh and Tao~\cite{MaoOhTao,MaoTao},
based on Bogovski\v \i-type operators which control well the support of the solutions, allows one to produce gluings as above in dimensions $n\ge 3$ with better control of the data and lower differentiability requirements. The method generalizes to the time-symmetric case $K_{ij}=0$ in $n=2$~\cite{CCN,CDN}, but does not  apply as is to the whole system of the constraints because in space-dimension two the divergence operator acting on trace-free tensors is elliptic, which prevents one to localize the support of the solutions of the vector constraint equation.

%% file: characteristic2d.tex
\renewcommand{\blue}[1]{{#1}}

\section{Characteristic Cauchy problem}
 \label{s18VIII24.1}

An alternative to the spacelike Cauchy problem, which plays a prominent role in current mathematical studies of the Einstein equations, is the characteristic Cauchy problem. 
Here we consider the $(2+1)$-dimensional vacuum characteristic Cauchy problem in Bondi gauge. 
In this section we exceptionally consider both cases $\Lambda<0$ and $\Lambda>0$.

Let  thus $\mcN$ be a null hypersurface  in a $3$-dimensional spacetime, with all its generators meeting a circle $S^1$ transversally. If the divergence scalar of $\mcN$ has no zeros one can construct on $\mcN$  Bondi-type coordinates $(u,r,\varphi)$, with 
$$
 \mcN=\{u=0\}
  \,,
$$
and where $\varphi$ is a local coordinate on $S^1$,  in which the metric takes the form (see~\cite{MaedlerWinicour} or~\cite[Appendix~B]{ChCong0})
\begin{align}
    g_{\alpha \beta}dx^{\alpha}dx^{\beta}
   &=  -\frac{V}{r}e^{2\beta} du^2-2 e^{2\beta}dudr
  \nonumber
\\
 &\qquad
   +r^2\gamma_{\varphi\varphi}\Big(d\varphi-U^\varphi du\Big)\Big(d\varphi -U^\varphi du\Big)
    \label{23VII22.1}
    \,,
\end{align}
with $\gamma_{\varphi\varphi} = \gamma_{\varphi\varphi}(\varphi) $,
and with $\int_{S^1}\sqrt{\gamma_{\varphi\varphi} } d\varphi = 2\pi$.
The inverse metric reads
\begin{equation}
    g^{\sharp} =  e^{-2\beta} \frac{V}{r}\, \partial_r^2 - 2 e^{-2\beta} \, \partial_u\partial_r - 2e^{-2\beta} U^\varphi \,  \partial_r \partial_\varphi +\frac{1}{r^2} \gamma^{\varphi\varphi} \, \partial_\varphi \partial_\varphi\,.
\end{equation}
 By a redefinition of $\varphi$ one can always achieve 
\begin{equation}\label{19VIII24.1}
 \gamma_{\varphi\varphi} \equiv 1
 \,.
\end{equation}

The characteristic Cauchy problem is usually formulated in terms of two  null hypersurfaces intersecting transversally on a spacelike submanifold $S$.  If $\mcN$ is one of these hypersurfaces, then in the Bondi coordinates associated with $\mcN$ the free characteristic data are

\begin{enumerate}
  \item the tensor field  $\gamma_{\varphi\varphi}d\varphi^2$ on $\mcN$,   
  \item and the collection of fields 
$(\beta,U^\varphi, 
\partial_r U^\varphi,
 \partial_u U^\varphi, 
V)$ on $S$.  
\end{enumerate}

The part of the data associated purely with $\mcN$ is therefore trivial, since $\gamma_{\varphi\varphi}d\varphi^2$ can always be transformed to $d\varphi^2$.  Note that this renders the $S^1$-coordinate $\varphi$  unique up to a rotation of the circle. Given the free characteristic data, the Einstein equations can be used to solve for all components of the metric on $\mcN$. 

%% file: summary.tex
Let us pass to an analysis of the vacuum Einstein equations,
\begin{equation}\label{19VIII24.2}
  G_{\mu\nu} + \Lambda g_{\mu\nu} = 0
  \,,
\end{equation}
 in this setting in the gauge \eqref{19VIII24.1}.
First, the vanishing 
of the $rr$ component of the Einstein tensor provides an equation restricting $\beta$ to a $r$-independent function:
\begin{equation}
         \label{19VII.1}
          0\equiv\frac{r}{2} G_{rr} = \partial_{r} \beta \implies 
          \beta=\beta(u,\varphi)  
          \,.
\end{equation}

Next, the $rA$-component of the Einstein equations can be written as
\begin{align}
         0 \equiv 2   r  e^{2\beta}  G_{rA} = \partial_r \left[r^{3} (\partial_r U^\varphi) + r \partial_\varphi e^{2\beta}\right] \,.
\end{align}
This can be integrated in the variable $r$ to give
\begin{align} 
  U^\varphi 
              =  \mU{} (u,\varphi)r^{-2} + \mmU (u,\varphi)  
               +\frac{\partial_\varphi e^{2\beta (u,\varphi)}} r
               \label{14VIII23.22}%
                  \,,
           \end{align}
for some arbitrary $S^1$-vector fields  $\mU$ and $\mmU$. 

The metric function $V$ can now be determined from the equation 
\begin{align}
           0 &\equiv 2\Lambda r^2 - r^2g^{AB}R_{AB} 
              \nn
              \\
              &=
                2\Lambda r^2+ 2  \Big[\partial_\varphi^2 \beta
                + (\partial_\varphi\beta) (\partial_\varphi \beta)\Big]
               - e^{-2\beta}\partial_\varphi \Big[ \partial_r (r^{2 }U^\varphi)\Big]
               \nonumber
\\
                &\qquad
                +\frac{1}{2}r^4 e^{-4\beta}  (\partial_r U^\varphi)(\partial_r U^\varphi)
                +r  e^{-2\beta} \partial_r(  V/r )\,.
                   \label{13VIII24.2}
           \end{align}
Using the solution for $U^\varphi$ gives:  
\begin{align} 
\frac{V}{r}  &
   = \frac{e^{-2\beta}\mU{}^2}{r^2} + 2 r \partial_\varphi \mmU  - e^{2\beta}\Lambda r^2 +\frac{4\mU \partial_\varphi \beta}{r}  + c(u,\varphi) 
   \,,
   \label{8XI24.w1}
\end{align}
where the \emph{Trautman-Bondi mass-aspect function}  $c$ is  determined globally from the initial data.

We continue with the equation $G_{u\varphi} = \Lambda r^2 U^\varphi$ , 
\begin{align}
    0 &\equiv re^{2\beta}(G_{u\varphi}-\Lambda r^2 U^\varphi)
    \nn
    \\
    &=
    -2 \mU \partial_u\beta + \partial_u\mU - 2 \mmU \mU \partial_\varphi \beta -  e^{2 \beta} c \partial_\varphi \beta + 8 e^{4 \beta} (\partial_\varphi \beta)^3 
    \nn
    \\
    &\quad
    -\frac{1}{2} e^{2\beta} \partial_{\varphi} c + 2 \mU \partial_\varphi \mmU + \mmU \partial_\varphi \mU + 4 e^{4 \beta} \partial_\varphi \beta \partial_\varphi^2 \beta 
    \,,
    \label{18VIII24.5}
\end{align}
which can be thought-of as a $\partial_u$-evolution equation for $ \mU$.

The equation  
$$
 G_{uu} = -\Lambda
  \big(
   - \frac{V e^{2\beta}}r+r^2 (U^\varphi)^2
   \big)
$$
provides  a $\partial_u$-evolution equation for  $c$:
\begin{align}
    0 &\equiv r
    (G_{uu} + \Lambda g_{uu})
    \nn 
    \\
    & = 
    -  {c \partial_u\beta}  +  \frac{\partial_uc}{2 } +  {2 e^{2 \beta} \Lambda \mU \partial_\varphi \beta}  
    -   {\mmU c \partial_\varphi \beta}  -   {8 e^{2 \beta} (\partial_\varphi \beta)^2 \partial_\varphi \mmU}  +  {c \partial_\varphi \mmU}  
    \nn
    \\
    &\quad 
    +\frac{\mmU \partial_\varphi c}{2} +  {e^{2 \beta} \Lambda \partial_\varphi \mU}  -   {8 e^{2 \beta} \partial_\varphi \beta \partial_\varphi \partial_u\beta}  -   {8 e^{2 \beta} \mmU \partial_\varphi \beta \partial_\varphi^2 \beta}  
    \nn
    \\
    &\quad 
    -   {4 e^{2 \beta} \partial_\varphi \mmU \partial_\varphi^2 \beta}  -   {2 e^{2 \beta} \partial_\varphi^2 \partial_u\beta}  -   {2 e^{2 \beta} \mmU \partial_\varphi^3 \beta}  +  {e^{2 \beta} \partial_\varphi^3 \mmU} 
    \,.
    \label{18VIII24.4}
\end{align}

The $G^r{}_A$-equation turns out to be identical to \eqref{18VIII24.5}.  

Similarly, the {$G^r{}_u$-equation turns out to coincide with \eqref{18VIII24.4}.%

We end the list by noting that the only remaining equation, out of the whole set of vacuum Einstein equations, namely $R = 6\Lambda$, reads, after making use of \eqref{19VII.1}, \eqref{14VIII23.22} and \eqref{8XI24.w1},
 \begin{align}
     6 \Lambda
     + 4 \underbrace{\partial_{r} \partial_{u}\beta}_{=  0}
      = 6 \Lambda
      \,,
 \end{align}
and is therefore trivially satisfied with the fields determined so far.

\input{NullGluing} 
 
Returning to our main line of thought, we conclude that for any freely-prescribable pair of  functions 
 $$
  \big(\beta(u,\varphi),\mmU(u,\varphi)\big)
 $$
 the vacuum Einstein equations  \eqref{18VIII24.5}-\eqref{18VIII24.4} are equivalent to the following  system of PDEs for $(c,\mU)$: 
\begin{align}
  - \partial_u\mU- \mmU \partial_\varphi \mU
    +\frac{1}{2} e^{2\beta} \partial_{\varphi} c
    & =
    {(2 \partial_\varphi \mmU  -2 \partial_u\beta   - 2 \mmU\partial_\varphi \beta  ) \mU}  -  {e^{2 \beta} c \partial_\varphi \beta } 
    \nn
    \\
    &\quad
     + 8 e^{4 \beta} (\partial_\varphi \beta)^3 + 4 e^{4 \beta} \partial_\varphi \beta \partial_\varphi^2 \beta 
    \,,
 \\ 
      -  \frac{\partial_uc}{2 }   
    -
    \frac{\mmU \partial_\varphi c}{2} 
     - {e^{2 \beta} \Lambda \partial_\varphi \mU}
      & = 
       ({ \partial_\varphi \mmU}  -  { \partial_u\beta}  
    -   {\mmU \partial_\varphi \beta}  
      ) c
     +  {2 e^{2 \beta} \Lambda \mU \partial_\varphi \beta}  
    \nn
    \\
    &\quad  
    -   8 e^{2 \beta} \Big( {(\partial_\varphi \beta)^2 \partial_\varphi \mmU} 
     + {\partial_\varphi \beta \partial_\varphi \partial_u\beta}
       + {\mmU \partial_\varphi \beta \partial_\varphi^2 \beta}  \Big)
    \nn
    \\
    &\quad 
    -  2 e^{2 \beta} \Big( {2  \partial_\varphi \mmU \partial_\varphi^2 \beta}  +   { \partial_\varphi^2 \partial_u\beta}  +   {\mmU \partial_\varphi^3 \beta}  -  {\frac{1}{2} \partial_\varphi^3 \mmU} \Big)
    \,.   
    \label{18VIII24.6}
\end{align}
Letting 
$$
 \breve c := c/2
 \,,
$$
this can be rewritten as
\begin{equation}\label{19VIII24.4}
  \bigg[
   \left(
     \begin{array}{cc}
       1 & 0 \\
       0 & 1 \\
     \end{array}
   \right)
   \partial_u
   +
   \left( 
     \begin{array}{cc}
       \mmU& -e^{2\beta} \\
      \Lambda e^{2\beta} & \mmU \\
     \end{array}
   \right)
   \partial_\varphi
   \bigg]
   \left(
     \begin{array}{c}
       \mU \\
       \displaystyle
       \breve c \\
     \end{array}
   \right)
   =
    F(\mU,\breve c)
   \,,
\end{equation}
with an affine function $F$. Rather surprisingly, the system is

\begin{enumerate}
  \item  manifestly symmetric hyperbolic if and only if $\Lambda=-1$.
  \item  Recall (cf., e.g., \cite{RauchBareHands}) 
   that a system with principal part of the form $\partial_u + A^\varphi \partial_\varphi$
   is hyperbolic if all the roots of the polynomial 
   $$
    P(\tau):=\det(\tau + A^\varphi p_\varphi)
    $$
     are real; strictly hyperbolic if the roots are real and distinct. One readily checks that \eqref{19VIII24.4} is hyperbolic if and only if $\Lambda \le 0$, strictly hyperbolic if $\Lambda <0$. 
\item The system \eqref{19VIII24.4} is elliptic when $\Lambda >0$. 
\end{enumerate} 

It follows that: 

\begin{enumerate}
  \item   When  $\Lambda <0$ the system has unique global smooth  solutions for any smooth initial data $(\mU,c)|_{u=u_0}$ (cf., e.g., \cite{KreissLorenz}). 
      
  \item  When $\Lambda=-1$ the system also has unique global solutions for any initial data $(\mU,c)|_{u=u_0}$ in Sobolev spaces.  
      
\item When $\Lambda >0$ the description of the set of solutions does not seem to be obvious in general. In the simplest case $\beta\equiv0\equiv\mmU$  and $\Lambda =1$, Equation~\eqref{19VIII24.4} is the requirement 
    (cf.\ \eqref{21VIII24.1} below) that the function $\mU-i c/2$ be   holomorphic in $u+i \varphi$. 
\end{enumerate}

%% file: NullGluing.tex
 \begin{remark}
 \label{R10X.1}
A recent addition to the study of the Cauchy problem in general relativity is the  gluing method for characteristic initial data of Aretakis, Czimek and Rodnianski, see
\cite{ACR3,CzimekRodnianski} and references therein. The question is,
 whether characteristic initial data on a null hypersurface $(r_0,r_1]\times S$ can be glued, together with a number of transverse derivatives of the metric, with initial data on  $[r_2,r_3)\times S$, using an intermediate set of data on  $[r_1,r_2]\times S$. Here $r_0<r_1<r_2<r_3$. The construction is based on the possibility to manipulate,  in the interpolating region   $[r_1,r_2]\times S$,  both 
 the free gravitational data and 
 the gauge freedom of the data. 
 The gluing is  obstructed by a collection of 
 radial charges in several cases of interest. 
 
 In $D$-dimensional spacetimes with $D\geq 4$, the free gravitational data is provided by the field $\gamma_{AB}$ \cite{ChCongGray1}. It follows from \eqref{19VIII24.1} that there are no free gravitational data in two dimensions which could be used for gluing, thus only  gauge freedom and obstructions remain. In the linearized vacuum case, \
 gauge freedom is parametrized by two functions, $\xi^u(u,\varphi)$ and $\xi^\varphi(u,\varphi)$, on $S$. As analyzed in~\cite[Section 3.2]{ChCongGray1} the field $\delta\beta$ can be gauged to zero by setting (cf.\ Equation (3.23) there)
$$ 
     \partial_u\xi^u - \partial_{\varphi}\xi^{\varphi} = -2 \delta\beta \,,$$ 
where $\delta g_{\mu\nu}$ denotes the linearized perturbation of the metric. In the gauge $\delta\beta = 0$, the field $\delta\mmU = - \frac{1}{2r} \partial_r \delta g_{uA} $ can also be gauged to zero, by setting (cf.\ \cite[Equation 3.29]{ChCongGray1})
$$ \partial_u \xi^{\varphi} + (\alpha^2 + \frac{2m}{r^2}) \partial_\varphi \xi^u = -\frac{1}{2r} \partial_r \delta g_{uA}\,.$$
The obstructions are provided by the fields $\delta c$ and $\delta \mU$, 
 which does not leave any freedom for the gluing. We believe that this statement carries over to the small-perturbation regime of nonlinear gluing as in 
\cite{ACR3,CzimekRodnianski,ChCongGray2}, but we have not attempted to work-out all the details of this.  

Summarizing, the $2+1$ vacuum characteristic gluing appears to be completely rigid, in the sense just explained.

Adding matter fields renders the problem more flexible, in that one can use the freedom in the matter fields data to carry out the gluing.
 \qed
\end{remark}  

%% file: BMS.tex
 It is natural to raise the question, which of the functions appearing above are gauge. 
As such, locally and in vacuum, all the functions are gauge, since the metric can always locally be brought to the (Anti-)de Sitter or Minkowski form depending upon the value of $\Lambda$. 
Whether or not this can be done globally is a more delicate issue.
In order to make contact with the analysis of Section~\ref{s11VII24.1},
following~\cite{Barnich:2012aw}%
\footnote{See~\cite{Barnich:2010eb} for a complementing analysis   when $\Lambda=0$.}
we consider  metrics in which only $\beta$ and $\mmU$ are gauged-away to zero. In order to make the
comparison with the formulae of Section~\ref{subsec:asym} easier,
 we change notation and make the replacement
%
\begin{equation}\label{30VIII24.1}
  c \mapsto - \mu
  \,,
  \qquad
  \mU \mapsto \frac j2
  \,,
\end{equation}
so that the metric reads 
\begin{align}
    g 
   &=  -
   \Big( \frac{\red{j}^2}{4r^2} - \Lambda r^2  \red{-\mu}
   \Big)  du^2-2  du\,dr 
   +r^2 \Big(d\varphi-
    \frac{\red{j}}{2 r^{2}}  
                du\Big)
                \Big(d\varphi-
    \frac{\red{j}}{2 r^{2}}  
                du\Big)
                \nonumber
\\ 
&=  -
   \Big(  - \Lambda r^2  \red{-\mu}
   \Big)  du^2-2  du\,dr 
    - \red{j} d\varphi \, dr 
   +r^2 d\varphi^2 
    \label{23VII22.1qf}
    \,.
\end{align}
The vacuum Einstein equations  \eqref{19VIII24.4} reduce now to 
\begin{equation}\label{21VIII24.10}
  \partial_u \red{\mu} = - \Lambda
  \partial_\varphi \red{j}
  \,,
  \qquad
  \partial_\varphi  \red{\mu} = 
  \partial_u \red{j}
  \,.
\end{equation}
Hence $\mu$ and $j$ are solutions of the second-order equations
\begin{equation}\label{21VIII24.1}
 (  \partial_u^2   +  \Lambda
  \partial_\varphi^2 )  \red{\mu} = 0 = 
   (  \partial_u^2 + \Lambda
  \partial_\varphi^2) \red{j}
  \,.
\end{equation}

\subsection{Global charges, asymptotic symmetries}

The Trautman-Bondi mass, defined as 
\begin{equation}\label{21VIII24.3x}
  \mTB : =  \frac{1}{2\pi}\int_{S^1}  \mu 
  \,d\varphi
  \,,
\end{equation}
is $u$-independent,
\begin{equation}\label{21VIII24.3xb}
\frac{d\mTB}{du} 
 =    \frac{\Lambda}{2\pi}\int_{S^1}  \partial_\varphi\red{j} d\varphi = 0
 \,, 
\end{equation}
regardless of the value of $\Lambda$; similarly for the total angular momentum
\begin{equation}\label{21VIII24.3b}
  J : =   \frac{1}{2\pi}\int_{S^1} \red{j}
  \,d\varphi
  \,,
  \qquad
   \frac{dJ}{du} = 0
   \,.
\end{equation}
However, these quantities suffer from the same problems as their spacelike counterparts of Section~\ref{s14VIII24.1}.
This requires an analysis of asymptotic symmetries, which proceeds as follows:

\input{AnalysisAS}

%% file: AnalysisAS.tex
 Note first that to  preserve the dominant  terms in the metric, 
$$
 r^2 h :=  r^2 (  \Lambda du^2 + d\varphi^2 ) 
  \,,
$$
the leading order of a coordinate transformation which preserves the asymptotic form of the metric, say   $(u,\varphi)\mapsto (\breve  u,\breve   \varphi)$, must be  
a conformal map from the Euclidean ($\Lambda>0$) or Minkowskian  $(\Lambda <0)$ plane $\Lambda du^2 + d\varphi^2$ to itself:
\begin{equation}
 \label{22VII24.1}
 \Lambda du^2 + d\varphi^2 = \psi^2 
 \big(\Lambda d\breve   u^2 + d\breve  \varphi^2
 \big)
 \,,
\end{equation}
for some function $\psi>0$.
As is well known, both functions $(u,\varphi)\mapsto (\breve  u,\breve   \varphi)$ are then solutions of the wave equation when $\Lambda <0$, or are harmonic when $\Lambda>0$. 

To continue, it is convenient to treat the cases of positive and negative cosmological constant
separately. 

\subsubsection{\texorpdfstring{$\Lambda<0$}{Lambda < 0}}
When $\Lambda<0$ there exist functions $\bv_{\pm}$ such that 
\begin{equation}\label{23VIII24.7}
  \breve  u = \frac{\bv_+({\bar x}^+) + \bv_-({\bar x}^-)}{2}
  \,,
   \ 
  \breve  \phi = \frac{\bv_+({\bar x}^+) - \bv_-({\bar x}^-)}{2}
  \,,
  \ 
  \mbox{where} 
  \
  x^\pm = \sqrt{-\Lambda} u \pm \varphi
   \,.
\end{equation}
In order to maintain orientation, we assume that $\bv_{\pm}'>0$. 
To preserve the asymptotic form of the metric we must have
%
\begin{align}
  \label{22VII24.2neg}
  \bar u &= 
  \frac{\bv_+({\bar x}^+)+\bv_-({\bar x}^-)}{2 \sqrt{-\Lambda }}
  +
  \frac{\left(\sqrt{\bv_+^\prime({\bar x}^+)}-\sqrt{\bv_-^\prime({\bar x}^-)}\right)^2}{ 2 r \Lambda
  } + O(r^{-2})
  \\
  \bar \varphi &= \frac{\bv_+({\bar x}^+)-\bv_-({\bar x}^-)}{2}  +\frac{
  \left(\bv_-^\prime({\bar x}^-)-\bv_+^\prime({\bar x}^+)\right)}{2 r \sqrt{-\Lambda
  }}+ O(r^{-2})\,, \\
 \bar r &=  \frac{r}{\sqrt{\bv_+^\prime({\bar x}^+) \bv_-^\prime({\bar x}^-)}} + O(1)
  \,,  \label{22VII24.2neg2}
 \end{align}
 where $x^\pm = \sqrt{-\Lambda } u \pm \varphi$.  
One thus obtains the following formulae for the transformation of the mass-aspect and angular-momentum aspect under asymptotic symmetries: 
\begin{align}
  \label{29VII24.2neg}
 \mu(u, \varphi) &=  \red{\frac{1}{2}}\left(\bar{\mu} \red{-}\sqrt{-\Lambda } \bar{\red{j}}\right) 
  (\bv_+^\prime({\bar x}^+))^2
  +  \red{\frac{1}{2}}\left(\bar{\mu} \red{+}\sqrt{-\Lambda } \bar{\red{j}}\right) (\bv_-^\prime({\bar x}^-))^2
  \nonumber \\ 
  &~~~ \red{-}
  \hSchw[\bv_-'](x^-) \red{-}
  \hSchw[\bv_+'](x^+)\\ 
  \sqrt{-\Lambda } \red{j}(u, \varphi) &
  = \red{-}\red{\frac{1}{2}} \left(\bar{\mu} \red{-}\sqrt{-\Lambda } \bar{\red{j}}\right) (\bv_+^\prime({\bar x}^+))^2
  \red{+} \red{\frac{1}{2}} \left(\bar{\mu} \red{+}\sqrt{-\Lambda } \bar{\red{j}}\right) 
 (\bv_-^\prime({\bar x}^-))^2 \nonumber  \\
  &~~~-  \hSchw[\bv_-'](x^-)
  +  \hSchw[\bv_+'](x^+)\,,
\end{align}
with 
\begin{align}\bar{\mu} &= \bar{\mu} \left(\frac{\bv_+({\bar x}^+)+\bv_-({\bar x}^-)}{2\sqrt{-\Lambda }}, 
  \frac{\bv_+({\bar x}^+)-\bv_-({\bar x}^-)}{2}\right) \,, 
\\
\bar{\red{j}} &=  \bar{\red{j}} \left(  \frac{\bv_+\left(x^+
\right)+\bv_-^\prime({\bar x}^-)}{2 \sqrt{-\Lambda }}, \frac{\bv_+\left(x^+
\right)-\bv_-({\bar x}^-)}{2}\right)
 \,,
\end{align}
where $\hSchw[\bv_+'](x)$ is defined in \eqref{10IX23.31}.  
This is formally identical with \eqref{eqmuchange1} and \eqref{eqmuchange3}, though the context is somewhat different.
(To obtain these equations one needs  the next order terms in the expansions \eqref{22VII24.2neg}-\eqref{22VII24.2neg2}, which we haven't included here as they are lengthy and not very illuminating.)

\subsubsection{\texorpdfstring{$\Lambda >0 $}{Lambda > 0}}

Returning to \eqref{22VII24.1},
when the   cosmological constant is positive, we find instead 
\rwc{checked with piotr leading and first subleading order\\ -- \\ ptc: do we need the next one for the mass transformation laws? \\ -- \\ rw: yes, we need the next one}
 \begin{align}
  \label{22VII24.2pos}
  \bar u &= \breve u +\frac{
  \left(\partial_\varphi {\breve  \varphi}-
  |  d\breve  \varphi |_{h}\right)}{r \Lambda
  } + O(r^{-2}) 
  \\
  \bar \varphi &= {\breve  \varphi}+\frac{
 \partial_u {\breve  \varphi}}{r \Lambda } +O(r^{-2})\,, \\
 \bar r &=  \frac{r}{|d \breve  \varphi |_{h}} + O(1)
  \,,
 \end{align}
 where $|  d\breve  \varphi |_{h} = \sqrt{(\partial_\varphi {\breve  \varphi})^2+
 (\partial_u {\breve  \varphi})^2/\Lambda}$, 
 and
 where $\breve  u = \breve  u(u \sqrt{\Lambda}, \varphi)$ 
 and $\breve  \varphi = \breve  \varphi(u \sqrt{\Lambda}, \varphi)$, with
 \begin{equation}
  \label{deruphi}
 \partial_u \breve  u = \partial_\varphi 
 \breve  \varphi \,, \qquad 
 \partial_\varphi \breve  u = - \frac{\partial_u 
 \breve  \varphi}{\Lambda}\,.
 \end{equation}
 For ease of notation,
  in what follows we normalize $\Lambda$ to $\Lambda =1$. 
Equivalently, instead of parameterizing the asymptotic change of coordinates \eqref{22VII24.2pos} in terms of 
real functions $\breve  u$, ${\breve  \varphi}$  satisfying \eqref{deruphi}, 
we can introduce a  function  $w_\pm$ holomorphic in $z=u + i \varphi$:  
  \begin{align}
    w_+(z) &=    \breve u+  i \breve  \varphi
     \,, 
    \end{align}
  yielding 
  \begin{align}
    \breve  u & 
    = \frac{1}{2}  (w_+(z) + \overline{w_+(z)} )
    \,.
    \label{breveuvarphi}
  \end{align}
 Under the above, the mass aspect function $\red{\mu}$ transforms as, 
 \begin{align}
\mu(u, \varphi)&= ((\partial_\varphi \breve  \varphi)^2 - (\partial_u \breve  \varphi)^2) \,\bar{\mu}(\breve  u, \breve  \varphi)
-
2 \bar{j} (\breve  u, \breve  \varphi) \partial_\varphi 
\breve  \varphi \partial_u
\breve  \varphi \nonumber 
\\
&~~~
+\frac{2 \left(\partial_\varphi  \breve  \varphi \partial^2_u \partial_\varphi  \breve  \varphi+\partial_u  \breve  \varphi \partial^3_u  \breve  \varphi\right)}{|  d\breve  \varphi |_{h}^2}
\nonumber 
 \\
&~~~ - \frac{3 \left((\partial_\varphi  \breve  \varphi)^2- (\partial_u  \breve  \varphi)^2\right)
   \left((\partial_u \partial_\varphi \breve  \varphi)^2- (\partial_u \partial_u \breve  \varphi)^2\right)}{|  d\breve  \varphi |_{h}^4}
\nonumber \\
&~~~ -\frac{12\, \partial_\varphi  \breve  \varphi \,\,
\partial_u  \breve  \varphi \,\,
\partial_u \partial_\varphi \breve  \varphi \,
\partial_u \partial_u \breve  \varphi}{|  d\breve  \varphi |_{h}^4}
\\ 
&= 
  \frac{1}{2}\left(\bar{\mu} + i  \bar{j}\right) (w_+^\prime(z))^2
  +
  \frac{1}{2}\left(\bar{\mu} - i \bar{j}\right) 
  \overline{(w_+^\prime({z}))^2}
 \nonumber \\
&~~~
+ \overline{\hSchw[w_+'](z)} + \hSchw[w_+'](z)
\\
&= 
   \Re\Big(\left(\bar{\mu} + i  \bar{j}\right) (w_+^\prime(z))^2
   + 2 \hSchw[w_+'](z)\Big)\,,
 \end{align}
 where in the second equality  $\bar{\mu} $ and $\bar j$ are understood as functions of $(\breve u,\breve \varphi) = (\Re w_+,  \Im w_+ )$.  
 As for  the angular aspect function $\red{j}$ we find 
 \begin{align}
   \red{j}(u, \varphi)&=
  ((\partial_\varphi \breve  \varphi)^2 - (\partial_u \breve  \varphi)^2) \,\bar{\red{j}}(\breve  u, \breve  \varphi)
  + 2\bar{\mu}(\breve  u, \breve  \varphi) \partial_\varphi 
  \breve  \varphi \partial_u
  \breve  \varphi \nonumber \\
  &~~~+ 2
  \frac{\partial_\varphi \breve  \varphi \, 
  \partial_u \partial^2_\varphi \breve  \varphi+\partial_u \breve  \varphi
  \partial^2_u \partial_\varphi \breve  \varphi}{|  d\breve  \varphi |_{h}^2} \nonumber 
  \\
  &~~~ + \frac{6 \left(\partial_u \breve  \varphi \partial_u \partial_\varphi \breve  \varphi-\partial_\varphi \breve  \varphi \partial^2_u \breve  \varphi\right)
  \left(\partial_\varphi \breve  \varphi
  \partial_u \partial_\varphi \breve  \varphi+\partial_u \breve  \varphi
  \partial^2_u \breve  \varphi\right)}{|  d\breve  \varphi |_{h}^4}
\\
  &= -\frac{i}{2} \left(\bar{\mu} + i  \bar{\red{j}}\right) (w_+^\prime(z))^2
  +\frac{i}{2} \left(\bar{\mu} - i  \bar{\red{j}}\right) 
 \overline{(w_+^\prime(z))^2} 
 \nonumber  \\
  &~~~ + i  \overline{\hSchw[w_+'](z)}
  - i \hSchw[w_+'](\red{z})  \\
  &= 
    \,\Im\Big(\left(\bar{\mu} + i  \bar{j}\right) (w_+^\prime(z))^2
   + 2 \hSchw[w_+'](z)\Big)\,,
\end{align}
which can be compared with \eqref{eqmuchange3}-\eqref{eqmuchange4}.

%% file: HPHomographie.tex

\section{Poles at the boundary of hyperbolic space}
\label{App11VIII24.2}

Let $n\ge 2$, and let $\Bn$ denote  the open unit ball centered at the origin
	$$
	\Bn := \left\{(x^1,\dots \,,x^n), \, |\vec x|^2 = \sum_{i=1}^{n} (x^i)^2 <1 \right\} \,,
	$$
	with boundary $\partial \Bn  = \mathbb S^{n-1}$. 
	We will write $\Bnr$ for a ball ($\Bpr$ for a disc,
 in dimension two), and   $S(p,r)$ for a sphere (circle,
in dimension two) of  radius $r>0$ centered at $p$.
The ball underlies the Poincaré disc model $(\Bn ,\tg)$ for hyperbolic space, with the hyperbolic metric taking the form
\begin{eqnarray} 
	\tg 
 & = &
   \frac 4{(1-|\vec x|^2)^2} \left((dx^1)^2 + \dots +(dx^n)^2 \right)
\\
 &=&
  |_{n=2} 
  \
 \frac 4{(1-x^2-y^2)^2} \left( dx^2 +dy^2 \right) \,.
\end{eqnarray}
\tqn{maybe change here since $x^{-2}\delta$ is used without any hats in \ref{ss11VIII24.2}}
An alternative useful model of two-dimensional hyperbolic space is the Poincaré half-plane 
$$
 \HP:=\left\{(\hx,\hy),\,\hx>0\right\}
 $$
 with metric
\begin{equation}\label{10V24.3}
	\hg
	=
	\frac{d\hx^2 + d\hy^2}{\hx^2}
	= \frac{d\hz \, d\bar \hz }{(\Re \hz)^2}
	\,,
\end{equation}
where the complex notation $\hz = \hx + {\rm i} \hy$ has been used.
%

The half-plane model is linked to the Poincar\'e-disc model  by the transformation
\begin{equation}
	\hx = \frac{2(\bx+1)}{\by^2 + (\bx+1)^2} -1 = \frac{1-\bx^2-\by^2}{\by^2 + (\bx+1)^2} \,,
	\qquad
	\hy = \frac{2\by}{\by^2 + (\bx+1)^2}
	\,,
\label{7XI24.3}
\end{equation}
and its inverted version
\begin{equation}
	\bx = \frac{1 - \hx^2 - \hy^2}{\hy^2 + (\hx+1)^2} \,,
	\qquad
	\by = \frac{2\hy}{\hy^2 + (\hx+1)^2} \,.
\end{equation}
Note that we can write 
\begin{equation}
	\bar \bz = \frac{1-\hz}{1+\hz} \,, \qquad \bar \hz = \frac{1-\bz}{1+\bz} \,,
\end{equation}
where a bar denotes complex conjugation.

A pole on $\partial \B $ corresponds to 
a pole on $\partial \HP$ (and reciprocally), which can also be seen from the explicit formula, for $\bang\in (-\pi,\pi)$,  
%
\begin{equation}\label{27V24}
	\frac 1 {\bz-e^{{\rm i}\bang}} = \frac 1{\bar \hz + {\rm i}\tan\left(\frac\bang2\right)} \left[\frac {-(1+{\rm i}\tan\left(\frac\bang2 \right) )}{1+e^{{\rm i}\bang}}\right] - \left[\frac 1{1+e^{{\rm i}\bang}}\right] \,.
\end{equation}

Lastly, we will sometimes  parameterize two-dimensional hyperbolic space as 
\begin{equation}
\bar{g} = \frac{dr^2}{r^2+1} +r^2 d \varphi^2\,,
\end{equation}
where $r \in \mathbb{R}_0^+$ and $\varphi \in \mathbb{S}^1$.
This model of hyperbolic space is connected to \eqref{10V24.3} by the 
transformation
\begin{equation}
	\label{trafo6924}
	\hx = \frac{1}{\sqrt{r^2+1}+r
   \cos (\varphi )}, \qquad 
   \hy =  \frac{r \sin
   (\varphi)}{\sqrt{r^2+1}+\cos (\varphi )}\,,
\end{equation}
where we assumed that $- \pi \leq \varphi < \pi$. 

%% file: Lichnerowicz3.tex
\section{The Lichnerowicz equation: existence, uniqueness, boundary regularity}
 \label{s26VII24.2}

In space-dimension $n=2$ the conformal method for solving the general relativistic constraint equations requires an understanding of equations of the form 
\begin{equation}
		\Delta_{\tilde{g}} u(\cdot) = \NL(\cdot,u) \,,
 \label{11VII24.1-x}
\end{equation}
on ALH manifolds $(M,\tg)$. This problem has been addressed in detail in~\cite{AndChDiss} assuming $n\ge3$. However, the results there apply without further due with $n=2$. For instance:

 \begin{enumerate}
   \item 
   The  monotone iteration scheme on a sequence of exhausting sets, together with a diagonalization argument, provides existence  of bounded solutions of \eqref{11VII24.1-x} whenever $\NL$ is continuous in $u$, and if there exists constants $c_1\le c_2$ such that 
\begin{equation}\label{1XI24.6}
\NL(\bcoord,c_1) \le 0 
 \,,
  \qquad
\NL(\bcoord,c_2) \ge 0 
 \,.	
\end{equation}
We emphasize that no further assumptions on the geometry of $(M,\tg)$ are necessary whenever \eqref{1XI24.6} holds. In particular, for the Lichnerowicz equation on a maximal hypersurface, 
\begin{equation}\label{29VII24.3rwe}
2\Delta_{\tg} u = -R_{\tg} + e^{2u} \normsq\tL\tg +  e^{-2u} 
    \big(- 2 + \matter)=: 2 F(\cdot, u)
\end{equation}
with $0\le \rho < 2-\epsilon$ for some $\epsilon >0$, 
existence holds for  manifolds $(M,\tilde g)$ with  an arbitrary number of cusps and ALH ends
if we assume  that $  \normsq\tL\tg$
is uniformly bounded and that the Ricci scalar 
$R_{\tg}$ is sandwiched between two negative constants.  

   \item Uniqueness of the above solutions on,  
   e.g., geometrically finite manifolds with ALH ends, with $u$ vanishing asymptotically in the ends,
   will hold when $\NL$ is monotonous in $u$, 
    which is the case for the Lichnerowicz equation of the form just described. 
 \end{enumerate}
 
 In order to obtain asymptotic information on $u$ for ALH metrics $\tg$,  let us write \eqref{11VII24.1-x} in the form assumed in \cite[Theorem 7.2.1]{AndChDiss},
\begin{equation}\label{22V24.1df}
	\aL u = \aF(\bcoord,u) + \aS(\bcoord) \,,
\end{equation}
with $\aF(\bcoord,0) = \partial_s \aF(\bcoord,s)|_{s=0} = 0$.
In half-space coordinates where $\tg = x^{-2}(dx^2 + dy^2)$ the 
Lichnerowicz equation for vacuum CMC data (and so, in particular, the Yamabe equation) takes the form \eqref{22V24.1df} with 
\begin{eqnarray}
&
\aL:= \Delta_{\tg} - 2(\qh +1)
\,,
\quad
\Delta_{\tg} =x^2 (\partial^2_\bx + \partial^2_\by)
\,;
 &
 \\
 & 
\aS:=\qh \,;
 &
 \\
 & 
\aF(\bcoord,u):= 1 + \qh e^{2u} - e^{-2u} - \qh- 2(\qh +1)u
\,,
 & 
 \label{10X24.1}
\end{eqnarray}
with 
\begin{equation}\label{28VIII24.1}
\qh= x^4 |\qf|^2
\end{equation} in the holomorphic-function representation of $TT$ tensors.
When the metric $\tg$ is conformally smooth at a conformal boundary $\{x=0\}$ and $\qf$ is holomorphic in a neighborhood of the conformal boundary, the solutions have a full asymptotic expansion in terms of functions $x^i\ln^jx$ for small $x$.
To justify this, we start by noting that  $\aL$ is an elliptic operator which can be defined as in \cite[(4.2.1)-(4.2.4)]{AndChDiss} with indicial roots $\mu_{-} = -1$ , $\mu_{+} = 2$.
Using \cite[Theorem 7.2.1]{AndChDiss}, we have $\mu_{\pm}=\alpha_{\pm}$ so that $(\mu_{-},\mu_{+})=(-1,2)$ is a strong regularity interval for $C^\alpha_{0+\lambda}(\B)$ for $\aL$.
Assuming  $\qh$ of the form \eqref{28VIII24.1} with a function $\qf$ holomorphic near the unit disc
gives 
$$
\aS\equiv \qh\in x^4 C^{k+\lambda}(\overline{\B})
 \subset x^2 C^{k+\lambda}(\overline{\B})
\,.
$$
%
%
Applying \cite[Theorem 7.4.5]{AndChDiss} with $\alpha=2$, a polyhomogeneous expansion follows. Smoothness at the boundary follows now by inspection of the coefficients in the equation.

%% file: PHdvp.tex

\section{The Lichnerowicz equation with poles on the conformal boundary}
\label{s11VII24.1b}

Our aim is to obtain some insight in the boundary behavior of solutions of the Lichnerowicz equation when the extrinsic curvature arises from a pole at the conformal boundary. 
We use  the notations of \cite{AndChDiss} for the  function spaces that appear here.

In the Poincar\'e-disc model we can use, e.g.,  
$$
 \bro := \frac{ 1-|\vec x|^2}2
$$
as the defining function for the conformal boundary at infinity.

\input{regularity2}

%% file: regularity2.tex
When the extrinsic curvature tensor arises from a function which is holomorphic near the circle at conformal infinity, smoothness at that circle of the solutions of the Lichnerowicz equation has been pointed out in Appendix~\ref{s26VII24.2} below. In the case of poles at the boundary, the analysis in~\cite{AndChDiss} provides
some partial information:

\begin{prop}
 \label{p2VIII24.1}
	Let us assume that 
the function $\qh$ of \eqref{28VIII24.1} satisfies
$\qh=\bro^4|\qf|^2 \in \bro C^{k+\lambda}(\ov \B)$ for some $k\ge1$, $\lambda\in(0,1)$. Then, there exists $\sigma \in (0,1)$ such as the solution $u$ of \eqrefl{22V24.1df}-\eqref{10X24.1} verifies
	\begin{equation}
		u \in \bigcap_{i=0}^2 \bro^{1-i}C^{i+\sigma | k}(\ov \B) \,.
	\end{equation}
	Moreover, in the case $k>1$, there exist  functions $u_j \in \bro^{2j}C^{k-1+\lambda}(\ov \B) \cap \bro^{2j-1}C^{k+\lambda}(\ov \B)$ for $1\le j\le N$ such as
	\begin{equation}
		u - \sum_{j=0}^{N} u_j \log^j(\bro) \in C^{k+1+\sigma}(\ov \B)
		\,,
	\end{equation}
	where $N$ is the smallest integer such as $N > \frac{k}{2} + 1 $.
	Finally, if $u_{1 |\partial \B} = 0$,
	\begin{equation}
		u \in C^{k+1+\sigma}(\overline{\B})\,.
	\end{equation}
\end{prop}
\proof
The result follows by the same arguments as in Appendix~\ref{s26VII24.2}, but this time we only have $\aS\equiv\qh \in \bro C^{k+\lambda}$, in which case \cite[Theorem 7.4.5]{AndChDiss} with $\alpha=1$ provides the regularity above.
\eproofskip
%

%% file: PHdvp-pole.tex
Let us now consider the half-plane model $\HP$ near its boundary $\partial \HP
 = \{\hx=0\}$. We consider a meromorphic function $\qf$  satisfying the finite-mass condition, so that $|\qf|^2\hx^4=\mathcal{O}(\hx)$. The  Lichnerowicz equation, written with respect to the hyperbolic background, takes the form
\begin{equation}
	Lu := \left[\hx^2 \big( \partial_{\hx\hx}^2 + \partial_{\hy\hy}^2 \big) \right]u = 2u + \hx^4|\qf|^2 + o(u)
 = 0 \,,
  \label{20XI24.1}
\end{equation}
so that the term containing $\qf$ can now be seen as a source.
Note that the operator $L$ has again indicial roots $\alpha_\pm = 2\,, -1$.
Now, let us choose
\begin{equation}\label{2VIII24.21}
\qf = \frac 1\hz
\,.
\end{equation}
Thus $\qf$ describes a pole sitting at the origin of the half-plane model. 
Neglecting the $o(u)$ term in \eqref{20XI24.1} one obtains
\begin{equation}
	\left[\big( \partial_{\hx\hx}^2 + \partial_{\hy\hy}^2 \big) - \frac{2}{\hx^2} \right]u = \frac{\hx^2}{\hx^2+\hy^2} \,.
\end{equation}
Using polar coordinates $(\hr,\hang)$ centered at the pole this becomes
\begin{equation}
	\left[\frac {1}{\hr} \partial_{\hr} (\hr\partial_{\hr} \cdot) + \frac{1}{\hr^2}\partial^2_{\hang\hang} - \frac{2}{\hr^2 \cos^2\hang}\right]u = \cos^2\hang\,.
\end{equation}
Searching for solutions as $u(\hr,\hang)=h(\hr)\cos^2\hang$,
 one can find an approximate solution
\ptcheck{22V}
$$
 \frac14  \hr^2 \ln \hr\cos^2 \hang 
 =
  \frac18 \hx^2\ln(\hx^2 + \hy^2) \,,
$$
so that, close to the boundary, we expect
\begin{equation}
	u = u_1 + o(\hx^2)
	\,,
	\label{22V24.2}
\end{equation}
where
\begin{equation}
	u_1 =  \frac18 \hx^2\ln(\hx^2 + \hy^2)
+ 
 \hx^2 (\lambda \hy + \mu)\,,
	\label{22V24.1}
\end{equation}
for some $\lambda, \mu \in \R$.

While Proposition~\ref{p2VIII24.1} guarantees an upper bound on the solution, the following only provides a formal development:
\begin{prop}\label{pole_PHdvp}
	Consider the Lichnerowicz equation with a $TT$-extrinsic curvature tensor defined by the function $\qf$ given by \eqref{2VIII24.21}.
There exists a formal asymptotic solution of the form,
	$\forall N\in\mathbb{N},$ 
	\begin{equation}
		u = \sum_{n=1}^{N} \sum_{k=0}^{n} \hr^{2n}|\ln \hr|^{k} F_{n,k}(\hang)  \, + o(\hr^{2N}) \,,
	\end{equation}
	with
	$F_{n,k}\in C^{\infty}_0([-\frac\pi2,\frac\pi2])
	$.\\
\end{prop}

\begproof
The result works for $N = 1$ by \eqref{22V24.2}-\eqref{22V24.1}.
Let us assume the result for some $N\ge1$. We look for
$$v = O(\hr^{2N+2} \ln ^{N+1}\hr)
$$
such that
\begin{equation}
	\begin{split}
		u & = \sum_{n=1}^{N} \sum_{k=0}^{n} \hr^{2n}|\ln \hr|^{k} F_{n,k}(\hang)  \, + \, v \, + o(\hr^{2(N+1)})  \\
		& = \sum_{n=1}^{N} u_{n}  \, + \, v \, + o(\hr^{2(N+1)}) \,.
	\end{split}
\end{equation}

Now, $u$ verifies 
$$ 
\hx^2 (\partial_\hx^2 + \partial_\hy^2) u = 1 - e^{-2u} + \hx^4 e^{2u}|f|^2 
 \,,
$$
 hence 
\begin{eqnarray}
	\Delta_{\tg}\left(\sum_{n=1}^{N} u_n + v\right)
	&= &
	-\sum_{k=1}^{N} \frac{(-2)^k (\sum_{n=1}^{N} u_n + v)^k}{k!}
	+ \frac{\hx^4}{\hr^2} \sum_{k=0}^{N-1}
	\frac{2^k (\sum_{n=1}^{N} u_n
		+  v)^k}{k!}
	\nonumber
	\\
	&&      + O(\hr^{2(N+1)}\ln ^{N+1} \hr) \,.
\end{eqnarray}
Using the equations of order $n$ for $n \le N$ and simplifying $o(\hr^{2(N+1)})$ terms, only $\hr^{2(N+1)} |\ln \hr|^k, \, k\le N+1 $ terms remain.
We obtain
\begin{equation}\label{Scalar N+1}
	\Delta_{\tg}v = 2 v + \hr^{2(N+1)} \sum_{k=0}^{N+1} |\ln \hr|^k S_{k}(\hang) + o(\hr^{2(N+1)}) \,,
\end{equation}
with $S_k \in  C^{\infty}_0([-\frac\pi2,\frac\pi2])$. 
We look for solutions  of the form
\begin{equation}
	v = \hr^{2(N+1)} \sum_{k=0}^{N+1} F_k(\hang) |\ln \hr|^k + o(\hr^{2(N+1)})\,,
\end{equation}
with $F_k \in C^{\infty}_0([-\frac\pi2,\frac\pi2])$.
\\

We consider  first the most singular terms at the right-hand side of 
 \eqref{Scalar N+1}, namely those which are of order $\hr^{2(N+1)} |\ln \hr|^{N+1}$. Developing the Laplacian using 
\begin{eqnarray}\label{LPH_HPr}
		\Delta_{\tg}\big(\hr^\alpha |\ln \hr|^\beta f(\hang)\big)   
 &=& 
 \hr^{ \alpha}  \cos^2\hang \Big[\big(\alpha^2 |\ln \hr|^\beta - 2\alpha\beta|\ln \hr|^{\beta-1} +
 \nonumber
		\\
		&&\beta(\beta-1)|\ln \hr|^{\beta-2}\big) f(\hang) + |\ln \hr|^\beta 
  \partial^2_{\hang\hang}f 
 \Big]
 \,, 
 \quad
\end{eqnarray}
\eqref{Scalar N+1} becomes
\begin{equation}\label{10V24.1}
	\left(L\left[F_{N+1}\right] -  S_{N+1} + o(1)\right) \, \hr^{2(N+1)}|\ln \hr|^{N+1} = 0 \,,
\end{equation}
for a smooth function of $S_{N+1}(\hang)$, where
\begin{equation}
L = \cos^2\hang \left(4(N+1)^2  + \partial^2_{\hang\hang} \right) - 2 \,.
 \label{19X24.1}
\end{equation} 
We are thus led to study solutions of the problem
\begin{equation}\label{10V24.1bc}
	L\left[F_{N+1}\right] =  S_{N+1}
\end{equation}
in $C^{\infty}_0([-\frac\pi2,\frac\pi2])$. We show in Appendix~\ref{App2VIII24.1}
that such a function $F_{N+1}$ exists.

The remaining terms at the right-hand side of  \eqref{Scalar N+1} can be handled by descending induction. For this,   set $ F_{N+2}\equiv0$; recall that $F_{N+1}$ has just been determined. Assume that for   $k\in [\![0;N]\!] $ the coefficients $F_l$ with
  $l>k$ have been determined. Then, at order $\hr^{2(N+1)} |\ln \hr|^{k}$, Equation~\eqref{Scalar N+1} reads
\begin{equation}
	L[F_{k}] = \tilde{S}_{k} \,,
\end{equation}
with 
$$
\tilde{S}_k = S_k + 2(N+1)(k+1)F_{k+1} - (k+1)(k+2)F_{k+2}\,,
$$
the source term taking into account the terms from derivatives of lower order (higher k).
As before, this equation has a solution in $C^{\infty}_0([-\frac\pi2,\frac\pi2])$.
The iteration over $k$ gives us the development we wanted for $v$, and the iteration over $N$ gives Prop. \ref{pole_PHdvp}.
\ptcheck{23V; quick first look}
\eproof
\\

The proof of Prop. \ref{pole_PHdvp} can then be adapted to any pole on $\partial \HP$. Using \eqref{27V24}, the same formal polyhomogeneous development in the variable $\hr = \frac{1-\hx^2 + \hy^2}{\sqrt{(1+\hx^2)^2 + \hy^2}}$ can be done for a pole on $\B$ (just the angular part changes a priori).

Note that $\hr\rightarrow0$ on the pole but nowhere else on the unit circle, this development is very localized. Anywhere else, the previous development for holomorphic $\qf$ should apply, but this remains to be established.

%% file: maximal.tex
\renewcommand{\blue}[1]{{#1}}

\section{Maximal hypersurfaces}
 \label{app10IX24.1} 
The aim of this section is to verify regularity at the conformal boundary at infinity of maximal surfaces, as needed for the Positivity Theorem~\ref{t8IX24.1}. 
\ptcn{notation from my paper with Greg removed to Greg.tex}

So let  $\scri$ be the conformal boundary at infinity and  $\cut$ be a compact spacelike submanifold of $\scri$ which is a graph 
\begin{equation}\label{10IX24.1}
  t= \bmaxgraph (\varphi)
   \,.
\end{equation}
We can apply an asymptotic symmetry 
 to obtain new  Ba\~nados  coordinates $(\mybar  t,\mybar  r, \mybar  \varphi)$ on a neighborhood of $\cut$ within $\scri$ so that 
$$
 \cut = \{\mybar  t=0\}
  \,.
$$
According to~\cite{ShiMax} (compare~\cite[Theorem~9.3]{ChGallowayMax}), when $\bmaxgraph $ is close to $0$  
(in a norm made precise in these references) there exists a maximal hypersurface spanned on $\cut$ which is the graph of a function $\maxgraph$ over the hypersurface $\{\mybar  t=0\}$. The function $\maxgraph$ is smooth in the interior,   polyhomogeneous at the conformal boundary, and, in general, only $C^2$  up-to-boundary.  More regularity holds in specific situations, which is the case under our hypotheses.
Indeed, one finds that the trace of the extrinsic 
curvature tensor, which we denote by  $\mybar  K $,  
 of the level sets of $\mybar   t$ in coordinates $\mybar  r=1/\red{\mybar  x}$ is given by   
 \begin{align}
  \mathrm{tr}(\mybar  K) &= \red{-}
  \frac{\red{\mybar  x}^3 \left(\mybarmcL_-'(\mybar  t - \mybar  \varphi )
   \left(\red{\mybar  x}^2\mybarmcL_+(\mybar  t + \mybar  \varphi
   )+1\right)+\left(\red{\mybar  x}^2\mybarmcL_-(\mybar  t - \mybar  \varphi
   )+1\right)\mybarmcL_+'(\mybar  t + \mybar  \varphi )\right)}{2
   \left(\left(\red{\mybar  x}^2\mybarmcL_-(\mybar  t - \mybar  \varphi
   )+1\right) \left(\red{\mybar  x}^2\mybarmcL_+(\mybar  t + \mybar  \varphi
   )+1\right)\right)^{3/2}} \\
    &= \red{-}\frac{\mybarmcL_-'(\mybar  t - \mybar  \varphi)+\mybarmcL_+'(\mybar  t + \mybar  \varphi )}{2} \red{\mybar  x}^3
    + O(\red{\mybar  x}^{5})\,.
 \end{align}
It now follows by the arguments in~\cite{ChGallowayMax} that $\maxgraph$ is smooth on the conformally rescaled manifold with  
\begin{equation}
\maxgraph( \red{\mybar  x}, \red{\mybar  \varphi}) =  \maxgraph_4(\red{\mybar  \varphi}) \red{\mybar  x}^4 + \cdots
\end{equation}
where 
\begin{equation}
    \maxgraph_4(\red{\mybar  \varphi} ) =  \frac{1}{8}  
    \left(\mybarmcL_-'(-\red{\mybar  \varphi}
    )+\mybarmcL_+'(\red{\mybar  \varphi} )\right) 
     \,.
\end{equation} 
(We note that $\red{\mybar  x}^2$-terms in $\mybar  K$ would have led to $\red{\mybar  x}^3\log \red{\mybar  x}$-terms in $\maxgraph$, creating problems in an attempt to prove positivity.)
The induced metric on the level sets of $\maxgraph$ reads
\begin{align}
\red{ \mybar {g}} &= (\theta^2 )^2  
+ \frac{\left(\red{\mybar  x}^2\mybarmcL_-(\mybar {t}- \mybar {\varphi}
)+1\right) \left(\red{\mybar  x}^2\mybarmcL_+(\mybar {t}+\mybar {\varphi}
)+1\right)}{\red{\mybar  x}^2} (\theta^1 )^2 
\\
&=(\theta^2 )^2 + (1 +  (\mybarmcL_-(-\red{\mybar  \varphi} )+\mybarmcL_+(\red{\mybar  \varphi} )) \red{\mybar  x}^2 )(\theta^1)^2 + O(\red{\mybar  x}^4)\,,
\end{align} 
where  $\theta^2 = d\red{\mybar  x}/\red{\mybar  x}$ and $\theta^1 = d\mybar \varphi/\red{\mybar  x}$.
We conclude that the mass aspect function of the maximal hypersurface is the same as that of the hypersurface $\{\mybar  t=0\}$, as desired.

%% file: Appendix.tex
\section{An ODE result}
 \label{App2VIII24.1}
  
In this appendix we prove existence of  solutions of the boundary-value problem 
\begin{equation}
	\left[\cos^2\varphi \left(  \alpha^2 + \frac{d^2 }{d\varphi^2}\right) - 2 \right]f(\varphi) = s(\varphi)
 \,,
\label{10V24.2}
\end{equation}
\begin{equation}
	f\left(-\frac{\pi}{2}\right) = f\left(\frac{\pi}{2}\right) = 0
\,,
\end{equation}
on $(-\frac{\pi}{2},\frac{\pi}{2})$, for some $\alpha\in\mathbb{R}$,
 and a source $s$. 
 The operator appearing at the left-hand side of~\eqref{10V24.2} is the operator $L$ of \eqref{19X24.1} with $\hang$ there replaced by $\varphi$ here, 
 and with
$$
 \alpha= 2(N+1) \ge 4
 \,.
$$
We will thus justify existence of solutions of \eqref{10V24.1bc}. 

We use the scalar product 
\begin{equation}
	\langle f,g\rangle := \int_{-\pi/2}^{\pi/2} f(\varphi)g(\varphi) \, d\varphi \,,
\end{equation}
so that the formal adjoint $L^\dagger$ is defined by the equation
\begin{equation}
	\int_{-\pi/2}^{\pi/2} \psi (L\phi) \, d\varphi = \int_{-\pi/2}^{\pi/2} (L^\dagger \psi) \phi \, d\varphi
\,.
\end{equation}
The study of $L^\dagger$ is of  interest since we have:
\begin{equation}
	\ker L^\dagger = \{0\} \quad \quad\Longleftrightarrow\quad \quad \forall\, \xi\  \exists\, \phi\,: \ L\phi = \xi\,.
\end{equation}
It will be convenient to use a new variable $X=\sin\varphi$, so that
\begin{equation}
	\langle f,g\rangle = \int_{-1}^{1} f(X)g(X) \, \frac{dX}{\sqrt{1-X^2}} \,.
\end{equation}
Equation~\eqref{10V24.2} becomes
\begin{equation}
	L=(1-X^2)\Big((1-X^2)\partial^2_X - X\partial_X+ \alpha^2
    \Big) -2 \,.
\end{equation}
Since
\begin{equation}
	d\varphi = \frac{dX}{\sqrt{1-X^2}}
\,,
\end{equation}
we have
\ptcheck{19X24}
$$
	L^\dagger \psi = \sqrt{1-X^2}\Big(
    \partial^2_X ((1-X^2)^{3/2}\psi) + \partial_X (X\sqrt{1-X^2}\psi) 
    \Big)+ \big(\alpha^2(1-X^2) -2) \psi
\,,
$$
which simplifies to
\ptcheck{19X}
\begin{equation}
	L^\dagger = 
 (1-X^2)
 \left[(1-X^2) \partial^2_X - 5X\partial_X + (\alpha^2-4)\right]
 \,.
\end{equation}
Using a power series, one can find the general family of solutions as
\begin{equation}\label{IntSerie}
	\begin{split}
		\psi(X) &=  \sum_{k=0}^{\alpha/2} c_{2k} X^{2k} + X \sum_{k=0}^{\infty} c_{2k+1} X^{2k} \\
		&=: c_0 \, P(X^2) + c_1 \, X Q(X^2)\,,
	\end{split}	
\end{equation}
with $c_0, \, c_1\in\mathbb{R}$ and
$$
\forall k \ge 0, \quad (k+1)(k+2)c_{k+2} = \left[k(k-1) + 5k 
- (\alpha^2 - 4) \right] c_k \,.
$$
Equivalently,
\begin{equation}
	\forall k\ge2, \quad c_k = \frac{k^2 - \alpha^2}{k(k-1)}c_{k-2} \,.
\end{equation}
Since $\alpha = 2(N+1)$ is even, we find
\begin{eqnarray} 
c_{2k} =  (-4)^k
 \frac{N+1-k}{N+1} 
  \binom{N+1+k}{2k} 
 c_0 
  \,,
\end{eqnarray}
for $k\le N$, and $c_{2k}=0$ for $ k >N$, hence the series \eqref{IntSerie}.

We have:

\begin{prop}
$P(1)\neq 0$ and $Q$ has a radius of convergence of 1 and diverges for $X=1$.
\end{prop}
\begin{proof}  
First, using {\sc Mathematica} one checks the identity
\begin{eqnarray}
	P(1)
 & = & \frac{1}{c_0} \sum_{k=0}^{{N+1}} c_{2k} = \sum_{k=0}^{{N+1}} (-4)^k
 \frac{{N+1}-k}{{N+1}} 
  \binom{{N+1}+k}{2k}
\\
 & = &
  -\frac13 (-1)^{N+1} \big( -1 + 3(-1)^{N+1} +4({N+1})^2 \big)
  \,.
\end{eqnarray}
Since $|P(1)| \ge \frac13 (4({N+1})^2 -4) >0$ for $N>0$, we obtain $P(1)\neq0$ as desired.
\\

Next, for all $k\ge0$ let us define
\begin{equation}
	u_k = c_{2k+1 + \alpha}\,,
\end{equation}
so that
\begin{equation}
	\forall k>0, \quad u_k = \frac{(2k+1 + \alpha)^2 - \alpha^2 }{(2k+1+\alpha)(2k+\alpha)}u_{k-1} \,.
\end{equation}
We have
\ptcheck{19X}
\begin{equation}
	\frac{u_k}{u_{k-1}} = 1 + \frac{1}{2k} + \mathcal{O}\left(\frac{1}{k^2}\right)\,,
\end{equation}
hence for all $n\in\N^*$,
\begin{equation}
	\ln(u_n) -\ln(u_0) = \sum_{k=1}^{n} \left[\ln(u_k) - \ln(u_{k-1}) \right] = \sum_{k=1}^{n} \left[\frac{1}{2k} + \mathcal{O}\left(\frac{1}{k^2}\right) \right] = \frac 12 \ln n + C + o(1) \,,
\end{equation}
with $C\in\R$. As a consequence, there exists $\tilde C\in\R$ such that, for large $n$,
\begin{equation}
	u_n \sim \tilde{C} \sqrt n \,,
\end{equation}
so that $\sum u_n$ diverges. 
Therefore, $Q(X)=\sum c_{2p+1} X^{2p}$ has a radius of convergence of 1 and diverges for $X^2=1$.
 \ptcheck{19X}
\hfill$\Box$
\end{proof}

\medskip

With these results, the constraint $\psi(-1)=\psi(1)=0$ leads to $c_0=c_1=0$ as the unique possibility.
As a consequence,
$$
\ker L^\dagger = \{0\}
$$
and $Lf=s$ admits solutions that vanish on the boundary for any source $s$.
%

%% file: spinortotalnew.tex
\section{Imaginary Killing spinors} 
 \label{s24X24.1}
 
We consider the imaginary Killing spinor equation, 
\begin{equation}
D_j \psi = - \frac{i}{2} \gamma_j\psi
 \,,
\end{equation}
which we write as 
\begin{equation}
D_j \psi = \frac{i \, \epsilon}{2} \gamma_j \psi
 \,,
  \label{24X24.3}
\end{equation}
where, for reasons that will become clear shortly, the parameter $\epsilon = \pm 1$ allows one to treat simultaneously 
a representation of the Clifford algebra obtained by multiplying the $\gamma$-matrices  by $-1$.
The Anti-de Sitter metric can be written as
\begin{equation}
    \label{halfspace}
g = \frac{-d \hat{t}^2 + d\hx^2+d\hy^2}{\hx^2}
 \,,
\end{equation}
where $\hx\in \mathbb{R}^+$ and  $\hat{t}, \hy \in \mathbb{R}$.
While  this coordinate system does not cover the whole Anti-de Sitter spacetime (cf., e.g, \cite[Equation~(3.11)]{Banados:1992gq}), 
it is good enough for our purposes as each level set of $\hht$ covers the whole two-dimensional hyperbolic space. 
Using a spin frame associated with the ON frame   
\begin{equation}\label{22X24.45}
  e_{\blue{0}} = \hat x\partial_{\hat t}
  \,,
  \quad
  e_{\blue{1}} = \hat x\partial_{\hat y}
  \,,
  \quad
  e_{\blue{2}} = \hat x\partial_{\hat x}
  \,,
\end{equation}
and  $\gamma$-matrices given by
\begin{equation}
         \gamma^0 = \sigma^3\,,
         \qquad
        \gamma^1 = i \sigma^1\,,
        \qquad 
        \gamma^2 = i \sigma^2\,,
 \label{24X24.1}
\end{equation}
all imaginary   Killing spinors are found to be 
\begin{equation}
\psi \vert_{\epsilon = 1} = \left(\frac{a_2 (-i \hat{t}+i
\hy +\hx +1)-2 i a_1}{2
\sqrt{\hx}}, \frac{2 a_1+ a_2
(\tau - \hy -i \hx +i)}{2 \sqrt{\hx}}\right)\,,
 \label{22X24.42}
\end{equation}
\begin{equation}
    \psi \vert_{\epsilon = - 1} = 
    \left(\frac{b_2 (\hx +1 + i (\hat{t}
    +\hy)) + 2 i b_1}{2
    \sqrt{\hx}}
    , \frac{2 b_1+b_2
    (\hat{t}+ \hy + i \hx - i)}{2
    \sqrt{\hx}}
    \right)
    \,,
     \label{22X24.43}
    \end{equation}
with $a_1$, $a_2$, $b_1$ and $b_2 \in \mathbb C$.  
The associated Killing vectors    $ X := \psi^\dagger  \gamma^\mu \psi \, e_\mu$
 have coordinate components given by  
 \begin{align}\label{23X24.2}
  X^{\hht} 
  &=  2 | a_1 |^2  +  \frac{|a_2|^2}{2} (\hx^2 + 1+ (\hy - \hat{t} )^2)
   - 2 \Re(\bar a_1 a_2) (\hy-\hat{t}) +2 \Im (a_1 \bar a_2)
   \,,
    \\
   X^{\hx}  &=  2 \hx (\Re(\bar a_1 a_2) -
   | a_2 |^2 (\hy - \hat{t}    ))
   \,,
    \\
   X^{\hy}  &=  - 2 | a_1 |^2 
  + \frac{|a_2|^2}{2} (\hx^2 -1 - (\hy - \hat{t} )^2)
  +2 \Re(\bar a_1 a_2) (\hy - \hat{t}) -2 \Im (a_1 \bar a_2)
  \,,
\end{align}  
when $\psi$ is given by \eqref{22X24.42} and   

\begin{align}\label{23X24.1}
  X^{\hht} 
  &= 2| b_1 |^2  +  \frac{|b_2|^2}{2} (1+ \hx^2 + (\hy + \hat{t} )^2)
   + 2 \Re(\bar b_1 b_2) (\hy +\hat{t} )
   -2 \Im (b_1 \bar b_2)
    \, ,
     \\
  X^{\hx} 
  &=  
   2  \hx \big (\Re(\bar b_1 b_2) + 
  | b_2 |^2 (\hy +\hat{t}
   )
    \big )
  \, , 
   \\
  X^{\hy} 
  &=  2 | b_1 |^2 + \frac{|b_2|^2}{2} (1- \hx^2 +(\hy +\hat{t})^2)
   + 2 \Re(\bar b_1 b_2) (\hy +\hat{t} )
  -2 \Im (b_1 \bar b_2)
   \,,
\end{align} 
when $\psi$ is given by \eqref{22X24.43}.

The above imaginary Killing spinors are all obviously well-defined throughout the whole coordinate range, in particular on the whole two-dimensional slice $\hht=0$. 
In particular they extend, as spinor fields (not necessarily imaginary Killing), to any manifold with spin structure inducing  near the conformal boundary at infinity 
the canonical spin structure
as defined in Remark~\ref{R3XI24.1}.

Recall that the hyperbolic cusp is obtained by choosing $\lambda\in \R^+$ and identifying $\hat y$ to a circle of period 
$\lambda$. On this manifold only the imaginary    Killing spinors  with $a_2=0=b_2$ survive the identification, leading to   
Killing vectors colinear with
\begin{equation}\label{24X24.2}
  X =  \partial_\hht  \red{-} \epsilon  \partial_\hy
   \,,
\end{equation}
with $\epsilon $ as in \eqref{24X24.3}. 

These imaginary  Killing spinors extend, as spinor fields (but not Killing in general), both  to manifolds equipped with the canonical spin structure
of Remark~\ref{R3XI24.1}, using the diffeomorphism
\eqref{7XI24.3} between the half-space model and the disc model,
 or the twisted spin structure using the identification 
 $\hy\mapsto \hy+2\pi$, which provides a 
 diffeomorphism between  the periodically identified half-space model and a disc with the origin removed.
 Note that the latter identification leads  to spinor fields which are periodic in the trivial intrinsic spin structure on the bounding circle, 
 while the former to anti-periodic ones. Further comments on this can be found in Remark~\ref{R7XI24.1} below.

It has been pointed-out
in \cite{CoussaertHenneaux} that extreme BTZ black holes provide another example where 
imaginary Killing spinors survive the transition to the quotient, and that this is the only remaining case in the BTZ family. 
This is far from apparent from \eqref{22X24.42}-\eqref{22X24.43} and so, 
to find the imaginary    Killing spinors in the extremal BTZ case, it is useful to consider the   coordinate
system \eqref{difffourg}. For this we start by noting that 
the local coordinate transformation from the 
half-space model to the metric \eqref{1811024} with $M = J \neq 0$ is given by 
 \ptcheck{23X; with the mathematica file}
\begin{align}
    \hht &= \frac{1}{2} \left(\frac{3 M-2 r^2}{\sqrt{2 M} \left(M-2
r^2\right)}-e^{\sqrt{2 M} (\varphi - t )}+ t +\varphi
\right)\,, \\ 
\hx &= \frac{2^{3/4} e^{\frac{\sqrt{M} (\varphi -t
)}{\sqrt{2}}}}{\sqrt{\frac{2 r^2-M}{\sqrt{M}}}}\,, \\
\hy &= \frac{1}{2} \left(\frac{M+2 r^2}{\sqrt{2 M} \left(M-2
r^2\right)}+e^{\sqrt{2 M} (\varphi - t )}+ t +\varphi
\right)\,.
\end{align}
Hence, the identification $(t, r, \varphi) \sim (t, r, \varphi+2 \pi)$
leads to the identification 
\begin{equation}
(\hht + \hy, \hht - \hy- (\sqrt{2 M})^{-1}, \hx) \sim (\hht + \hy + 2 \pi, 
 \big(\hht - \hy - (\sqrt{2 M})^{-1}\big) e^{2 \pi \sqrt{2 M}}, 
\hx e^{\sqrt{2 M} \pi })\,.
 \label{24X24.4}
\end{equation}
Whether or not either of the spinors \eqref{22X24.42}-\eqref{22X24.43} survive \eqref{24X24.4} 
 is not obvious, as the last identification requires an associated adjustment of the spin frame. 
  There are two cases to consider, $M = J$ and $M = -J$.
For $M = J$ the metric reads 
\begin{equation}
    {\fourg }  = -\left(\myr^2 - M\right) dt^2+ \frac{4 \myr^2 d\myr^2}{M^2-4 M \myr^2+4 \myr^4} - M dt d\varphi
    +\myr^2 d\varphi^2\,.
  \end{equation}
This metric has one well-defined imaginary Killing spinor when $\epsilon = 1$.
In a spin frame associated with the tetrad
\begin{equation} 
e_0 = \frac{2 r}{|M-2 r^2|} \partial_t
+ \frac{M}{r |M-2 r^2|} \partial_\varphi
\,, \quad e_1 = \frac{1}{r} \partial_\varphi 
 \,,\quad e_2 = 
 \frac{|M-2 r^2|}{2 r} \partial_r
 \,,
\end{equation}
the imaginary  Killing spinor
has components, for some complex constant $c_1$, 
\begin{equation}
   \psi =c_1  \left(\sqrt{r-\frac{M}{2 r}},- i  \sqrt{r-\frac{M}{2 r}}\right)\,,
    \label{1II25.11}
\end{equation}
where we have assumed that  
$r^2 > M/2$. 
The resulting Killing vector  is 
\begin{equation}
 2 |c_1|^2 \left(\partial_t + \partial_\varphi\right)\,.
\end{equation}
  For $M = - J$, the metric reads
\begin{equation}
    {\fourg }  = -\left(\myr^2 - M\right) dt^2+ \frac{4 \myr^2 d\myr^2}{M^2-4 M \myr^2+4 \myr^4} + M dt d\varphi
    +\myr^2 d\varphi^2\,.
  \end{equation}
It has one well-defined imaginary Killing spinor for $\epsilon = \red{-} 1$. In the spin frame 
defined by the tetrad 
\begin{equation} 
    e_0 = \frac{2 r}{|M-2 r^2|} \partial_t
    - \frac{M}{r |M-2 r^2|} \partial_\varphi
    \,, \quad e_1 = \frac{1}{r} \partial_\varphi 
     \,,\quad e_2 = 
     \frac{|M-2 r^2|}{2 r} \partial_r\,,
    \end{equation} 
it is given by
\begin{equation}
    \psi = c_2 \left(\sqrt{r-\frac{M}{2 r}}, i \sqrt{r-\frac{M}{2 r}}\right)
     \,,
 \end{equation}
 where we have again assumed that we are in the regime 
$r^2 > M/2$. 
The associated Killing vector  reads
\begin{equation}
    2 |c_2|^2 \left(\partial_t \red{-} \partial_\varphi\right)\,.
\end{equation} 

For completeness, and in order to make clear the periodicity of spinor fields near the conformal boundary at infinity,
 we consider imaginary  Killing spinors for Anti-de Sitter in 
the global coordinate system
\begin{equation}
\mathfrak{g} = 
-\frac{
   \left(x^2+y^2+1\right)^2}{\left(x^2+y^2-1\right)^2} dt^2
+4\frac{dx^2+ dy^2}{\left(x^2+y^2-1\right)^2}\,.
\end{equation}
This may be obtained by considering AdS in the coordinate system
\begin{equation}
    \mathfrak{g} = -\left(r^2+1\right) dt^2+\frac{dr^2}{r^2+1}+r^2 d\varphi^2
\end{equation}
and performing the change of coordinates
\begin{equation}
   x =  \frac{\left(\sqrt{r^2+1}+1\right) \cos (\varphi )}{r}\,,
   \qquad 
   y =  \frac{\left(\sqrt{r^2+1}+1\right) \sin (\varphi )}{r}\,.
\end{equation}
We now calculate imaginary  Killing spinors in a spin frame associated with
\begin{align}
  e_0 &= \left(\frac{1-x^2 -y^2}{1+x^2 +y^2}\right) \partial_t
  \,, & e_1 &= \frac{1}{2} 
  \left(1-x^2-y^2\right) \partial_y\,, \\
   e_2 &= \frac{1}{2}
   \left(1-x^2-y^2\right) \partial_x\,,
   & &
\end{align} 
and obtain 
\begin{equation}
  \label{61124spin1}
  \psi \vert_{\epsilon = 1} = 
  \left(-\frac{i e^{-\frac{i t}{2}}
  \left({\tila}_1+ {\tila}_2
  e^{i t} (x+i
  y)\right)}{\sqrt{1-x^2-y^2
  }}, -\frac{e^{-\frac{i t}{2}}
  \left({\tila}_1 (x-i
  y)+ {\tila}_2 e^{i
  t}\right)}{\sqrt{1-x^2-y^2
  }}  \right)
  \,,
  \end{equation}
  \begin{equation}
    \label{61124spin2}
      \psi \vert_{\epsilon = - 1} = 
      \left(
        -\frac{i e^{-\frac{i t}{2}}
   \left({\tilb}_2 (x+i
   y)+{\tilb}_1 e^{i
   t}\right)}{\sqrt{1 -x^2-y^2
   }}, 
   \frac{e^{-\frac{i t}{2}}
   \left({\tilb}_2+{\tilb}_1
   e^{i t} (x-i
   y)\right)}{\sqrt{1 -x^2-y^2
   }}
    \right)
      \,.
      \end{equation}
      with ${\tila}_1$, ${\tila}_2$, ${\tilb}_1$ and ${\tilb}_2 \in \mathbb C$,
       with associated 
      Killing vectors given by
      \begin{align}\label{61124kill1}
        X^{t} 
        &=  |{\tila}_1|^2 + |{\tila}_2|^2  + 
        \frac{4 \Re({\tila}_1 \bar{{\tila}}_2 
        e^{- i t}(x- i y)
        ) }{1+x^2+y^2}
         \,,
          \\
         X^{x}  &=  - |{\tila}_1|^2 y  - |{\tila}_2|^2 y
         +  \Re \big( i \bar{{\tila}}_1 {\tila}_2 e^{i t}
         (-1+ (x+i y)^2)
         \big)
         \,,
          \\
         X^{y}  &= |{\tila}_1|^2 x  + |{\tila}_2|^2 x
         +  \Re \big({\tila}_1 \bar{{\tila}}_2 e^{- i t}
         (1+(x- i y))^2\big)
        \,,
      \end{align}  
      when $\psi$ is given by \eqref{61124spin1} and 
      \begin{align}\label{61124kill2}
        X^{t} 
        &=  |{\tilb}_1|^2 + |{\tilb}_2|^2  +
        \frac{4 \Re({\tilb}_1 \bar{{\tilb}}_2 
        e^{ i t}(x- i y)
        ) }{1+x^2+y^2}
         \,,
          \\
         X^{x}  &= |{\tilb}_1|^2 y  + |{\tilb}_2|^2 y
         +  \Re \big( i {\tilb}_1 \bar{{\tilb}}_2 e^{i t}
         (-1+ (x-i y)^2) \big)
         \,,
          \\
         X^{y}  &= -|{\tilb}_1|^2 x  - |{\tilb}_2|^2 x
         - \Re \big({\tilb}_1 \bar{{\tilb}}_2 e^{ i t}
         (1+(x- i y))^2\big)
        \,,
      \end{align}  
      when $\psi$ is given by \eqref{61124spin2}. 
      
\begin{Remark}
 \label{R7XI24.1}
The imaginary  Killing spinor fields \eqref{61124spin1}-\eqref{61124spin2} are manifestly periodic on any circle $S^1$ in the spin frame underlying these equations.
Under the identification \eqref{trafo6924} of the half-space with  $\R^2$, 
 the orthonormal frame $\{\hat x \partial_{\hat x}, \hat x \partial_{\hat y}\}$ rotates by $2\pi$ with respect to the frame $\{\sqrt{1+r^2}\partial_r, r^{-1} \partial_\varphi\}$ when going around a circle of constant $r$. In the spin frame  associated  to 
   $\{\sqrt{1+r^2}\partial_r, r^{-1} \partial_\varphi\}$  
 the spinor fields change sign after such a trip. This is made explicit by the formulae for the imaginary  Killing spinors  in \cite{CoussaertHenneaux}. 
\qed
\end{Remark} 